\documentclass[11pt]{article} 

\usepackage{amsmath,amsfonts,amssymb,amsthm,graphicx,color,mathrsfs,dcolumn,multirow} %,url,ulem}

\usepackage{array, bbm, centernot}

\usepackage{natbib}

\usepackage{bbm} 

\usepackage[T1]{fontenc} 

\usepackage{algorithm}   
\usepackage[noend]{algpseudocode}

\usepackage{subcaption}

\usepackage{setspace}

\usepackage[hypertexnames=false]{hyperref}

\usepackage{xr}

\makeatletter
\def\BState{\State\hskip-\ALG@thistlm}
\makeatother

 \newtheorem{theorem}{Theorem}[section]
 \newtheorem{lemma}[theorem]{Lemma}
 \newtheorem{corollary}[theorem]{Corollary}

 \newtheorem{remark}[theorem]{Remark}

 \theoremstyle{definition}
 \newtheorem{example}{Example}[section]
 \numberwithin{equation}{section}

  \newtheorem{definition}{Definition}[section]
	
\usepackage{chngcntr}
\newcommand{\N}{\mathbb{N}}

\newcommand{\R}{\mathbb{R}}

\renewcommand{\emptyset}{\text{\usefont{OMS}{cmsy}{m}{n}\symbol{59}}}

\DeclareMathOperator*{\argmin}{arg\,min}
\DeclareMathOperator*{\argmax}{arg\,max}

\renewcommand{\theta}{\vartheta}

\newcommand{\bs}{\boldsymbol}
\newcommand{\out}[1]{{}}

\begin{document}

\allowdisplaybreaks

\baselineskip3.75ex

\title{\vspace{-1cm}\bf High-dimensional variable selection via 
low-dimensional adaptive learning}

\author{Christian Staerk$^1$, Maria Kateri$^2$ and Ioannis Ntzoufras$^3$ 
        \footnote{ To cite this work, please refer to: Staerk, C., Kateri, M., and Ntzoufras, I. (2021). High-dimensional variable selection via low-dimensional adaptive learning. Electronic Journal of Statistics, 15(1), 830-879.  e-mails: {\tt staerk@imbie.uni-bonn.de}, {\tt maria.kateri@rwth-aachen.de},
 {\tt ntzoufras@aueb.gr}
       } 
\vspace{2ex}\\ 
$^1$ Department of Medical Biometry, Informatics and Epidemiology, \\ University Hospital Bonn, Germany\\
$^2$ Institute of Statistics, RWTH Aachen University, Germany\\
$^3$ Department of Statistics, Athens University of Economics and Business, Greece
}

\date{}
\maketitle
\vspace{-0.5cm}
\begin{abstract}
A stochastic search method, the so-called Adaptive Subspace (AdaSub) method, is proposed for variable selection in high-dimensional linear regression models. The method aims at finding the best model with respect to a certain model selection criterion and is based on the idea of adaptively solving low-dimensional sub-problems in order to provide a solution to the original high-dimensional problem. Any of the usual \(\ell_0\)-type model selection criteria can be used, such as Akaike's Information Criterion (AIC), the Bayesian Information Criterion (BIC) or the Extended BIC (EBIC), with the last being particularly suitable for high-dimensional cases. 
The limiting properties of the new algorithm are analysed and it is shown that, under certain conditions, AdaSub converges to the best model according to the considered criterion. 
In a simulation study, the performance of AdaSub is investigated in comparison to alternative methods. The effectiveness of the proposed method is illustrated via various simulated datasets and a high-dimensional real data example. \\[1mm]
\end{abstract}

\section{Introduction}
\label{intro}
Rapid developments during the last decades in fields such as information technology or genetics have led to an increased collection of huge amounts of data. Nowadays one often faces the challenging scenario, where the number of possible explanatory variables \(p\) is large while the sample size \(n\) can be relatively small. In this high-dimensional setting with \(p\) possibly much larger than \(n\) (abbreviated by \(p\gg n\)), statistical modelling and inference is possible under the assumption that the true underlying model is sparse. Hence, we are particularly interested in variable selection, that is we want to identify a sparse, well-fitted model with only a few of the many candidate explanatory variables. %that fits and ideally explains the observed data well. 

Although the proposed Adaptive Subspace method can be applied in a more general setup, in this paper we focus on variable selection in linear regression models with a response \(Y\) and explanatory variables \(X_1,\dots,X_p\), i.e. 
%\vspace{-2mm}
%
\begin{equation} Y_i = \mu + \sum_{j=1}^p \beta_j  X_{i,j} + \epsilon_i , ~~ i=1,\dots,n, \label{linearmodel}\end{equation}
where \(\epsilon_i\) are i.i.d. random errors, \(\epsilon_i\sim N(0,\sigma^2)\), with variance \(\sigma^2>0\), \(\mu\in\R\) is the intercept and \(\bs \beta=(\beta_1,\dots,\beta_p)^T\in\R^{p}\) is the vector of regression coefficients.  The matrix \(\bs X=( X_{i,j})\in\R^{n \times p}\) is the design or data matrix with its \(i\)-th row \(\bs X_{i,*}\) corresponding to the \(i\)-th observation and its \(j\)-th column \(\bs X_{*,j}\) to the values of the \(j\)-th explanatory variable.
Let \(\{X_j:j\in\mathcal{P}\}\) be the set of all possible explanatory variables, where \(\mathcal{P}=\{1,\dots,p\}\) is the corresponding set of indices. Then, for \(S\subseteq\mathcal{P}\), let \(\bs X_{S}\in\R^{n\times |S|}\) denote the design matrix restricted to the columns with indices in \(S\) and let \(\bs \beta_{S}\in\R^{|S|}\) denote the coefficient vector restricted to indices in \(S\). 
Furthermore let \(S_0=\{j\in\mathcal{P}:\,\beta_j\neq 0\}\) be the set of indices corresponding to the true underlying model, the so-called true active set.

As already mentioned, a usual theoretical assumption in the high-dimensional regime is the sparsity of the true model. Thus, for the linear model (\ref{linearmodel}), the cardinality of \(S_0\) is assumed to be small, that is \(s_0=|S_0|\ll p\). The aim is to identify the active set \(S_0\), so a variable selection method tries to ``estimate'' \(S_0\) by some subset \(\hat{S}\subseteq\{1,\dots,p\}\). It is desirable that a selection procedure has the following frequentist properties: The probability \(P(\hat{S}=S_0)\) of selecting the correct model should be as large as possible and the procedure should be variable selection consistent in the sense that \(P(\hat{S}=S_0)\rightarrow 1\) in an asymptotic setting where \(n\rightarrow\infty\) and (possibly) \( p\rightarrow\infty\) with some specified rate. Although the assumption that the ``truth'' is linear and sparse cannot be expected to hold in practice, it is desirable to identify the ``best'' linear, sparse approximation to the ``truth'' in order to find an interpretable model that avoids overfitting (see e.g. \citealp{vandeGeer2011}). %\\
%\indent 

Many different methods have been proposed to solve the variable selection problem in a high-dimensional situation, including the Lasso \citep{Tibshirani1996} and its variants (see \citealp{Tibshirani2011}, for an overview), the SCAD \citep{Fan2001} or Stability Selection \citep{Meinshausen2010}. Here we propose an alternative approach, the Adaptive Subspace (AdaSub) method, which tackles the original high-dimensional selection problem by appropriately splitting it into many low-dimensional sub-problems, based on a certain form of adaptive learning. 

In Section \ref{sec:background} a selective overview of existing high-dimensional variable selection methods is given along with a motivation for the proposed new approach. The AdaSub algorithm is presented in Section \ref{sec:adasub}.  
Its limiting properties are analysed in Section \ref{sec:limiting} where it is shown that, under the ordered importance property (OIP), AdaSub converges to the best model according to the adopted criterion (Theorem \ref{OIP_Theorem}). 
It is further argued that, even when OIP is not satisfied, AdaSub provides a stable thresholded model. 
The performance of AdaSub is investigated through low- and high-dimensional examples in Section \ref{sec:simstudy}, demonstrating that AdaSub can outperform other well-established methods in certain situations with small sample sizes or highly correlated covariates. In Section \ref{sect:realdata}, the effectiveness of AdaSub is further illustrated via a very high-dimensional real data example with \(p=22,575\) explanatory variables. Finally, the results along with directions for future work are discussed in Section \ref{sec:discussion}. %Finally a discussion follows in Section 7.

\section{Background and motivation}\label{sec:background}

Many different methods have been proposed to solve the variable selection problem in a linear model. 
Classical selection criteria include the Akaike Information Criterion AIC \citep{Akaike1974} aiming for optimal predictions and the Bayesian Information Criterion BIC \citep{Schwarz1978} aiming at identifying the ``true'' generating model. 
The BIC can be obtained as an approximation to 
a fully Bayesian analysis with a uniform prior on the model space. 
\citet{Chen2008} argue that this model prior underlying BIC is not suitable for a high-dimensional framework where the truth is assumed to be sparse. 
Therefore they propose a modified version of the BIC, called the Extended Bayesian Information Criterion (EBIC), with an adjusted underlying prior on the model space: For a fixed additional parameter \(\gamma\in[0,1]\) and a subset \(S\subseteq\mathcal{P}\) let the prior of the corresponding model be \(\pi(S)\propto {\binom{p}{|S|} }^{-\gamma}\,.\)
If \(\gamma=1\), the model prior is \(\pi(S)=\frac{1}{p+1}{\binom{p}{|S|}}^{-1}\) and it gives equal probability to each model size, and to each model of the same size. The choice \(\gamma=1\) also corresponds to a default beta-binomial model prior providing automatic multiplicity correction (see \citealp{scott2010}). For \(\gamma=0\), the original BIC is obtained. 

Similarly to the derivation of the BIC, for a subset \(S\subseteq\mathcal{P}\), the EBIC with parameter \(\gamma\in[0,1]\) is asymptotically obtained as
\begin{equation}
\text{EBIC}_\gamma(S) = -2\log\left(f_{\hat{\bs \beta}_S,\hat{\mu},\hat{\sigma}^2}(\bs Y| \bs X_{S})\right) + \Big(\log(n)+2\gamma\log(p)\Big) |S| ,\label{EBIC}
\end{equation}
where \(f_{\hat{\bs \beta}_S,\hat{\mu},\hat{\sigma}^2}(\bs Y| \bs X_{S})\) denotes the maximized normal likelihood under model (\ref{linearmodel}) with restricted design matrix \(\bs X_{S}\) \citep{Chen2012}.
According to EBIC, the active set \(S_0\) is estimated by \( \hat{S} = \argmin_{S} \text{EBIC}_\gamma(S). \) 
It has been shown by \citet{Chen2008} that, under a mild asymptotic identifiability condition, the EBIC is variable selection consistent for a linear model if \(p=\mathcal{O}(n^k)\) for some \(k>0\) and \(\gamma>1-\frac{1}{2k}\), where the size of the true active set \(s_0=|S_0|\) is assumed to be fixed. The result has been extended by \citet{Foygel2010} and \citet{Luo2013} to the setting of a diverging number of relevant explanatory variables. 
 
The identification of the best model according to an \(\ell_0\)-type selection criterion leads to combinatorial optimization problems which are very difficult to solve in the presence of many possible explanatory variables \(p\), since there are \(2^p\) possible models for which the criterion has to be evaluated. 
In fact, best subset selection with an \(\ell_0\)-penalty is in general NP-hard (see e.g. \citealp{Huo2007}). Different alternatives have been proposed to circumvent the costly full enumeration approach. Clever branch-and-bound strategies (see e.g. \citealp{Furnival1974}; \citealp{Narendra1977}) reduce the number of model evaluations and in practice allow an exact solution up to \(p\approx 40\). Very recently, a mixed integer optimization approach has been proposed by \citet{Bertsimas2016} which practically solves problems with \(n\approx 1000\) and \(p\approx 100\) exactly and finds approximate solutions for \(n\approx 100\) and \(p\approx 1000\). Methods like classical forward-stepwise selection, genetic algorithms (see e.g. \citealp{Yang1998}) as well as the the more recently proposed ``shotgun stochastic search'' algorithm of \citet{Hans2007} and the stochastic regrouping algorithm of \citet{cai2009} try to trace good models in a heuristic way, but there is no guarantee that one obtains the optimal solution according to the selected criterion. 

In the 90's the focus shifted from solving discrete optimization problems to solving continuous, convex relaxations of the original problem. \citet{Tibshirani1996} proposes the celebrated Lasso, which solves a convex optimization problem with an \(\ell_1\)-penalty on the regression coefficients and then selects those variables whose corresponding regression coefficients are non-zero in the optimal solution. 
Many modifications of the Lasso have been proposed such as the Elastic Net \citep{Zou2005} or the Group Lasso \citep{Yuan2006} and efficient algorithms for solving the corresponding optimization problems have been developed (see e.g. \citealp{Efron2004}; \citealp{Friedman2007}). 
A drawback of \(\ell_1\)-regularization methods like the Lasso is that, in order to be variable selection consistent, they typically require quite strong conditions on the design matrix \(\bs X\). For the Lasso in linear regression models, it has been shown that the design matrix \(\bs X\) has to satisfy the restrictive ``Irrepresentable Condition'' to obtain variable selection consistency (\citealp{Meinshausen2006}; \citealp{Zhao2006}). Alternative methods like SCAD \citep{Fan2001} --- yielding a non-convex optimization problem --- or the Adaptive Lasso \citep{Zou2006} provide consistent variable selection under weaker conditions. 
The performance of regularization methods such as the Lasso, the Adaptive Lasso and SCAD strongly depends on a sensible choice of their penalty parameters which control the sparsity of the resulting estimators. In practice, penalty parameters are often tuned for predictive performance via cross-validation \citep{shao1993, feng2013}; alternatively, information criteria such as the EBIC can also be used for tuning parameter selection in regularization methods \citep{fan2013}. 

A general problem with procedures based on either \(\ell_0\)- or \(\ell_1\)-type criteria is that their optimal solution is not very stable with respect to small changes in the sample. In particular, it has been noted that the discrete nature of the \(\ell_0\)-penalty can lead to ``overfitting'' of the criterion, if the optimization is carried out among all possible \(2^p\) models (see e.g. \citealp{Breiman1996}; \citealp{Loughrey2005}). Another problem of \(\ell_1\)-type criteria is that they do not provide any information about the uncertainty concerning the best model, per se. Further, it is well-known that standard confidence intervals for regression coefficients are too narrow if the data-driven variable selection is not taken into account. Recent works in post-selection inference aim to yield valid inference after high-dimensional variable selection with methods including the Lasso \citep{zhang2014, van2014, dezeure2015}. 

\citet{Meinshausen2010} propose a procedure called Stability Selection which addresses the particular issue of variable selection (in)stability. 
It is based on the idea of applying a given variable selection method (e.g. the Lasso) multiple times (say \(L\) times) on subsamples of the data. At the end, one selects those explanatory variables whose relative selection frequencies exceed some threshold (which is chosen in a way to control the false discovery rate). The subsampling scheme is to draw subsets \(I_l\), \(l\in\{1,\dots,L\}\), of size \(\left\lfloor \frac{n}{2}\right\rfloor\) without replacement from \(\{1,\dots,n\}\) and then repeatedly consider the model (\ref{linearmodel}) with observations \(i\in I_l\) only. 
Even though Stability Selection has nice theoretical properties and also seems to be used more and more in practice,  
one might observe that in a high-dimensional situation with $p\gg n$, Stability Selection in combination with Lasso successively applies a possibly inconsistent selection procedure on even more severe high-dimensional problems with \(p\ggg\left\lfloor \frac{n}{2}\right\rfloor\). 

The main idea of the proposed AdaSub method is to successively apply a consistent selection procedure (\(\ell_0\)-type criteria like EBIC) on data with the original sample size \(n\) and only a few \(q\) covariates (where \(q\ll\min(n,p)\)). 
So the concept behind AdaSub can be summarized as: 
%\vspace{-1mm}
\begin{center}
\textit{``Solve a high-dimensional problem \\ by solving many low-dimensional sub-problems.''} 
\end{center}
%\vspace{-1mm}
Two issues naturally arise in this regime: Which low-dimensional problems should be solved? And how can the information from the solved low-dimensional problems be combined in order to solve the original problem? 
AdaSub links the answers to those questions using a certain form of adaptive learning: 
In each iteration of the algorithm, the solutions from the already solved low-dimensional problems are used to propose (or more precisely ``sample'' in a stochastic way) a new low-dimensional problem of potentially higher relevance. The construction is based on the principle that a significant explanatory variable 
for the full model space should also be identified as significant in ``many'' of the considered low-dimensional problems it is involved in. 
 
The idea of applying variable selection methods subsequently to different model subspaces appears also in other methods like 
the Random Subspace Method (\citealp{Ho1998}; \citealp{Lai2006}), Tournament Screening \citep{Chen2009}, the stochastic regrouping algorithm \citep{cai2009}, 
the Bayesian split-and-merge (SAM) approach \citep{Song2015}, extensions of Stability Selection \citep{Beinrucker2016} and DECOrrelated feature space partitioning \citep{Wang2016}. 
Relevant are also the PC-simple algorithm \citep{buhlmann2010} and Tilting \citep{cho2012}, which are discussed later in Sections \ref{sec:limiting} and \ref{sec:simstudy}.  
A characteristic feature of the proposed AdaSub method is that it makes explicit and effective use of the information learned from the subspaces already considered by using a certain form of adaptive stochastic learning. In particular, the inclusion probabilities of the individual variables to be selected in the subspaces are adjusted after each iteration of AdaSub, based on their currently estimated ``importance''. Therefore, the sizes of the sampled subspaces in AdaSub are not fixed in advance but are automatically adapted during the algorithm. In addition, the solution of the sub-problems in AdaSub does not necessarily rely on relaxations of the original \(\ell_0\)-type problem (such as the Lasso with an \(\ell_1\)-penalty) or on heuristic methods (such as stepwise selection methods). These features distinguish AdaSub from other subspace methods that have been previously considered in the literature.

\section{The Adaptive Subspace (AdaSub) method}\label{sec:adasub}

\subsection{Notation and assumptions}

We first introduce some general notation in a setting with a criterion-based variable selection procedure. 
For the full set of explanatory variables \(\{X_j:j\in\mathcal{P}\}\) 
we identify a subset \(S\subseteq\mathcal{P}\) with the linear model (\ref{linearmodel}) where the sum on the right hand side is restricted to the indices \(j\in S\); i.e. in matrix notation the model induced by \(S\) is given by 
%\vspace{-1mm}
\begin{equation} \bs Y = \bs \mu + \bs X_{S}\bs\beta_S + \bs \epsilon \,, \end{equation}
where \(\bs Y=(Y_1,\dots,Y_n)^T\), \(\bs \mu=(\mu,\dots,\mu)^T\in\R^n\) and \(\bs \epsilon=(\epsilon_1,\dots,\epsilon_n)^T\) with error variance \(\sigma^2>0\). 
We consider the model space \( \mathcal{M}=\{S\subseteq\mathcal{P}:\, |S|< n-2 \} \,\). 
Here we exclude subsets \(S\subseteq\mathcal{P}\) with \(|S|\geq n-2\) to avoid obvious overfitting and non-identifiability of the regression coefficients. 
Given that we have observed some data \(\mathcal{D}=(\bs X,\bs Y)\), let $C_{\mathcal{D}}:\mathcal{M}\rightarrow \R$ be a certain model selection criterion. In the following we will write \(C\equiv C_{\mathcal{D}}\) for brevity, but one should always recall that the function \(C\) depends on the observed data \(\mathcal{D}\).  
We aim at identifying the best model, which is assumed to be, without loss of generality, the one that maximizes the given criterion \(C\). 

\begin{example}
Examples for \(C\) include posterior model probabilities (within the Bayesian setup) or the negative AIC, BIC or EBIC (within the \(\ell_0\)-penalized criteria framework).
\begin{enumerate}
\item[(a)] To be more specific, in a fully Bayesian framework, posterior model probabilities \(\pi( S \,|\, \mathcal{D})\) are proportional to
\begin{equation}\label{eq:bayesianC}
C(S) = \pi( \bs Y \,|\,\bs X_S, S) \, \pi( S )  \, , ~~ S\in \mathcal{M} \,,
\end{equation}
where \(\pi(S)\) denotes the prior probability of model \(S\) and \(\pi( \bs Y \,|\,\bs X_S, S)\) the marginal likelihood of the data under model \(S\). Maximizing \(C\) in equation (\ref{eq:bayesianC}) corresponds to the identification of the maximum-a-posteriori (MAP) model. 
\item[(b)] In the context of linear regression, (negative) \(\ell_0\)-type information criteria with penalty parameter \(\lambda_{n,p}>0\) can be written as 
\begin{equation} C(S) =  - \left( n \cdot \log\left(\lVert \bs Y - \bs X_S \hat{\bs \beta}_S \rVert_2^2 /n\right) + \lambda_{n,p} |S| \right)\, , ~~ S\in \mathcal{M} \label{l0:normal}\,. \end{equation}
In particular, the penalty parameter choice \(\lambda_{n,p}=2\) corresponds to the AIC, \(\lambda_{n,p}=\log(n)\) corresponds to the BIC and \(\lambda_{n,p}=\log(n)+2\gamma\log(p)\) with \(\gamma\in[0,1]\) corresponds to the \(\text{EBIC}_{\gamma}\) (see equation (\ref{EBIC})). Maximizing \(C\) in equation (\ref{l0:normal}) yields the best model according to the particular \(\ell_0\)-type selection criterion. In this work we mainly focus on \(\ell_0\)-type selection criteria. 
\end{enumerate}
\end{example}

\noindent We define the function
%\vspace{-0.5mm}
\begin{equation} f_C:\mathfrak{P}(\mathcal{P})\rightarrow\mathcal{M},~ f_C(V):=\argmax_{S\subseteq V,\,S\in\mathcal{M}} C(S) \, ,\end{equation}
where \(\mathfrak{P}(\mathcal{P}) = \{ V \subseteq \mathcal{P}\}\) denotes the power set of \(\mathcal{P}=\{1,\dots,p\}\).
So for a given \(V\subseteq\mathcal{P}\), \(f_C(V)\) is the best model according to criterion \(C\) among all models included in \(V\). In the following we will assume 
that $f_C$ is a well-defined function which maps any \(V\subseteq\mathcal{P}\) to a single model \(f_C(V)\in\mathcal{M}\). 
In the \(\ell_0\)-penalized likelihood framework (see equation (\ref{l0:normal})) this assumption almost surely holds if the values of the covariates are generated from an absolutely continuous distribution with respect to the Lebesgue measure \citep{Nikolova2013}; see Remark~\ref{remark:unique} in Section~\ref{sec:theory} for a further discussion of this assumption.
Let 
%\vspace{-0.5mm}
\begin{equation} S^*:=f_C(\mathcal{P})=\argmax_{S\in\mathcal{M}} C(S) \end{equation} 
with $s^*=|S^*|$ denote the best model according to criterion \(C\) which is unique under the made assumptions. Hereafter, \(S^*\) will be referred to as the \(C\)-optimal model.

\begin{remark}\label{remark:fC} The following basic properties can immediately be derived from the definitions of the function \(f_C\) and the \(C\)-optimal model \(S^*\):
\vspace{-1mm}
\begin{enumerate}
\item[(a)] It holds $f_C(V)\subseteq V$ for all $V\subseteq\mathcal{P}$.
\item[(b)] It holds $f_C(V)=S^*$ if and only if $S^*\subseteq V$. 
\item[(c)] If \(f_C(V)\subseteq V' \subseteq V\) with \(V,V'\subseteq\mathcal{P}\), then it holds \(f_C(V')=f_C(V)\).
\end{enumerate}
\end{remark}

Property (b) in Remark \ref{remark:fC} already hints at a strategy for the identification of the \(C\)-optimal model \(S^*\): one aims to identify and solve those low-dimensional sub-problems \(f_C(V)\) with $V\supseteq S^*$, i.e.\ \(V\) should include at least all important variables according to the criterion $C$ (i.e.\ the variables in \(S^*\)), so that \(f_C(V)=S^*\). This property will be particularly exploited in the theoretical analysis of the proposed AdaSub algorithm in Section \ref{sec:limiting}. 

Finally, in the following let \(\N\) denote the set of natural numbers %=\{1,2,\dots\}\) 
and \(\N_0=\N\cup\{0\}\) the set of non-negative integers. For a set \(\Omega\) and a subset \(A\subseteq\Omega\) the indicator function of \(A\) is denoted by \(1_A\), i.e. we have \(1_A(\omega)=1\) if \(\omega\in A\), and \(1_A(\omega)=0\) if \(\omega\in\Omega\setminus A\).

\vspace{-2mm}
\subsection{The algorithm}

We will now describe the generic AdaSub method, given as Algorithm \ref{AdaSub}. A first version of the algorithm has been presented at the 31st International Workshop on Statistical Modelling \citep{Staerk2016}. 

\begin{algorithm}[h]
\caption{Adaptive Subspace (AdaSub) method}\label{AdaSub} 
\vspace{0.2cm}
\textbf{Input:} 
\vspace{-1mm}
\begin{itemize}
\item Data $\mathcal{D}=(\bs X,\bs Y)$ 
%\vspace{1mm}
\item $C:\mathcal{M}\rightarrow \R$ model selection criterion (\(C\equiv C_{\mathcal{D}}\))
%\vspace{1mm}
\item Initial expected search size $q\in(0,p)$ %(small, e.g.\ $q=10$)
%\vspace{1mm}
\item Learning rate $K>0$ %(e.g.\ $K=p$)
%\vspace{1mm}
\item Number of iterations $T\in\N$ 
\end{itemize}

%\vspace{-1mm}
\textbf{Algorithm: }
\vspace{-1mm}
\begin{enumerate}
\item[(1)] For $j=1,\dots,p$ initialize selection probability of variable $X_j$ as
$r_j^{(0)}:=\frac{q}{p}.$
%\vspace{1mm}
\item[(2)] For $t=1,\dots,T$:
%\vspace{1mm}
\begin{enumerate}
\vspace{-1mm}
\item[(a)] Draw $b_j^{(t)}\sim\text{Bernoulli}(r_j^{(t-1)})$ indep. for $j\in\mathcal{P}$.
%\vspace{1mm}
\item[(b)] Set $V^{(t)}=\{j\in\mathcal{P}:\,b_j^{(t)}=1\}$.
%\vspace{1mm}
\item[(c)] Compute $S^{(t)}=f_C(V^{(t)})$.
%\vspace{1mm}
\item[(d)] For $j\in\mathcal{P}$ update \(r_j^{(t)}=\frac{q+K\sum_{i=1}^t 1_{S^{(i)}}(j)}{p+K\sum_{i=1}^t 1_{V^{(i)}}(j)}.\)
\end{enumerate}
\end{enumerate}

\textbf{Output (Final subset selected by AdaSub):} 
%\vspace{-1mm}
\begin{enumerate}
\item[(i)] ``Best'' sampled model: \(\hat{S}_\text{b} = \argmax\{C(S^{(1)}),\dots,C(S^{(T)})\}\) %\\[1.5mm]
%\textbf{or}
\vspace{-1mm}
\item[(ii)] Thresholded model for some threshold $\rho\in(0,1)$: $\hat{S}_{\rho}=\{j\in\mathcal{P}:\,r_j^{(T)}>\rho\}$ 
\end{enumerate}
\end{algorithm}

Suppose that we have observed some data \(\mathcal{D}=(\bs X,\bs Y)\) and we want to identify the \(C\)-optimal model. As described in Section \ref{sec:background}, the basic idea of AdaSub is to solve many low-dimensional problems (i.e. compute \(f_C(V)\) for many subspaces \(V\subseteq\mathcal{P}\) with \(|V|\) relatively small) in order to obtain a solution for the given high-dimensional problem (i.e. identify \(S^*=f_C(\mathcal{P})\)). AdaSub is a stochastic algorithm which in each iteration \(t\), for \(t=1,\dots,T\), samples a subset \(V^{(t)}\subseteq\mathcal{P}\) of the set of all possible explanatory variables \(\mathcal{P}=\{1,\dots,p\}\) and then computes \(S^{(t)}=f_C(V^{(t)})\). The probability that \(j\in\mathcal{P}\) is included in \(V^{(t)}\) at iteration \(t\) is given by \(r_j^{(t-1)}\). The selection probabilities \(r_j^{(t)}\) are automatically adapted after each iteration \(t\) in the following way: 
\begin{equation} r_j^{(t)}=\frac{q+K\sum_{i=1}^t 1_{S^{(i)}}(j)}{p+K\sum_{i=1}^t 1_{V^{(i)}}(j)} \, , \end{equation}
where \(q\in(0,p)\) and \(K>0\) are tuning parameters of the algorithm.

If \(j\in V^{(t)}\) but \(j\notin S^{(t)}=f_C(V^{(t)})\), then \(r_j^{(t)}<r_j^{(t-1)}\), so the selection probability of variable \(X_j\) decreases in the next iteration. If \(j\in V^{(t)}\) and also \(j\in S^{(t)}\), then \(r_j^{(t)}>r_j^{(t-1)}\), so the selection probability increases. If \(j\notin V^{(t)}\), then obviously \(j\notin S^{(t)}\), so the selection probability does not change in the next iteration. Note that \(r_j^{(t)}\) depends on the whole history (from iteration 1 up to iteration \(t\)) of the number of times variable \(X_j\) has been considered in the search (\(j\in V^{(i)}\)) and the number of times it has been included in the best subset (\(j\in S^{(i)}\)). Clearly we have $0<r_j^{(t)}<1$ for all $t=1,\dots,T$ and $j\in\mathcal{P}$. So at each iteration \(t\) each variable $X_j$ has positive probability $r_j^{(t)}$ of being considered in the model search ($j\in V^{(t)}$) and also has positive probability $1-r_j^{(t)}$ of not being considered ($j\notin V^{(t)}$).

As the final subset selected by AdaSub one can either (i) choose the ``best'' sampled model \(\hat{S}_{\textrm{b}}\) for which \(C(\hat{S}_{\textrm{b}})=\max\{C(S^{(1)}),\dots,C(S^{(T)})\}\), or (ii) consider the thresholded model \(\hat{S}_{\rho}=\{j\in\mathcal{P}:\,r^{(T)}_j > \rho\}\) with some threshold \(\rho\in(0,1)\). While \(\hat{S}_{\textrm{b}}\) is obviously more likely to coincide with the \(C\)-optimal model \(S^*\), it can be beneficial in terms of variable selection stability to consider the thresholded model \(\hat{S}_{\rho}\) instead (with \(\rho\) relatively large). A detailed relevant discussion follows in Section \ref{sec:limiting}. 

Note that we implicitly assume that it is computationally feasible to compute \(S^{(t)}=f_C(V^{(t)})\) in each iteration \(t\). In fact, if the underlying ``truth'' is sparse and the criterion used enforces sparsity, \(|V^{(t)}|\) is expected to be relatively small. Otherwise one might use heuristic algorithms in place of a full enumeration.  Alternatively, if \(|V^{(t)}|\) is bigger than some computational bound \(U_C\), one might replace \(V^{(t)}\) by a subsample of \(V^{(t)}\) of size \(U_C\). 
In the case of variable selection in linear regression with $C(S)=-\text{EBIC}(S)$ using the fast branch-and-bound algorithm \citep{Lumley2009} one might set $U_C\leq 40$. 
However, in the following we will assume that the original version of AdaSub (Algorithm \ref{AdaSub}) is used.

The AdaSub method requires that we initialize three parameters: \(q, K\) and \(T\). Here \(q\in(0,p)\) is the initial expected search size, which should be relatively small (e.g.\ $q=10$). The initial expected search size \(q\) reflects our prior belief about the sparsity of the problem, i.e. \(q\) should be a first rough ``estimate'' of the size of the \(C\)-optimal model. We have \(E\left(|V^{(1)}|\right) = \sum_{j=1}^p r_j^{(0)} = q \), 
so the expected search size in the first iteration is indeed $q$. In the following iterations \(t\), \(t\in\{2,\dots,T\}\), the expected search size is automatically adapted depending on the sizes of the previously selected models $S^{(i)}$, \(i<t\); see Section \ref{sec:illustrative} of the appendix for an illustrative example. 
The parameter \(K>0\) controls the learning rate of the algorithm. The larger \(K\) is chosen, the faster the selection probabilities \(r_j^{(t)}\) of the variables \(X_j\) are adapted. 
Based on our experience with numerous simulated and real data examples, we recommend the choices \(K=n\) and \(q\in[5,15]\). A more detailed discussion of the tuning parameters is given in Section \ref{sec:choice}, where we investigate the performance of AdaSub with respect to the choices of \(q\) and \(K\) in a simulation study. 
The number of iterations \(T\in\N\) can be specified in advance. Alternatively one might impose an automatic stopping criterion for the algorithm, but we strongly advise to inspect the output of AdaSub by appropriate diagnostic plots and assess the convergence of the algorithm interactively; see Section \ref{sec:additional} of the appendix for suggested diagnostic plots. 

%\vspace{-2mm}
\section{Limiting properties of AdaSub}\label{sec:limiting}

In this section we summarize theoretical results concerning the limiting properties of AdaSub while a detailed exposition and proofs of the results can be found in the appendix to this paper. In particular, we address the question under which conditions it can be guaranteed that AdaSub ``converges correctly'' against the \(C\)-optimal model \(S^*=f_C(\mathcal{P})\). 

\begin{definition}\label{def:correct_conv}
For a given selection problem with model selection criterion \(C\), the AdaSub algorithm is said to converge to the \(C\)-optimal model \(S^*\) if and only if for all \(j \in \mathcal{P}\) we have for the selection probability of explanatory variable \(X_j\) that 
\vspace{-1ex} 
\begin{equation} r_j^{(t)}\overset{\text{a.s.}}{\rightarrow}\begin{cases} 1 &\mbox{, if } j\in S^*, \\
0 & \mbox{, if } j\notin S^*, \\ \end{cases}  ~~~\text{ for } t\rightarrow\infty \, .   \end{equation}
\end{definition}

By definition, AdaSub converges to the \(C\)-optimal model \(S^*\) if the selection probabilities \(r_j^{(t)}\) converge almost surely against one (zero) for explanatory variables included (not included) in \(S^*\). The \(C\)-optimal convergence of AdaSub implies that, for any fixed threshold \(\rho\in(0,1)\), the thresholded model \(\hat{S}_\rho=\{j\in\mathcal{P}: r_j^{(T)}>\rho\}\) will coincide with the \(C\)-optimal model \(S^*\) if the number of iterations \(T\) of AdaSub is large enough. Note that even when AdaSub does not converge to the \(C\)-optimal model in the sense of Definition \ref{def:correct_conv}, it is still possible that the \(C\)-optimal model is identified by AdaSub, by considering the ``best'' model \(\hat{S}_\text{b}\) found by AdaSub after a finite number of iterations. 

We now introduce the so called ordered importance property (OIP) of a given variable selection problem with criterion \(C\), which turns out to be a sufficient condition for the \(C\)-optimal convergence of AdaSub.

\begin{definition}\label{def:OIP}
Given that dataset \(\mathcal{D}=(\bs X,\bs Y)\) is observed, let \(C_{\mathcal{D}}:\mathcal{M}\rightarrow\R\) be a selection criterion with well-defined function \(f_C\) and \(C\)-optimal model \(S^*=f_C(\mathcal{P})=\{j_1,\dots,j_{s^*}\}\) of size \(s^*=|S^*|\).
Then the selection criterion \(C\) is said to fulfil the \textit{ordered importance property (OIP)} for the sample \(\mathcal{D}\), if there exists a permutation \((k_1,\dots,k_{s^*})\) of \((j_1,\dots,j_{s^*})\) such that for each \(i=1,\dots,s^*-1\) we have
\begin{equation} k_i \in f_C(V) \text{ for all } V\subseteq \mathcal{P} \text{ with } \{k_1,\dots,k_i\}\subseteq V. \label{OIP:eq}\end{equation}
\end{definition}

\begin{theorem}\label{OIP_Theorem}
Given that dataset \(\mathcal{D}=(\bs X,\bs Y)\) is observed, let \(C_{\mathcal{D}}:\mathcal{M}\rightarrow\R\) be a selection criterion with well-defined function \(f_C\) and \(C\)-optimal model \(S^*\). Suppose that the ordered importance property (OIP) is satisfied. Then AdaSub converges to the \(C\)-optimal model~\(S^*\).  
\end{theorem}

We briefly describe the main idea behind OIP and the proof of Theorem \ref{OIP_Theorem}: 
OIP assumes that there exists an \(k_1\in S^*\) (the ``most important'' variable \(X_{k_1}\)) such that it is always selected to be in the best subset \(f_C(V)\) for all sets \(V\subseteq \mathcal{P}\) with \(k_1\in V\). By Theorem \ref{keylemma} of the appendix we conclude that \(r^{(t)}_{k_1}\rightarrow 1\) (almost surely).  
Furthermore, by OIP there exists an \(k_2\in S^*\) (the ``second most important'' variable \(X_{k_2}\)) such that it is always selected to be in the best subset \(f_C(V)\) for all sets \(V\subseteq \mathcal{P}\) with \(k_1,k_2\in V\). In other words, variable \(X_{k_2}\) is always selected to be in the best subset as long as variable \(X_{k_1}\) is also considered. By Theorem \ref{keylemma} we similarly conclude that \(r^{(t)}_{k_2}\rightarrow 1\) (a.s.). We continue in the same way and obtain that \(r^{(t)}_{k_i}\rightarrow 1\) (a.s.) for each \(i=1,\dots,s^*-1\). Now by the definition of the map \(f_C\) and the \(C\)-optimal model \(S^*\) it holds \(f_C(V)=S^*\) for all \(V\subseteq\mathcal{P}\) with \(S^*\subseteq V\) (see Remark \ref{remark:fC}). Thus with Theorem \ref{keylemma} we conclude that \(r_{k_{s^*}}^{(t)}\rightarrow 1\) (a.s.) and that \(r_j^{(t)}\rightarrow 0\) (a.s.) for each \(j\in\mathcal{P}\setminus S^*\). 
In the appendix of this paper we prove the \(C\)-optimal convergence of AdaSub under a slightly different (weaker) sufficient condition OIP' (see Definition \ref{def:OIP'} and Theorem \ref{OIP_Theorem'}). %which is implied by the OIP condition
For ease of presentation here we focused on the more intuitive version of OIP in Definition \ref{def:OIP}. Theorem \ref{OIP_Theorem'} of the appendix implies Theorem \ref{OIP_Theorem} above. 

Note that OIP requires only the existence of such a permutation of the variables with indices in \(S^*\) and not its identification or uniqueness. So in order to guarantee that OIP holds, we do not need to know any concrete permutation, but only that such a permutation exists. On the other hand, this condition cannot be easily checked, since we do not know the set \(S^*\), which AdaSub actually tries to identify.  
Despite this, note that if we observe that the AdaSub algorithm does not converge to the \(C\)-optimal model, i.e.\ if there exists \(j\in\mathcal{P}\) with \(r_j^{(t)}\rightarrow r_j^*\), \(r_j^*\in(0,1)\) with positive probability, then we can conclude that OIP is not satisfied. In that situation we actually might not wish to select \(S^*=f_C(\mathcal{P})\), since then there is no ``stable learning path'' in the sense of OIP. 
Instead, we propose to consider the thresholded model \(\hat{S}_{\rho}\) for some large threshold value (e.g. \(\rho = 0.9\)). 
%\begin{remark}\label{rem:stable}

Indeed, even if OIP does not hold, Corollary \ref{OIP_replacement} of the appendix implies that the thresholded model \(\hat{S}_{\rho}\) will (for fixed \(\rho\in(0,1)\) and \(T\) large enough) contain at least those variables in \(S^*\) that are included in a maximal ``learning path'' in the sense of OIP. 
Simulation studies (Section \ref{sec:simstudy}) show that in most of the cases when OIP is not satisfied the thresholded model \(\hat{S}_{\rho}\) provides a sparser and more stable solution (with less false positives) in comparison to the ``best'' model \(\hat{S}_{\textrm{b}}\) found by AdaSub and in comparison to other competitive variable selection procedures including the Lasso, the Adaptive Lasso and the SCAD (e.g.\ Figure \ref{HIGH_Toeplitz}); see also the examples discussed 
in Sections \ref{sec:illustrative} and \ref{sec:additional} of the appendix. 
In particular, simulation results with the BIC as the selection criterion indicate that in ``unstable'' situations with small sample sizes (e.g.\ \(n\in\{40,60\}\), \(p=30\)) the thresholded model \(\hat{S}_{\rho}\) leads to a large reduction in the mean number of false positives in comparison to the ``best'' model \(\hat{S}_{\textrm{b}}\) found by AdaSub and particularly in comparison to the BIC-optimal model \(S^*\) (see Figures~\ref{LOW_Global00}, \ref{LOW_Toeplitz} and \ref{fig:lowADD}). 
Note that in practice the threshold \(\rho\in(0,1)\) should not be chosen too close to one, since otherwise the selection probabilities \(r_j^{(T)}\) of ``important'' variables may not have exceeded that threshold after a finite number of iterations \(T\in\N\). We observe that the choice \(\rho=0.9\) works empirically well in combination with a sufficiently large number of iterations \(T\) (see Sections \ref{sec:simstudy} and \ref{sect:realdata}). 

The idea behind the ordered importance property (OIP) is connected to the concept of partial faithfulness (PF) underlying the PC-simple algorithm for variable selection of \citet{buhlmann2010}. 
In a random design setting, let \(\rho(Y,X_j~|~X_{S})\) denote the partial correlation between the response \(Y\) and variable \(X_j\) given the set of variables \(X_S:=\{X_k: k\in S\}\) for some subset \(S\subseteq\mathcal{P}\). 
\citet{buhlmann2010} show that if the covariance matrix of \((X_1,\dots,X_p)\) 
is strictly positive definite and if \(\{\beta_j: j\in S_0\} \sim f(b)db\), where \(f\) denotes a density on a subset of \(\R^{|S_0|}\) of an absolutely continuous distribution with respect to the Lebesgue measure, then the PF property holds almost surely with respect to the distribution generating the non-zero regression coefficients, which implies that for each \(j\in\mathcal{P}\) we have
\begin{equation}
\rho(Y,X_j~|~X_{S})\neq 0 \text{ for all } S\subseteq\mathcal{P}\setminus\{j\} ~\Longleftrightarrow~ j \in S_0 = \{ k\in\mathcal{P} :\,\beta_k\neq 0 \} \,. \label{PF_property}
\end{equation}
This means that any truly important variable \(X_j\) (i.e. \(\beta_j\neq 0\)) remains ``important'' when conditioning on any subset \(S\subseteq\mathcal{P}\setminus \{j\}\) (i.e. the corresponding partial correlation is non-zero). Therefore, if PF holds, one would hope that the criterion \(C\), which aims at identifying \(S_0\), does also satisfy the following analogous property (for each \(j\in\mathcal{P}\)):
%\vspace{-2mm}
\begin{equation} j\in f_C(V) \text{ for all } V\subseteq\mathcal{P} \text{ with } j\in V ~\Longleftrightarrow~ j \in S^*=f_C(\mathcal{P}) \,. \label{PF_analog}\end{equation} 
In the following, equation (\ref{PF_analog}) is referred to as the \textit{finite-sample PF property} for the criterion \(C\). 
Note that OIP is significantly weaker than the finite-sample PF property in the sense that in order to have \(j=k_i\in S^*\)
, we do not need to have \(j\in f_C(V)\) for \textbf{all} \(V\subseteq\mathcal{P}\) with \(j\in V\), but only for each \(V\subseteq\mathcal{P}\) with \(k_1,\dots,k_i\in V\). Similarly, an OIP on the population level (which is a weaker condition than the PF property) assumes that, if \(j=k_i\in S_0\)
, then it holds 
\( \rho(Y,X_j~|~X_S) \neq 0 \) for all \(S\subseteq\mathcal{P}\setminus\{j\}\) with \(\{k_1,\dots,k_{i-1}\}\subseteq S\). 
One cannot generally expect that the PF property (\ref{PF_property}) on the population level implies the finite-sample PF property (\ref{PF_analog}) or the weaker OIP in the given finite sample situation. But if OIP does not hold, then this indicates that the best model \(S^*\) according to the criterion \(C\) is not ``stable'' in the sense of (\ref{PF_analog}) and that there does not even exist a ``learning'' path \((k_1,\dots,k_{s^*})\), such that variable \(X_{k_i}\) is selected to be important in each ``relevant experiment'' in which \(X_{k_1},\dots, X_{k_i}\) are considered.

Theorem \ref{OIP_Theorem} guarantees the correct convergence of AdaSub against the \(C\)-optimal model \(S^*\) as the number of iterations \(t\) diverges, provided that OIP holds for the employed criterion \(C\) on the given dataset. However, it does not address the speed of convergence in terms of the required number of iterations to
identify the \(C\)-optimal model \(S^*\) as well as the number of iterations needed so
that the thresholded model of AdaSub \(\hat{S}_\rho\) equals the C-optimal model \(S^*\). 
The general analytical investigation of the speed of convergence under OIP is difficult without further assumptions regarding the particular selection properties for variables in the \(C\)-optimal model \(S^*\). 
In Remark \ref{remark:speed} of Section~\ref{sec:theory} of the appendix we provide analytical results for the speed of convergence under the finite-sample PF property (\ref{PF_analog}), which can be viewed as a \textit{best case} scenario where variables \(X_j\) in the \(C\)-optimal model \(S^*\) are \textit{always} selected to be in the best sub-model \(f_C(V)\) for all possible subspaces \(V\subseteq\mathcal{P}\) with \(j\in\mathcal{P}\). 
Here, we compare via simulations the best case scenario of the finite-sample PF property with a \textit{worst case} scenario under a minimal requirement of OIP. The \textit{minimal OIP} holds if there exists a unique OIP permutation \((k_1,\dots,k_{s^*})\) of variables in \(S^*=\{k_1,\dots,k_{s^*}\}\) such that, for \(i=1,\dots,s^*\), variable \(X_{k_i}\) is \textit{only} selected to be in the best sub-model \(f_C(V)\) if \(\{k_1,\dots,k_i\}\subseteq V\), i.e.\ for \(i=1,\dots,s^*\) it holds
\begin{equation}
k_i\in f_C(V) ~\Longleftrightarrow~ \{k_1,\dots,k_i\}\subseteq V \,.
\end{equation}
This means that, under minimal OIP, variable \(X_{k_i}\) is never selected to be in the best sub-model \(f_C(V)\) if any of the variables \(X_{k_1},\dots,X_{k_{i-1}}\) are not included in the subspace \(V\) (i.e.\ \(k_l\notin V\) for some \(1\leq l<i\)).

\begin{figure}[h!]
\centering
\includegraphics[width=0.95\textwidth]{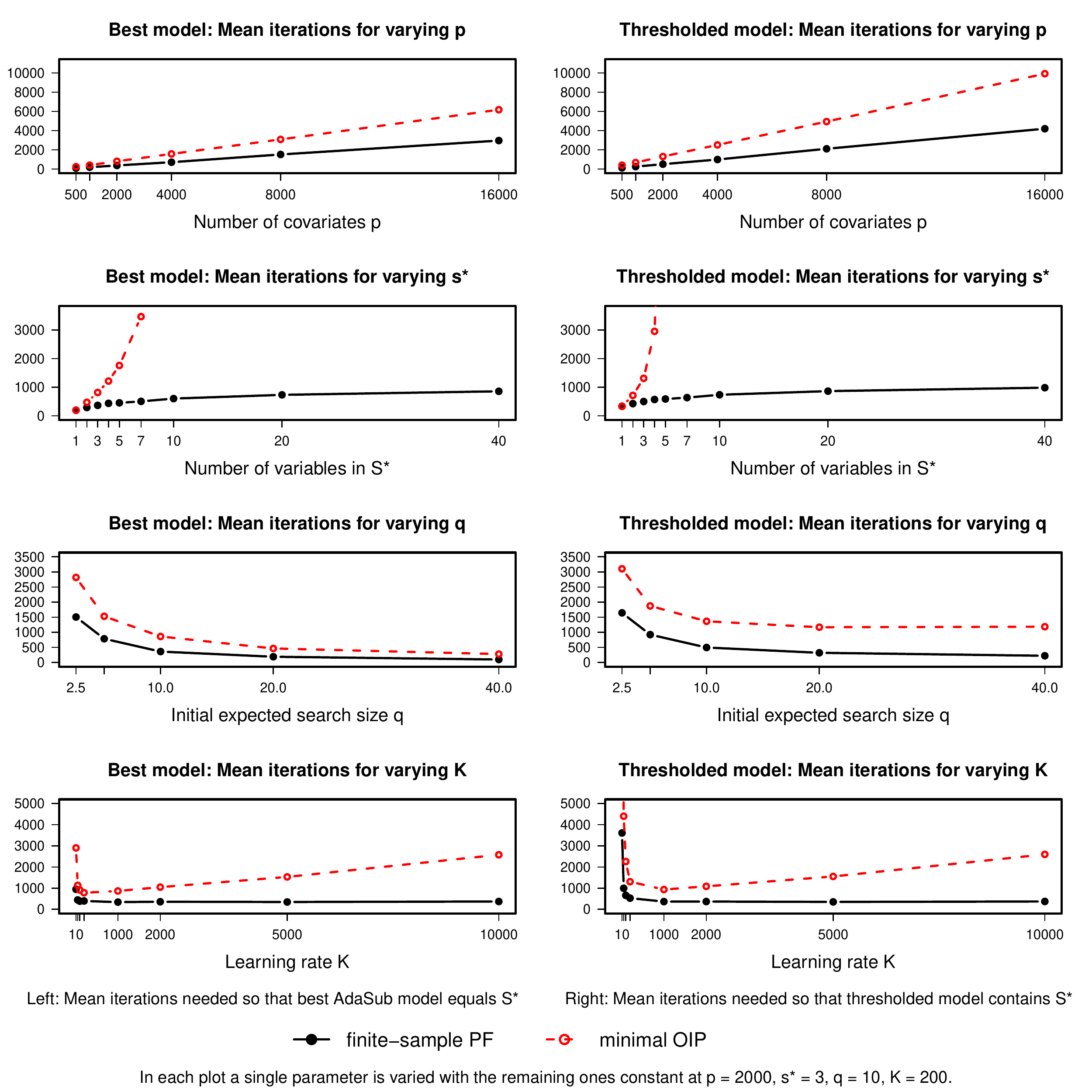} % \makebox{}
\caption{\small Speed of convergence of AdaSub under finite-sample PF and minimal OIP assumptions in terms of mean numbers of iterations needed so that the best AdaSub model \(\hat{S}_{\text{b}}\) equals the \(C\)-optimal model \(S^*\) (left side) and that the thresholded model \(\hat{S}_{0.9}\) of AdaSub contains \(S^*\) (right side). 
Empirical means are based on 500 simulations. } 
\label{fig:speed}
\end{figure}

Figure~\ref{fig:speed} illustrates the speed of convergence of AdaSub under the best case scenario (finite-sample PF) and the worst case scenario (minimal OIP), with respect to the number of covariates \(p\), the number of variables \(s^*\) in \(S^*\), the initial expected search size \(q\) and the learning rate \(K\) in AdaSub. The values of these parameters are set to \(p=2000\), \(s^*=3\), \(q=10\) and \(K=200\); however, in order to investigate the individual effects on the required numbers of iterations of AdaSub, in each plot one of these parameters is varied while the remaining ones are held constant (at the values given above). It can be observed that the mean numbers of iterations required to first identify the \(C\)-optimal model \(S^*\) (left side of Figure~\ref{fig:speed}) 
tend to be smaller than the mean numbers of iterations required so that all variables in \(S^*\) are included in the thresholded model \(\hat{S}_\rho\) with threshold \(\rho=0.9\) (right side of Figure~\ref{fig:speed}), as in the second case the selection probabilities \(r_j^{(t)}\) of variables \(X_j\) with \(j\in S^*\) have to be adjusted multiple times in order to exceed the threshold \(\rho=0.9\).  

Figure~\ref{fig:speed} further shows that the mean numbers of required iterations scale approximately linearly with the number of possible covariates \(p\) under both scenarios of finite-sample PF and minimal OIP (compare Remark~\ref{remark:speed}). The discrepancy between the best case (finite-sample PF) and worst case scenario (minimal OIP) becomes more pronounced for an increasing number of variables \(s^*\) in the \(C\)-optimal model \(S^*\): while mean numbers of required iterations scale approximately logarithmically with increasing \(s^*\) under the finite-sample PF assumption (compare Remark~\ref{remark:speed}), the required iterations quickly increase (non-linearly) with \(s^*\) under the minimal OIP assumption. 
Note that the minimal OIP is a worst case scenario; in practice, when applying AdaSub with a criterion such as the EBIC for a given dataset, the finite-sample PF property often holds for a subset of the variables in \(S^*\) (see also Figure~\ref{fig:PF}), while some of the variables in \(S^*\) may only be selected when certain other variables are also considered at the same time in \(V^{(t)}\). Generally, the AdaSub method is very efficient in sparse scenarios (with small \(s^*\)); on the other hand, the method is not primarily designed for dense settings (with large \(s^*\)), which may occur less likely in high-dimensional situations with limited sample sizes.     

Regarding the tuning parameters of AdaSub, the required iterations are approximately inversely proportional to the initial search size \(q\) under the finite-sample PF assumption (compare Remark~\ref{remark:speed}); a similar decline, though at a larger level, is observed under the minimal OIP assumption.  For an  increasing learning rate \(K\) the required iterations are monotonically decreasing under the finite-sample PF assumption, with the limiting case \(K\rightarrow\infty\) yielding the fastest convergence under finite-sample PF (compare Remark~\ref{remark:speed}, d), as in this case for \(j\in S^*\) it holds \(r_j^{(t)}\approx 1\) for \(t>T_j^{(1)}\), where \(T_j^{(1)}\) denotes the iteration in which variable \(X_j\) is considered (and selected) in $V$ for the first time. However, there is an important trade-off between small and large learning rates \(K\) under the minimal OIP assumption, as in this situation variables in the \(C\)-optimal model \(S^*\) may not always be selected when they are considered in the model search. If for example an important variable $X_j$ with \(j\in S^*\) is not selected when it is first considered in the model search (i.e. $j\in V^{(t)}$ but $j\notin S^{(t)}$), then $r_j^{(t)}=\frac{q}{p+K}$ is close to zero for large \(K\), so variable $X_j$ will probably not be considered in the model search for a long time.

Note that the presented results regarding the speed of convergence of AdaSub, with respect to the parameters \(p\), \(s^*\), \(q\) and \(K\), are based on idealized best case (finite-sample PF) and worst case (minimal OIP) scenarios. In practice, the realized scenario between these two extremes (and thus the speed of convergence) depends also on the properties of the employed selection criterion \(C\) as well as on the characteristics of the particular data situation, including the sample size \(n\), the correlation structure and the signal strength of important covariates. In particular, it can be observed that, with increasing sample size \(n\), variables in the \(C\)-optimal model tend to be selected more frequently for the different sub-problems in AdaSub, so that the finite-sample PF property is more likely to hold for a larger number of variables in the \(C\)-optimal model \(S^*\) (see Figure~\ref{fig:PF} in Section~\ref{sec:additional}), leading to a possibly faster convergence of the algorithm. A more detailed discussion follows in Section~\ref{sec:simstudy} where the performance of AdaSub is investigated in a simulation study for different selection criteria and for various data situations.  

Finally, we would like to emphasize that we have focused on the algorithmic convergence of AdaSub against the best model \(S^*\) according to a given criterion \(C\) (as the number of iterations \(T\) diverges). Based on the presented analysis, depending on the properties of the employed selection criterion \(C\), one may derive specific statistical consistency results for recovering the true underlying model \(S_0=\{j\in\mathcal{P}:\,\beta_j\neq 0\}\) (as the sample size \(n\) and the number of variables \(p\) diverge with a certain rate). 
We briefly indicate how such a consistency result can be obtained in case the employed selection criterion \(C\) is the (negative) BIC. 

For this, note that optimizing a given selection criterion \(C\) inside subspaces \(V\subseteq\mathcal{P}\) with \(S_0\not\subseteq V\) corresponds to variable selection in the situation of misspecified models. 
It has been shown that the BIC is a quasi-consistent criterion in such situations under mild regularity conditions for the classical asymptotic setting where the number of variables \(p\) is fixed and the sample size \(n\) diverges, i.e.\ with probability tending to one, the BIC selects the model that minimizes the Kullback-Leibler divergence to the true model (see e.g.\ \citealp{nishii1988}; \citealp{lv2014}; \citealp{Song2015}). 
By using such a result for each variable selection sub-problem \(f_C(V)=\argmax_{S\subseteq V,\,S\in\mathcal{M}} C(S)\) for all possible subspaces \(V\subseteq\mathcal{P}\), one can deduce that AdaSub in combination with the BIC yields a variable selection consistent procedure for the classical asymptotic setting, provided that the OIP condition on the population level (or alternatively the more stringent PF condition (\ref{PF_property})) is satisfied; this implies that, with probability tending to one, the thresholded model \(\hat{S}_\rho\) of AdaSub equals the true model \(S_0\) when the sample size \(n\) and the number of iterations \(T\) go to infinity for fixed \(p\). The detailed investigation of the variable selection consistency of AdaSub, including high-dimensional asymptotic settings where the number of variables \(p\) diverges with the sample size \(n\), is an interesting topic for future work.

\section{Simulation study}\label{sec:simstudy}

We have investigated the performance of AdaSub in extensive simulation studies and here we present some representative results. The discussion is divided into three parts: First, we examine relatively low-dimensional simulation examples where it is feasible to identify the best model according to an \(\ell_0\)-type criterion \(C\), so that 
it can be compared to the output of AdaSub. In the second part, we apply AdaSub on high-dimensional simulation examples and compare its performance with different well-known methods. Finally, we investigate the algorithmic stability of AdaSub and the effects of the choice of its tuning parameters. 
  
The following simulation setup is used: For a given sample size \(n\in\N\) and a number of explanatory variables \(p\in\N\) we simulate the design matrix \(\bs X=(X_{i,j})\in\mathbb{R}^{n\times p}\) with \(i\)-th row \(\bs X_{i,*}\sim\mathcal{N}_p(\bs 0,\bs \Sigma)\), where \(\bs \Sigma\in\R^{p \times p}\) is a positive definite correlation matrix with \(\Sigma_{k,k}=1\) for \(k=1,\dots,p\). Here, we consider a Toeplitz-correlation structure, i.e. for some \(c\in(-1,1)\) let \(\Sigma_{k,l}=c^{|k-l|}\) for all \(k\neq l\). Results for further correlation structures are presented in Section \ref{sec:additional} of the appendix. 

In particular, we examine the case of independent covariates (\(c=0\)) and the case of highly correlated covariates (\(c=0.9\)).  
For each dataset, we select \(s_0\in\{0,\dots,10\}\) and \(S_0\subset \mathcal{P}\) of size \(|S_0|=s_0\) randomly; then for each \(j\in S_0\) we independently simulate \(\beta_j^0\sim\mathcal{U}(-2,2)\) from the uniform distribution on \([-2,2]\), while we set \(\beta_j^0=0\) for all \(j\notin S_0\). 
The response \(\bs Y=(Y_1,\dots,Y_n)^T\) is then simulated via \(Y_i\overset{\text{ind.}}{\sim} N(\bs X_{i,*}\bs \beta^0, 1)\), \(i=1,\dots,n\), where \(\bs \beta^0 = (\beta^0_1,\dots,\beta^0_p)^T\). 
We apply AdaSub in combination with the (negative) \(\text{EBIC}_\gamma\) as a selection criterion for different regularization constants \(\gamma\in[0,1]\) (recall that \(\gamma=0\) corresponds to the usual BIC). In AdaSub we use the ``leaps and bounds'' algorithm implemented in the \textsf{R}-package \texttt{leaps} \citep{Lumley2009} to compute at iteration \(t\) the best model \(S^{(t)}\) according to \(\text{EBIC}_\gamma\) contained in \(V^{(t)}\).

%\vspace{-2mm}
\subsection{Low-dimensional setting}\label{sec:low}
It is illuminating to analyse the performance of AdaSub in a situation where we actually can compute the best model according to the criterion used (here BIC). We are thus able to answer the question whether AdaSub really recovers the BIC-optimal model. In order to compute the BIC-optimal model in reasonable computational time using the ``leaps and bounds'' algorithm we set \(p=30\). 
For a given correlation structure, the sample size \(n\) is increased from 40 to 200 in steps of size 20 and for each value of \(n\) we simulate 100 different datasets according to the simulation setup described above. In AdaSub we set \(q=5\), \(K=n\) and \(T=2000\).

\begin{figure}[h]
\centering
\includegraphics[width=14cm]{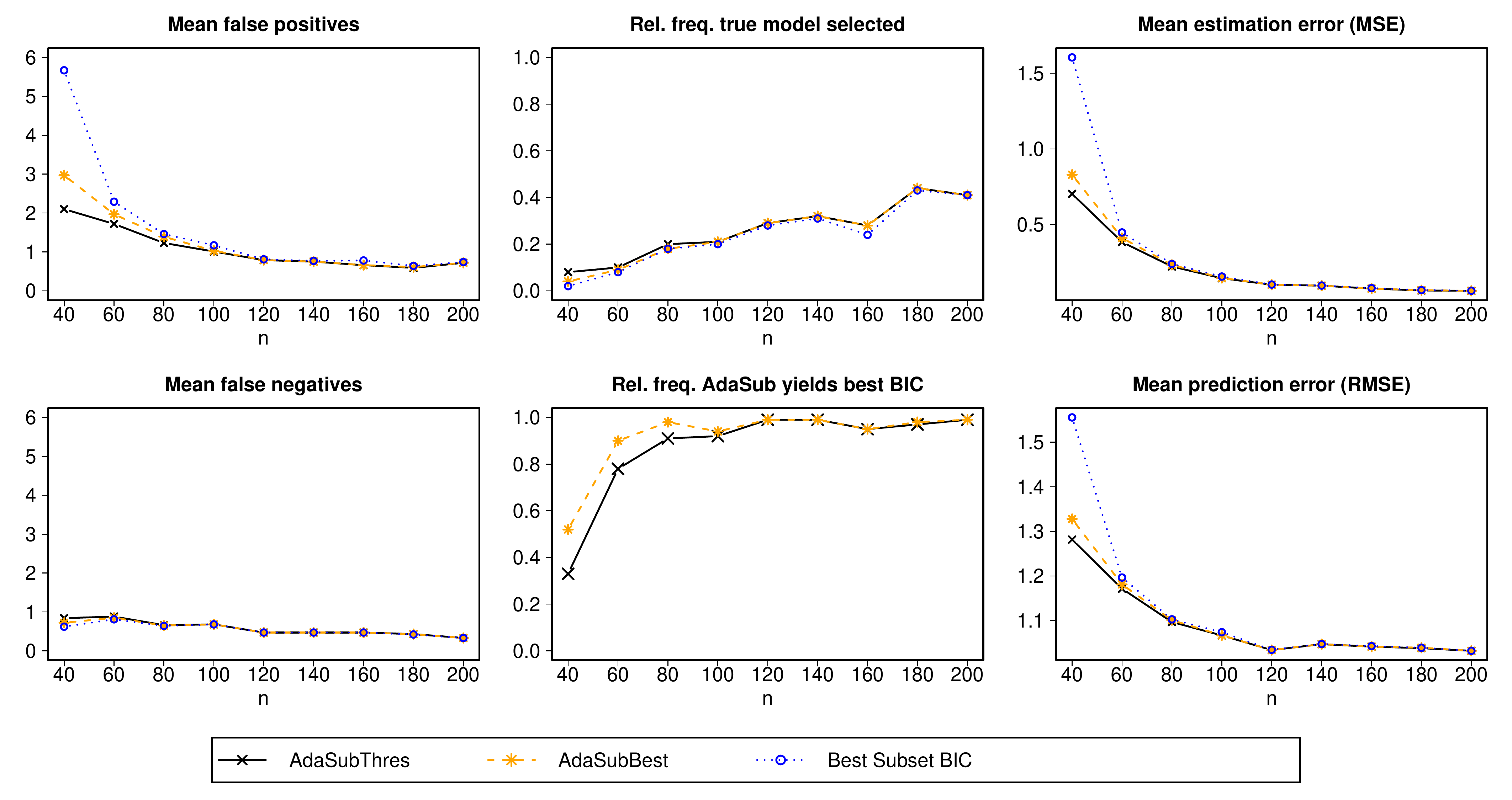} % \makebox{}
\caption{\small Low-dimensional example (\(p=30\)) with independent covariates (\(c=0\)): Comparison of thresholded model \(\hat{S}_{0.9}\) (AdaSubThres) and ``best'' model \(\hat{S}_{\textrm{b}}\) (AdaSubBest) from AdaSub with BIC-optimal model \(S^*\) (Best Subset BIC) in terms of mean number of false positives/ false negatives, relative frequency of selecting the true model \(S_0\), relative frequency of agreement between AdaSub models and \(S^*\), Mean Squared Error (MSE) and Root Mean Squared Prediction Error (RMSE) on independent test set with sample size 100.} 
\label{LOW_Global00}
\end{figure}

Figure~\ref{LOW_Global00} summarizes the results of the low-dimensio\-nal simulation study in the case of independent explanatory variables. 
The BIC-optimal model \(S^*\) tends to select many false positives for small sample sizes and to overfit the data. On the other hand, \(\hat{S}_{0.9}\) and \(\hat{S}_{\textrm{b}}\) from AdaSub yield sparser models and often reduce the number of falsely selected variables in a situation where the BIC is too liberal. This comes at the price of a slightly increased number of false negatives (for small \(n\)), but the overall effect of selecting a sparser model with AdaSub is beneficial for the given situation yielding higher relative frequencies of selecting the true model \(S_0\), smaller Mean Squared Errors (MSE) and smaller Root Mean Squared Prediction Errors (RMSE). Although the ``best'' sampled model \(\hat{S}_{\textrm{b}}\) from AdaSub identifies the BIC-optimal model more often than the thresholded model \(\hat{S}_{0.9}\) from AdaSub, the choice of \(\hat{S}_{0.9}\) is beneficial for the given situation. 
When the sample size increases, the BIC-optimal model becomes more ``stable'' and the relative frequencies that the models selected by AdaSub agree with the BIC-optimal models tend to one. 
We note that the tendency of AdaSub to suggest sparser models in unstable situations is also observed in further simulations with different correlation structures of \(\bs X\) (see Section \ref{sec:additional} of the appendix). 
%\vspace{-2mm}
\subsection{High-dimensional setting}\label{sec:high}
 
We now turn to a high-dimensional scenario, in which both the sample size \(n\) and the number of explanatory variables \(p\) tend to infinity with a certain rate. In particular, we consider the setting \(p=10\times n\) where \(n\) increases from 40 to 200 in steps of size 20 (and thus \(p\) increases from 400 to 2000). For each pair \((n,p)\) we simulate 500 datasets according to the simulation setup described above. 
We compare the ``best'' model \(\hat{S}_{\text{b}}\) from AdaSub and the thresholded model \(\hat{S}_{\rho}\) with \(\rho=0.9\) from AdaSub with different well-known methods for high-dimensional variable selection: We consider the Lasso, Forward Stepwise Regression, the SCAD, the Adaptive Lasso, Stability Selection with Lasso and Tilting. For the computation of the Lasso and the Adaptive Lasso we use the \textsf{R}-package \texttt{glmnet} \citep{friedman2010}, for Stability Selection the \textsf{R}-package \texttt{stabs} \citep{hofner2017}, 
for the SCAD the \textsf{R}-package \texttt{ncvreg} \citep{breheny2011} and for Tilting the \textsf{R}-package \texttt{tilting} \citep{tilting2016}. In AdaSub we choose the \(\text{EBIC}_{\gamma}\) with parameter \(\gamma=0.6\) or \(\gamma=1\) as the criterion \(C\); additionally we set \(q=10\), \(K=n\) and \(T=5000\). Note that \(p=O(n^k)\) with \(k=1\), so that we have \(\gamma>1-\frac{1}{2k}\) and thus \(\text{EBIC}_{\gamma}\) is a variable selection consistent criterion for the given asymptotic setting for both choices of \(\gamma\in\{0.6,1\}\). The choice of \(\gamma\) in \(\text{EBIC}_{\gamma}\) adds flexibility regarding the focus of the analysis (as illustrated in Figures~\ref{HIGH_Global00} and \ref{HIGH_Toeplitz}): if the main aim is variable selection with a small number of selected false positives then the choice \(\gamma=1\) is to be preferred inducing more sparsity, while the choice  \(\gamma=0.6\) yields more liberal variable selection which can be beneficial for predictive performance.
 
For comparison reasons we also choose the regularization parameter of the Lasso, the SCAD and Forward Stepwise Regression according to \(\text{EBIC}_{\gamma}\) (with \(\gamma=0.6\) or \(\gamma=1\)). 
Instead of the usual Lasso and SCAD estimators we use versions of the Lasso-OLS-hybrid (see also \citealp{Efron2004}; \citealp{Belloni2013}) where we compute the \(\text{EBIC}_{\gamma}\)-values of all models along the Lasso-path (and the SCAD-path, respectively) using the ordinary least-squares (OLS) estimators and finally select the model (with corresponding OLS estimator) yielding the lowest \(\text{EBIC}_{\gamma}\)-value.
The additional tuning parameter of the SCAD penalty is set to the default value of 3.7 (as recommended in \citealp{Fan2001}).
 For the Adaptive Lasso we derive the initial estimator with the usual Lasso where the regularization parameter is chosen using 10-fold cross-validation and compute in the second step an additional Lasso path where the regularization parameter is chosen according to \(\text{EBIC}_{\gamma}\). In Section~\ref{sec:additional} of the appendix the performance of the AdaSub models is additionally compared with Lasso, Adaptive Lasso and SCAD estimators where the final regularization parameters are tuned with cross-validation instead of using the \(\text{EBIC}_{\gamma}\) (see Figures~\ref{fig:CV1} and~\ref{fig:CV2}). 
We make use of the complementary pairs version of Stability Selection yielding improved error bounds \citep{shah2013}. The parameters for Stability Selection are chosen such that the expected number of type I errors is bounded by 1 (using the per-family error rate bound), while using the threshold 0.6 and considering 100 subsamples. 
The final estimator for Stability Selection is the OLS estimator for the model identified by Stability Selection.  

Relevant is also the adaptive variable selection approach of \citet{cho2012} via Tilting. Note that this approach is conceptually different from AdaSub in the sense that it builds a sequence of nested subsets \(S^{(1)}\subset S^{(2)} \subset \ldots \subset S^{(m)}\) by gradually adding explanatory variables based on ``tilted'' correlations and then selecting \(\hat{S}=\argmin_{S^{(i)}} \text{EBIC}_\gamma(S^{(i)})\).  
For the Tilting procedure we consider the version TCS2 based on rescaling rule 2 (see \citealp{cho2012}) and we always use the \(\text{EBIC}_{\gamma}\) with \(\gamma=1\) for final model selection, since we observe that the choice \(\gamma=0.6\) yields unreasonably large numbers of false positives. Due to the increasing computational demand of Tilting for larger values of \(p\), the maximum number of selected variables is set to 10 and results are only reported for \(p\leq 1200\) (i.e.\ \(n\leq120\)). 
Our simulations confirm the observation in \citet{cho2012} that Tilting tends to outperform the PC-simple algorithm, thus we do not report the detailed results for the PC-simple algorithm here.

\begin{figure}[!ht]
\centering
\subfloat[Results for \(\text{EBIC}_\gamma\) with \(\gamma = 0.6\).]{\includegraphics[width=14cm]{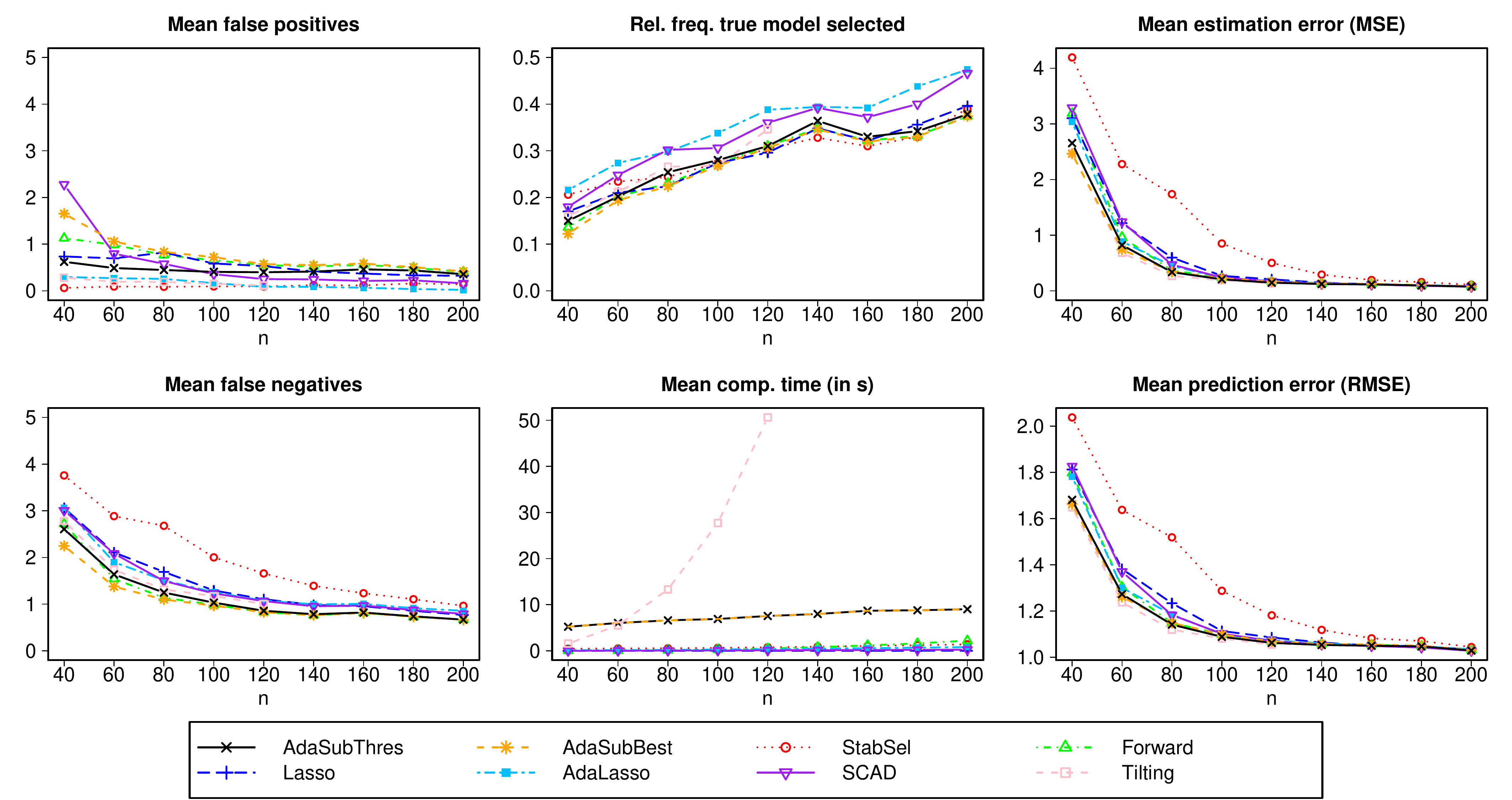}}  \\[2mm]
\subfloat[Results for \(\text{EBIC}_\gamma\) with \(\gamma = 1\).]{\includegraphics[width=14cm]{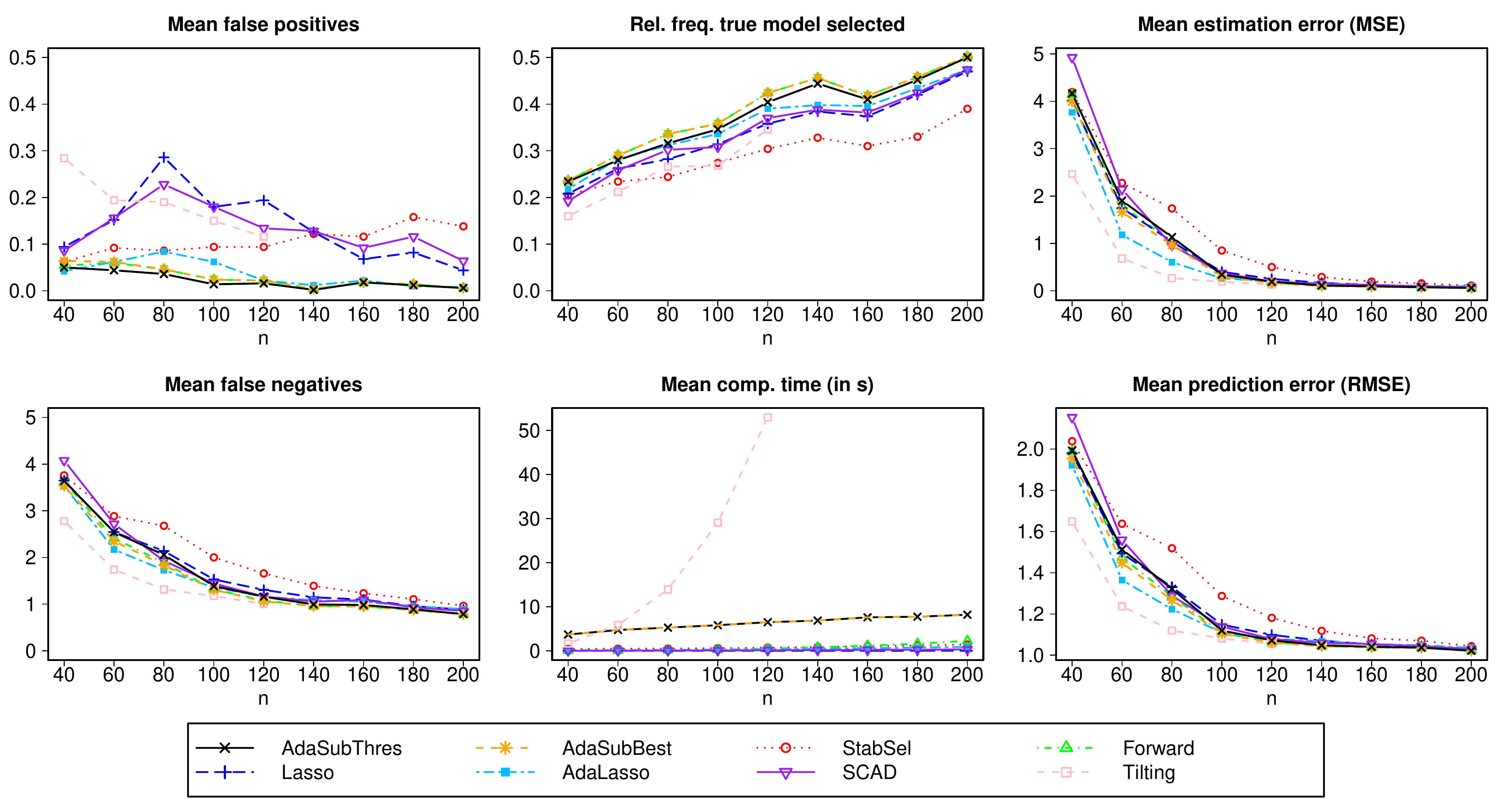}}
\caption{\label{HIGH_Global00}High-dimensional example (\(p=10n\)) with independent covariates (\(c=0\)): Comparison of thresholded model (AdaSubThres) and ``best'' model (AdaSubBest) from AdaSub with Stability Selection (StabSel), Forward Stepwise, Lasso, Adaptive Lasso (AdaLasso), SCAD and Tilting in terms of mean number of false positives/ false negatives, rel. freq. of selecting the true model, mean comp. time, MSE and RMSE.} 
\end{figure}

Figure~\ref{HIGH_Global00} summarizes the results of the high-dimen\-sional simulation study in the case of independent explanatory variables. 
For \(\gamma=0.6\), the ``best'' model \(\hat{S}_{\text{b}}\) from AdaSub tends to include more false positives than the thresholded model \(\hat{S}_{0.9}\), while the number of mean false negatives in \(\hat{S}_{\text{b}}\) is only slightly reduced for small sample sizes. Thus, in this situation with a quite liberal choice of the selection criterion \(\text{EBIC}_{0.6}\), considering the thresholded model is beneficial and yields more ``stable'' variable selection than the ``best'' model according to the criterion identified by AdaSub. On the other hand, for \(\gamma=1\), the \(\text{EBIC}_\gamma\) criterion enforces more sparsity and the performance of the thresholded and ``best'' model from AdaSub is very similar, with slight advantages of the ``best'' model yielding on average less false negatives.      
For \(\gamma=0.6\), the SCAD selects too many false positives if the sample size is small. 
On the other hand, Stability Selection with the Lasso tends to reduce the number of mean false positives in comparison to a single run of the Lasso (for \(\gamma=0.6\)), but at the prize of a larger number of mean false negatives, leading to an undesirable estimative and predictive performance. Furthermore, when the aim is the identification of the true underlying model, Stability Selection is uniformly outperformed by the AdaSub models when considering \(\text{EBIC}_1\) as the selection criterion in AdaSub. As might have been expected in a situation with independent explanatory variables, the performance of Forward Stepwise Selection is quite similar to the ``best'' model identified by AdaSub. In the considered setting it is generally observed that the AdaSub models, Forward Stepwise Selection and the Adaptive Lasso in combination with \(\text{EBIC}_1\) tend to yield the best results with respect to variable selection, while the AdaSub models with \(\text{EBIC}_{0.6}\) and Tilting with \(\text{EBIC}_1\) tend to perform best with respect to estimation and prediction.

\begin{figure}[!ht]
\centering
\subfloat[Results for \(\text{EBIC}_\gamma\) with \(\gamma = 0.6\).]{\includegraphics[width=14cm]{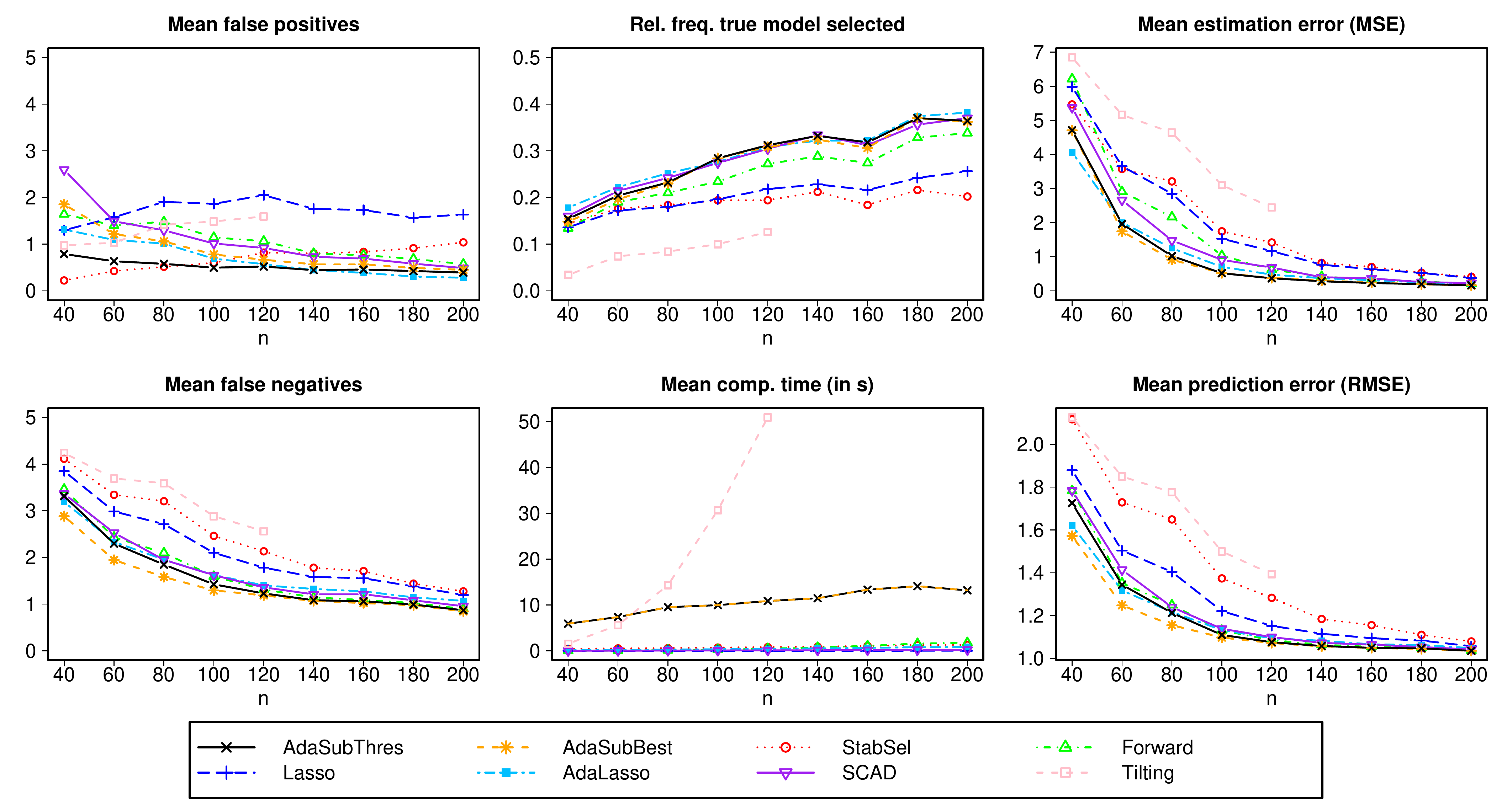}}  \\[2mm]
\subfloat[Results for \(\text{EBIC}_\gamma\) with \(\gamma = 1\).]{\includegraphics[width=14cm]{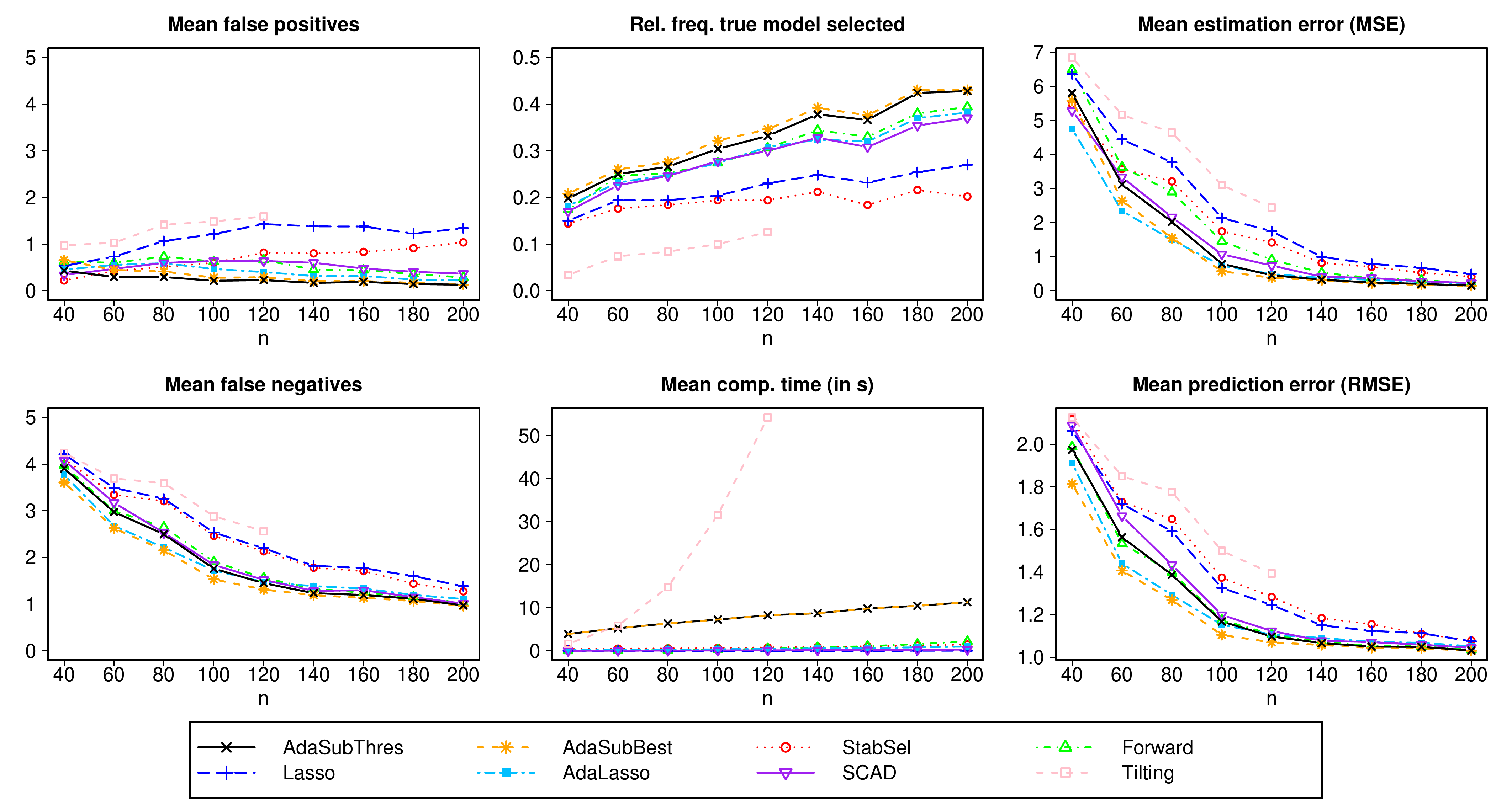}}
\caption{\label{HIGH_Toeplitz} High-dimensional example (\(p=10n\)) with Toeplitz-correlation structure (\(c=0.9\)): Comparison of thresholded model (AdaSubThres) and ``best'' model (AdaSubBest) from AdaSub with Stability Selection (StabSel), Forward Stepwise, Lasso, Adaptive Lasso (AdaLasso), SCAD and Tilting in terms of mean number of false positives/ false negatives, rel. freq. of selecting the true model, mean comp. time, MSE and RMSE.} 
\end{figure}

Figure~\ref{HIGH_Toeplitz} summarizes the results of the high-dimen\-sional simulation study for a Toeplitz-correlation structure with large correlation \(c=0.9\). 
In this setting the thresholded model from AdaSub again tends to select significantly less false positives than the ``best'' model from AdaSub (particularly for \(\gamma=0.6\)), but at the prize of missing some truly important variables (particularly for \(\gamma=1\)). It is generally observed that the AdaSub models for \(\text{EBIC}_1\) tend to yield the best variable selection results, while the ``best'' model selected by AdaSub for \(\text{EBIC}_{0.6}\) tends to show the best predictive performance. Note that using a more liberal selection criterion is beneficial for prediction in the given situation with large correlations among the explanatory variables. The Adaptive Lasso performs generally well, but the AdaSub models with \(\text{EBIC}_1\) show a significantly better variable selection performance. Similarly as in the independence case, although Stability Selection reduces the number of false positives in comparison to the usual Lasso, it is generally outperformed by the AdaSub models. In contrast to the independence scenario, Forward Stepwise Selection does not perform similarly to AdaSub, but tends to include more false positives on average. Tilting seems not to be competitive for the situation of highly correlated covariates.

The summary of the results of additional simulations  
can be found in Section~\ref{sec:additional} of the appendix for this paper. All in all the performance of AdaSub is very competitive to state-of-the-art methods like the SCAD or the Adaptive Lasso and can lead to improved results in situations with small sample sizes or highly correlated covariates. 
Additionally, AdaSub tends to outperform Stability Selection with the Lasso in all of the situations considered. 
We note that the practical computational time needed for a decent convergence behaviour of AdaSub is generally larger in comparison to the considered competitors except for the Tilting method. 
However, the computational times for AdaSub (on an Intel(R) Core(TM) i7-7700K, 4.2 GHz processor) are not prohibitively large 
with on average less than 30 seconds in all considered settings for up to \(p=2000\) variables 
and we are convinced that the extra computational time spent for AdaSub can pay off in many practical situations, as illustrated in this simulation study. 

\subsection{Sensitivity analysis}\label{sec:choice}
 
In order to illustrate the effects of the tuning parameters \(q\) (the initial expected search size) and \(K\) (the learning rate) on the performance of AdaSub, we specifically reconsider the high-dimensional simulation setting of Section \ref{sec:high} with \(n=100\) (\(p=1000\)) and \(n=200\) (\(p=2000\)) for the Toeplitz correlation structure with high correlation \(c=0.9\) and the (negative) \(\text{EBIC}_{0.6}\) as the selection criterion. For both values of \(n\), 100 datasets are simulated as before and for each dataset AdaSub is applied ten times with \(T=5000\) iterations and specific choices of \(q\) and \(K\): For the first five runs of AdaSub \(K=n\) is fixed while \(q\in\{1,2,5,10,15\}\) is varied; for the remaining five runs \(q=10\) is fixed while \(K\in\{1,100,200,1000,2000\}\) is varied. 

In this sensitivity analysis we investigate the efficiency in terms of computational time and the effectiveness with respect to optimizing the given criterion \(\text{EBIC}_{0.6}\) for the ten considered choices of \(q\) and \(K\) in AdaSub. In order to evaluate the optimization effectiveness, we proceed as follows: Let \(\hat{S}_{\text{b}}^{(i,j)}\) denote the ``best'' model identified by the \(j\)-th run of AdaSub for the \(i\)-th dataset, \(i=1,\dots,100\), \(j=1,\dots,10\). Furthermore, let 
\begin{align*}\hat{S}_{\text{b}}^{(i)}=\argmin\big\{\text{EBIC}_{0.6}\big(\hat{S}_{\text{b}}^{(i,1)}\big),\dots,\text{EBIC}_{0.6}\big(\hat{S}_{\text{b}}^{(i,10)}\big)\big\}\end{align*} denote the ``best'' model according to \(\text{EBIC}_{0.6}\) among all ten runs of AdaSub for the \(i\)-th dataset. If  \(\hat{S}_{\text{b}}^{(i,j)}=\hat{S}_{\text{b}}^{(i)}\) then the number of iterations needed to identify the ``best'' model \(\hat{S}_{\text{b}}^{(i)}\) is considered as a measure for the effectiveness of the \(j\)-th run of AdaSub; if \(\hat{S}_{\text{b}}^{(i,j)}\neq\hat{S}_{\text{b}}^{(i)}\) then the \(j\)-th run of AdaSub counts as a ``failure'' and the required number of iterations is set to the maximum number of iterations (\(T=5000\)). 

\begin{figure}[ht]
\centering
\includegraphics[width=0.9\linewidth]{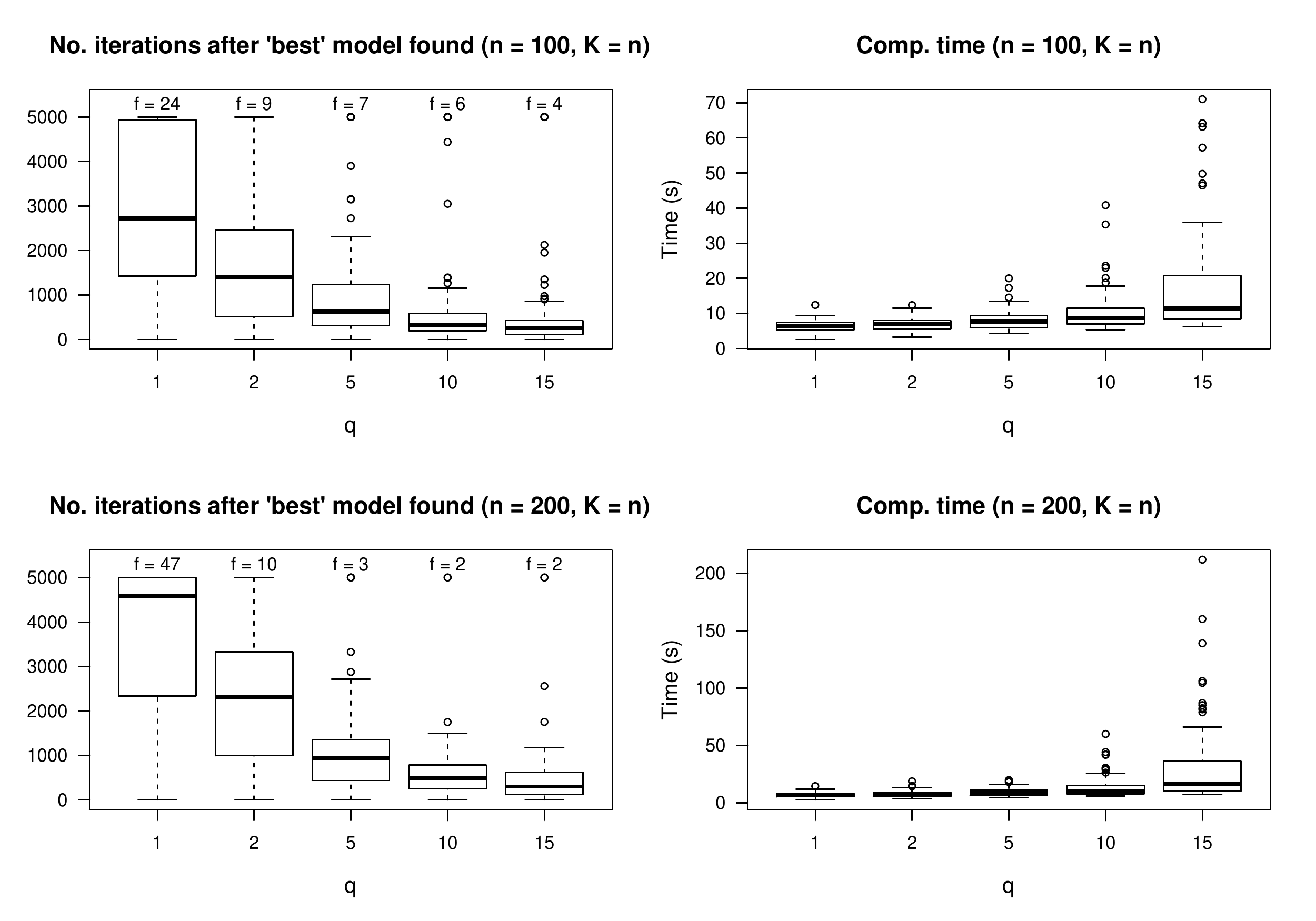} % \makebox{}
\caption{\label{fig:choice_q} Results of AdaSub for different choices of \(q\) (\(K=n\) fixed): Boxplots of the number of iterations needed to identify the ``best'' model (left) and of the computational times (right). In this context, the ``best'' model refers to the model with the smallest EBIC value among all ten runs of AdaSub for that dataset. The number of times the ``best'' model has not been identified is also reported (denoted by f for ``failures''; in such cases 5000 is depicted as the required number of iterations).} 
\end{figure}

Figure~\ref{fig:choice_q} indicates that there is a trade-off between computational efficiency and effectiveness regarding the choice of the initial expected search size \(q\): If \(q\) is small (e.g.\ \(q=1\)), then the algorithm needs more iterations in order to adapt the search sizes accordingly, while a larger value of \(q\) (e.g.\ \(q=15\)) results in larger sampled sub-problems, leading to an increased computational time. However, note that AdaSub automatically adjusts the search sizes so that the choice of \(q\) is not crucial for the limiting behaviour of AdaSub (for a large number of iterations). In practice, we recommend to choose the search size \(q\in[5,15]\).

\begin{figure}[ht]
\centering
\includegraphics[width=0.9\linewidth]{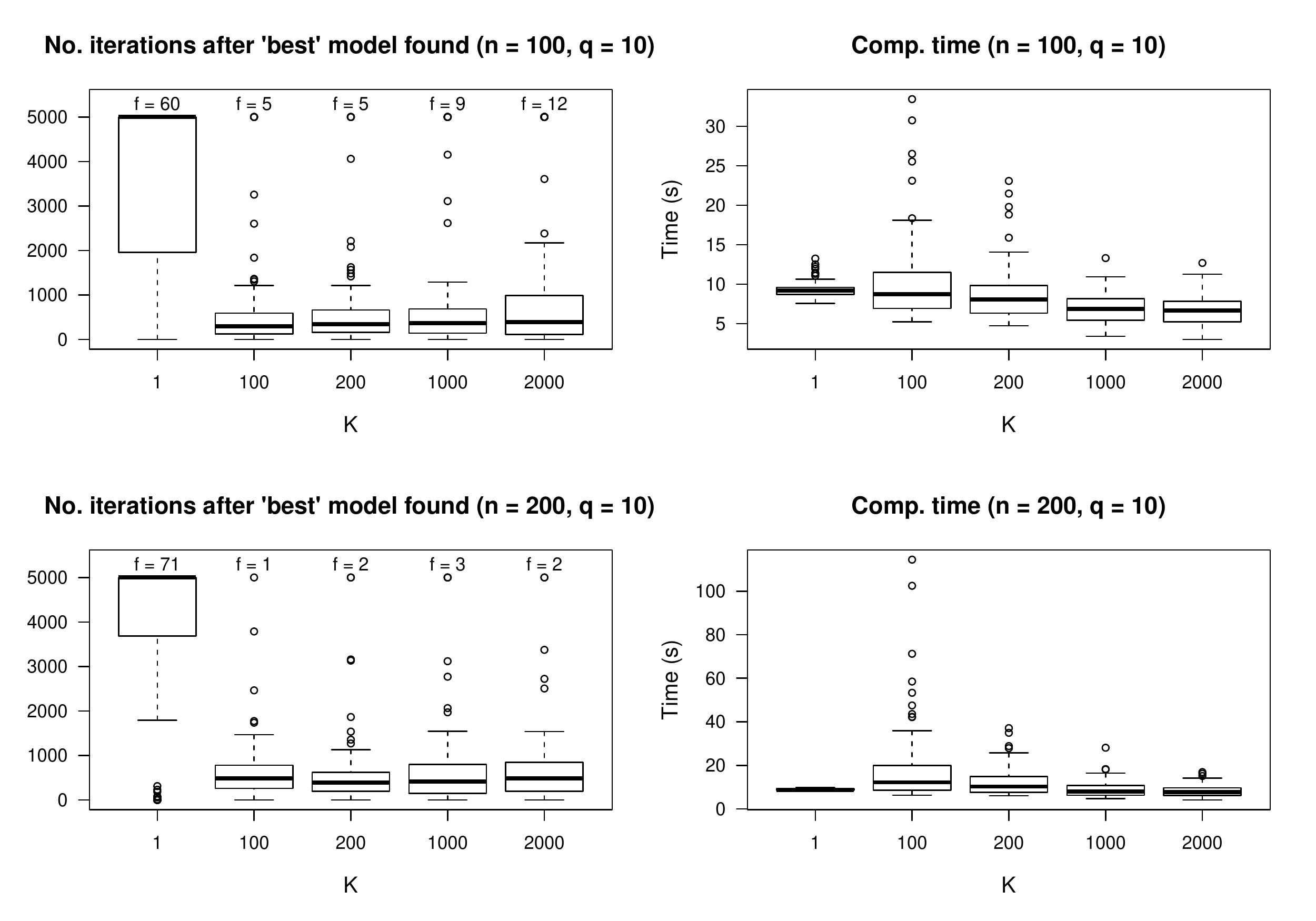} % \makebox{}
\caption{\label{fig:choice_K} Results of AdaSub for different choices of \(K\) (\(q=10\) fixed). The description of the illustrated boxplots is as in Figure~\ref{fig:choice_q}.} 
\end{figure}

Figure~\ref{fig:choice_K} shows that there is another trade-off regarding the choice of the learning rate $K>0$: If $K$ is small (e.g.\ \(K=1\)), then we are learning slowly from the data in order to sample more promising low-dimensional sub-problems, resulting in a slow convergence of the algorithm. If instead $K$ is large (e.g.\ \(K=2000\)), the algorithm might focus too quickly on specific classes of sub-problems and thus often a larger number of iterations is needed to identify the ``best'' model. 
It can be argued that a sensible choice of \(K\) depends on the sample size \(n\) of the considered dataset, since larger sample sizes come with less uncertainties regarding the ``best'' model and a faster convergence of the algorithm might be achieved with larger values of \(K\). We recommend to choose the learning rate \(K=n\); this choice of \(K\) is also supported by the results in Figure~\ref{fig:choice_K} regarding the required number of iterations to identity the ``best'' models. We refer to \citet[Sections 3.4, 3.5]{staerk2018} for additional discussions regarding the choice of \(K\) and \(q\).

Since AdaSub is a stochastic algorithm, it is desirable that the selected models by AdaSub do not largely vary if one repeatedly runs the algorithm for the same dataset and the same selection criterion, but with possibly different choices of the tuning parameters of AdaSub. 
In order to investigate the algorithmic stability of AdaSub we consider the same setting as in the high-dimensional simulation study of Section \ref{sec:high} and rerun the AdaSub algorithm ten times with \(T=5000\) iterations for a particular dataset with random choices of \(K\) and \(q\) from a sensible range. Here, we simulate 20 different datasets for each value of \(n\in\{40,60,\dots,200\}\) (with \(p=10n\)) for both the independence and Toeplitz correlation structure and consider again the (negative) \(\text{EBIC}_\gamma\) with \(\gamma\in\{0.6,1\}\) as the selection criterion, yielding in total \(2\times2\times10\times20\times9 = 7200\) different runs of AdaSub. For each application of AdaSub, the initial expected search size \(q\) is randomly generated from the uniform distribution \(\mathcal{U}(5,15)\) and the learning rate \(K\) is randomly generated from the uniform distribution \(\mathcal{U}(n/2,2n)\).

%\vspace{-2mm}
 
\begin{figure}[ht]
\centering
\includegraphics[width=0.9\linewidth]{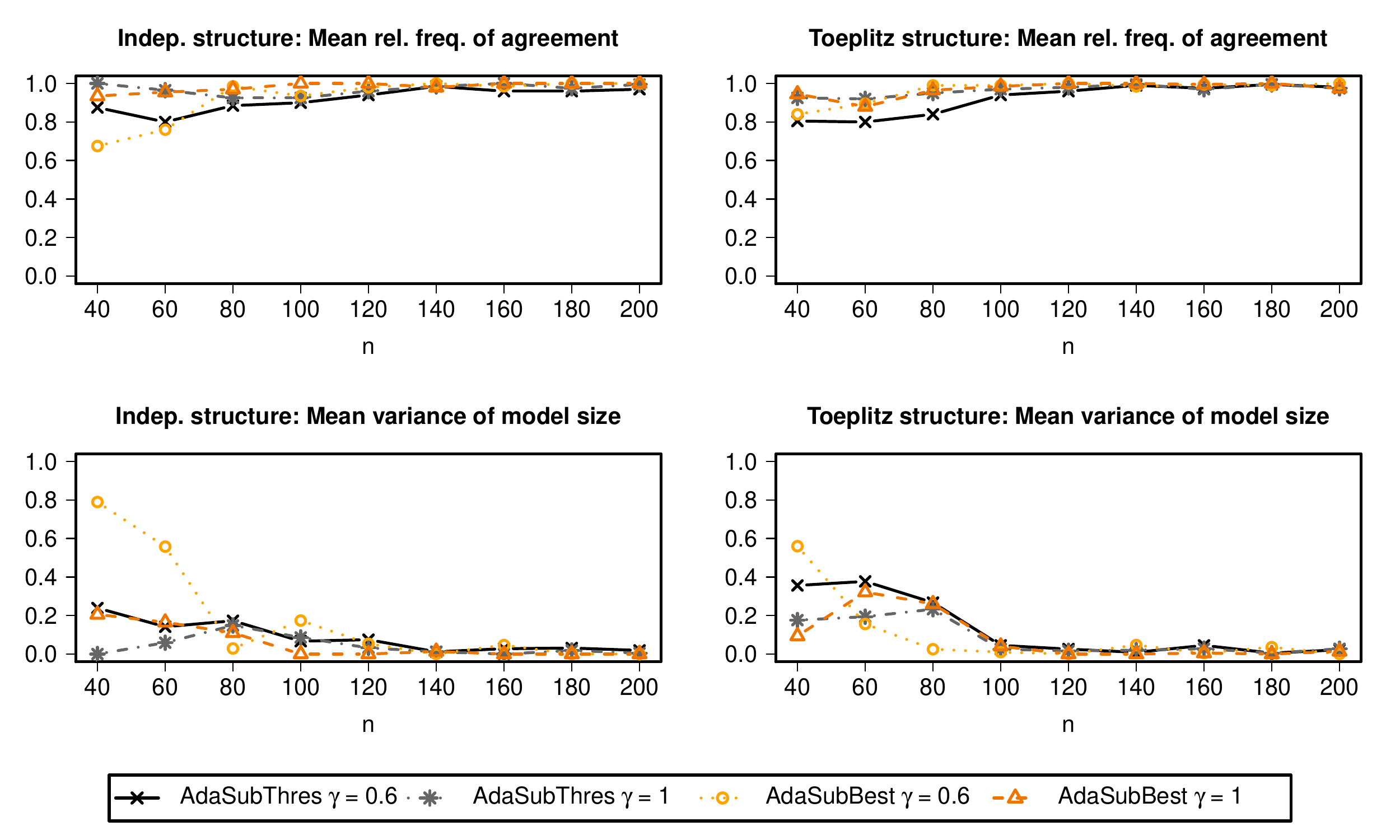} % \makebox{} 
\caption{\label{fig:stability} Sensitivity analysis for the tuning parameters \(q\) and \(K\), assuming independence (\(c=0\)) and Toeplitz (\(c=0.9\)) correlation structures: 
Mean relative frequency of model agreement and mean variance of model sizes across the ten runs of AdaSub (averaged over 20 simulated datasets for each sample size) for the thresholded model \(\hat{S}_{0.9}\) (AdaSubThres) and the ``best'' model \(\hat{S}_{\text{b}}\) (AdaSubBest) for multiple runs of AdaSub with \(\text{EBIC}_\gamma\) for \(\gamma\in\{0.6,1\}\).} 
\end{figure}

In Figure~\ref{fig:stability} it can be seen that the average relative frequencies of model agreement for both the thresholded and the ``best'' model are reasonably large across different runs of AdaSub for the same datasets (with random choices of \(q\) and \(K\)). Furthermore, the variances of the sizes of the AdaSub models are small, indicating that the selected models are quite similar even if they differ between certain runs of AdaSub. Note that the algorithmic stability of AdaSub further improves with increasing samples size \(n\), i.e.\ the relative frequencies of agreement tend to one and the variances of model sizes tend to zero.

\section{Real data example} \label{sect:realdata}
In this section we consider the application of AdaSub on (ultra)-high-dimensional real data. For comparison reasons we examine a polymerase chain reaction (PCR) dataset which has already been analysed in \citet{Song2015}. 
They demonstrate that their Bayesian split-and-merge approach (SAM) performs favourably in comparison to hybrid methods like (I)SIS-lasso and (I)SIS-SCAD, so we do not include the results of these methods here.  
(I)SIS-lasso and (I)SIS-SCAD are acro\-nyms for the combination of a screening step with (Iterated) Sure Independence Screening \citep{Fan2008} and then a selection step of the final model with lasso and SCAD, respectively. A special intention of this section is to show that it is computationally feasible to apply the AdaSub method even in the situation of ultra-high-dimensional data with ten thousands of explanatory variables and that an additional screening step is not necessarily needed.

We consider the preprocessed PCR data from \citet{Song2015}, available in JRSS(B) Datasets Vol. 77(5), which consists of \(n=60\) samples (mice) with \(p=22,575\) explanatory variables (expression levels of genes). Phospho\-enol\-pyru\-vat-carboxykinase (physiological phenotype) is chosen as the response variable. For details concerning this data example we refer to \citet{Lan2006} and \citet{Song2015}. We first apply the AdaSub algorithm with \(q=5\), \(K=n\) and \(T=500,000\) and choose the (negative) \(\text{EBIC}_{0.6}\) as the selection criterion (computational time approximately 20 minutes). 

%\vspace{-0.5cm}

\begin{figure}[ht]
\centering
\subfloat[\(\gamma=0.6\)]{\includegraphics[width=.45\linewidth]{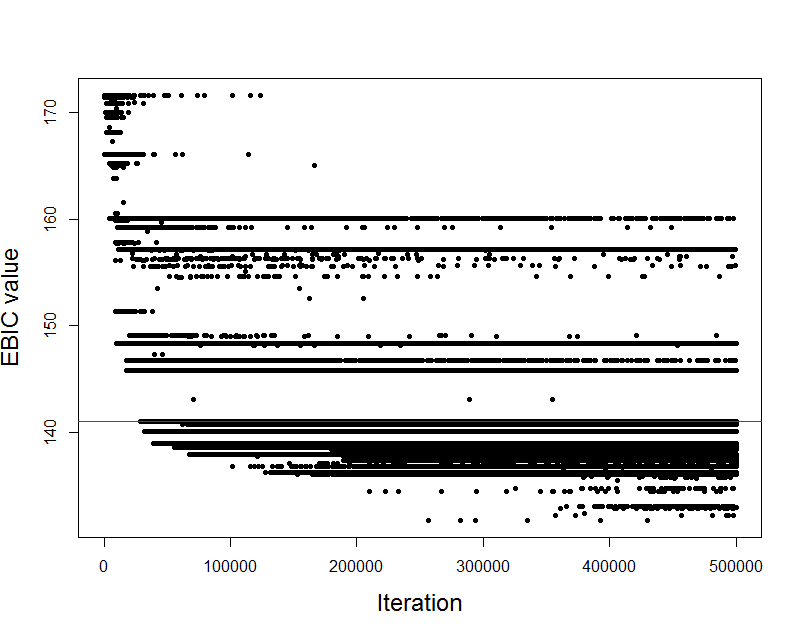}
 \label{Evolution_PCR_06}}\qquad
\subfloat[\(\gamma=1\)]{\includegraphics[width=.45\linewidth]{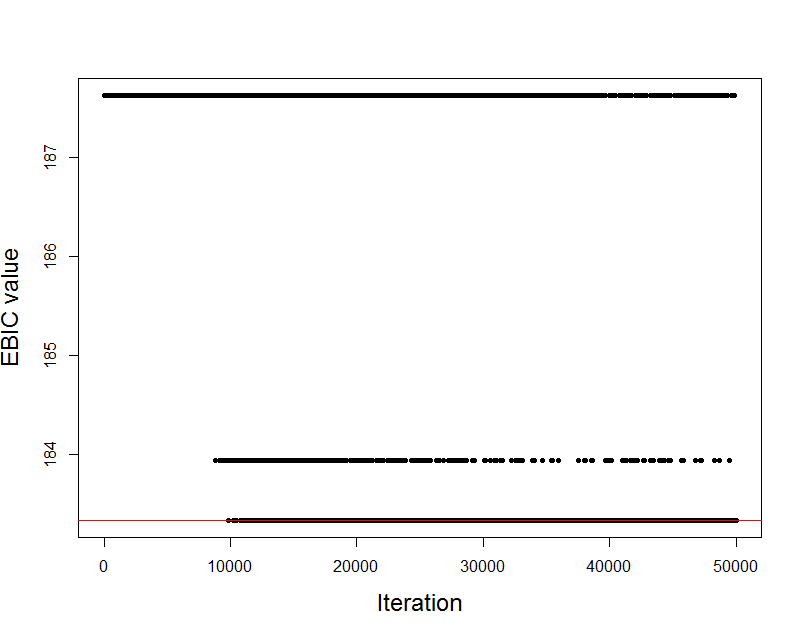}
 \label{Evolution_PCR_1}}
\caption{AdaSub for PCR-data. Plot of the evolution of \(\text{EBIC}_\gamma(S^{(t)})\) along iterations (\(t\)). The red line indicates the \(\text{EBIC}_{\gamma}\)-value of the thresholded model \(\hat{S}_{0.9}\).}
\label{Evolution_PCR}
\end{figure}

The evolution of the values \(\text{EBIC}_{0.6}(S^{(t)})\) along the iterations (\(t\)) is given in Figure \ref{Evolution_PCR_06}. The criterion \(\text{EBIC}_{0.6}\) seems to be too liberal for the given situation resulting in high uncertainty concerning the \(\text{EBIC}_{0.6}\)-optimal model and (possibly) failure of the OIP condition.  
The thresholded model \(\hat{S}_{0.9}\) selected by AdaSub consists of five variables (genes), while the ``best'' model \(\hat{S}_{\textrm{b}}\) consists of ten variables (genes); see Table \ref{tab:PCR} for a summary of the results. 

\begin{table}[ht]
\small
\caption{\small Results for PCR data in terms of selected genes and mean/median CV-errors for the final model selected by SAM as well as the ``best'' models (\(\hat{S}_{\text{b}}\)) and thresholded models (\(\hat{S}_{0.9}, \hat{S}_{6}\)) from AdaSub for \(\text{EBIC}_{0.6}\) and \(\text{EBIC}_{1}\).  
}
\label{tab:PCR}      
\tabcolsep=0.11cm
\small
\resizebox{\textwidth}{!}{%
\begin{tabular}{llll}
\hline\noalign{\smallskip} \small
Model & Selected variables (genes) & Mean CV & Median CV  \\
\noalign{\smallskip}\hline\noalign{\smallskip}
SAM model & 1429089\_s\_at, 1430779\_at, 1432745\_at, 1437871\_at, 1440699\_at, 1459563\_x\_at & 0.084 & 0.044 \\
\(\text{EBIC}_{0.6}\): \(\hat{S}_{\text{b}}\)& 1428239\_at, 1433056\_at, 1437871\_at, 1438937\_x\_at,  1440505\_at,  & 0.030 & 0.012 \\
																								& 1442771\_at, 1444\allowbreak471\_at, 1445645\_at, 1446035\_at, 1455361\_at &   & \\
\(\text{EBIC}_{0.6}\): \(\hat{S}_{0.9}\)& 1437871\_at,  1438937\_x\_at, 144277\allowbreak1\_at, 1446035\_at, 1455361\_at & 0.116 & 0.056 \\ 
\(\text{EBIC}_{0.6}\): \(\hat{S}_{6}\)& 1428239\_at, 1437871\_at, 1438937\_x\_at, 1442771\_at, 1446035\_at, 1455361\_at & 0.090 & 0.041 \\
\(\text{EBIC}_{1}\):  \(\hat{S}_{0.9}\), \(\hat{S}_{\text{b}}\) & 1438937\_x\_at & 0.403 & 0.158 \\ 
\noalign{\smallskip}\hline
\end{tabular}}
\end{table}

In order to compare the predictive performances of the selected models we compute the mean and median leave-one-out-cross-validation squared errors (CV-errors) for each fixed model as described in \citet{Song2015}. Note that the CV-errors of the final models (with variables selected based on the full dataset) generally tend to underestimate the true generalization errors on independent test data (compare \citealp{ambroise2002}) and only serve for a comparison of models with the same number of selected variables. 
It can be seen that the CV-errors of the thresholded model \(\hat{S}_{0.9}\) with five genes and the CV-errors of the ``best'' model \(\hat{S}_{\textrm{b}}\) with ten genes are of the same order or even lower than the errors of the best SAM model with five and ten explanatory variables, respectively (compare Figure 5 in \citealp{Song2015}). 
In order to compare the final model from SAM to a model with six genes selected by AdaSub we proceed in the following way: Let \(g:\mathcal{P}\rightarrow\mathcal{P}\) be a permutation such that \(r_{g(1)}^{(T)} \geq r_{g(2)}^{(T)} \geq \ldots \geq r_{g(p)}^{(T)}\). Assuming no ``ties'', for \(k\in\mathcal{P}\) we define \(\hat{S}_k := \{j\in\mathcal{P}:\, g^{-1}(j)\leq k\}\) to be the thresholded model from AdaSub with exactly \(|\hat{S}_k| = k\) variables. 
In Table \ref{tab:PCR} it can be seen that even though the thresholded model \(\hat{S}_6\) from AdaSub with six genes is totally different from the model selected by SAM, it has similar predictive performance.

We now apply AdaSub with \(q=5\), \(K=n\) and \(T=50,000\) and choose the (negative) \(\text{EBIC}_{1}\) as a selection criterion that enforces more sparsity (comp. time approximately 1 minute and 30 seconds). The evolution of the values \(\text{EBIC}_{1}(S^{(t)})\) along the iterations (\(t\)) is given in Figure \ref{Evolution_PCR_1}. Now \(\hat{S}_{0.9}\) and \(\hat{S}_{\textrm{b}}\) coincide, consisting both of only one gene (1438937\_x\_at). Note that for the criterion \(\text{EBIC}_{0.6}\) the gene 1438937\_x\_at is also included in the thresholded model \(\hat{S}_{0.9}\), in the ``best'' model \(\hat{S}_{\textrm{b}}\) and in the thresholded model \(\hat{S}_{1}\) with exactly one gene selected by AdaSub, whereas it is not included in the final model selected by SAM.

%\vspace{-3mm}

\section{Discussion}\label{sec:discussion}
AdaSub has been introduced in order to solve the natural \(\ell_0\)-regularized optimization problem for high-dimen\-sional variable selection. 
If the ordered importance property (OIP) is satisfied, then AdaSub converges against the optimal solution of the generally NP-hard \(\ell_0\)-re\-gu\-lar\-ized optimization problem. Furthermore, AdaSub provides a stable thresholded model even when OIP is not guaranteed to hold. 
It has been demonstrated through simulated and real data examples that the performance of AdaSub is very competitive for high-di\-men\-sio\-nal variable selection in comparison to state-of-the-art methods like the Adaptive Lasso, the SCAD, Tilting or the Bayesian split-and-merge approach (SAM). Furthermore, in the considered sparse high-dimensional settings, AdaSub in combination with the EBIC as the selection criterion performs favorably in comparison to widely used regularization methods tuned with cross-validation (see Section\ \ref{sec:additional}). It is notable that AdaSub outperforms Stability Selection with the Lasso in many situations, which underpins the argument that usual subsampling in combination with an \(\ell_1\)-type method might not be optimal in a high-dimensional situation. On the contrary, the application of adaptive ``subsampling'' in the space of explanatory variables can efficiently reduce the intractable \(\ell_0\)-type high-dimensional problem to solvable low-dimensional sub-problems even in very high-dimensional situations with ten thousands of possible explanatory variables. 

In this paper we have focused on variable selection in linear regression models, but the proposed AdaSub method is more general and can for example be applied to any variable selection problem in the framework of generalized linear models (GLMs). The practical problem is then that --- to the best of our knowledge --- there is no efficient algorithm like ``leaps and bounds'' which could be used for solving the low-dimensional sub-problems for a GLM within reasonable computational time. In particular, a full enumeration is costly since the ML-estimators for the single models are not given in closed form, in general. A possible solution would be to use heuristic algorithms in place of a full enumeration in order to derive approximate solutions for the sub-problems. It is then desirable to extent the convergence properties of AdaSub also to those situations. 

Furthermore, even though we have focused on the EBIC as the selection criterion, the AdaSub method is very general and can be combined with any other selection criterion. It is also possible to use other variable selection methods such as \(\ell_1\)-type methods (like the Lasso) for ``solving'' the sampled sub-problems in AdaSub. However, the theoretical results concerning the limiting properties of AdaSub are based on the assumption of optimizing a discrete function on the model space, so the presented limiting properties are not directly applicable for such alternative methods. The investigation of the performance of AdaSub for different choices of the selection procedure is an interesting topic for future research.

Another line of our current research concerns the further exploration of the sufficient condition for the \(C\)-optimal convergence of AdaSub and particularly attempts to relax OIP by weaker sufficient conditions. 
We want to emphasize that in this work we have focused on the algorithmic convergence of AdaSub against the best model according to a given criterion (as the number of iterations \(T\) diverges). Based on the presented analysis, depending on the properties of the employed selection criterion, one may derive specific model selection consistency results (as the sample size \(n\) and the number of variables \(p\) diverge with a certain rate), as indicated in Section \ref{sec:limiting}.   Furthermore, it would be desirable to obtain more general theoretical results concerning the ``speed of convergence'' of AdaSub.  
Finally, in subsequent work we develop modifications of the presented algorithm for sampling from high-dimensional posterior model distributions in a fully Bayesian framework.

\clearpage

\appendix

\section{Appendix}

\noindent In Section \ref{sec:theory} of this appendix we provide the theoretical details concerning the limiting properties of AdaSub. In Section \ref{sec:illustrative} we illustrate the application of the AdaSub algorithm on a high-dimensional simulated data example and discuss typical ``diagnostic plots'' for the convergence of the algorithm. In Section \ref{sec:additional} we present further results of the simulation study given in Section \ref{sec:simstudy}.

\subsection{Theoretical details}\label{sec:theory}

In this section we theoretically investigate the limiting properties of AdaSub (see Algorithm \ref{AdaSub}) by analysing the evolution (along the iterations \(t\in\N\)) of the selection probabilities 
%\vspace{-1ex}
\begin{equation}r_j^{(t)} = \frac{q+K \sum_{i=1}^t W_j^{(i)}}{p + K \sum_{i=1}^t Z_j^{(i)} } \, ,\end{equation}
where \(Z_j^{(i)}=1_{V^{(i)}}(j)\) and \(W_j^{(i)}=1_{f_C(V^{(i)})}(j) = 1_{S^{(i)}}(j)\) for \(j\in\mathcal{P}\), \(i\in\N\).

In order to describe the information available after iteration \(t\) of the AdaSub algorithm, we define a filtration \((\mathcal{F}^{(t)})_{t\in\N_0}\) on the underlying probability space \(\Omega\) of the process: Let \(\mathcal{F}^{(0)}:=\{\emptyset,\Omega\}\) and for \(t\in\N\) let
%\vspace{-2mm}
\begin{equation}\mathcal{F}^{(t)}:= \sigma(W_1^{(1)},Z_1^{(1)}, W_2^{(1)},Z_2^{(1)}, \dots , W_p^{(1)},Z_p^{(1)}, \dots, W_p^{(t)},Z_p^{(t)}) \label{eq:filtration}\end{equation}
be the \(\sigma\)-algebra generated by \(W_1^{(1)}, \dots, Z_p^{(t)}\).
Then by the construction of AdaSub we have for \(t\in\N_0\) and \(j\in\mathcal{P}\): 
%\vspace{-2mm}
\begin{equation} r_j^{(t)} = P(Z_j^{(t+1)}=1 ~|~ \mathcal{F}^{(t)} ) = 1 - P(Z_j^{(t+1)}=0 ~|~ \mathcal{F}^{(t)} ) \, .\end{equation}
In addition, for \(t\in\N_0\) and \(j\in\mathcal{P}\) we define 
%\vspace{-2mm}
\begin{align}p_j^{(t+1)}&:= P( W_j^{(t+1)} =1 ~|~ Z_j^{(t+1)} = 1, \mathcal{F}^{(t)}) \nonumber \\
&= 1 - P(W_j^{(t+1)}=0 ~|~Z_j^{(t+1)}=1 \, , \mathcal{F}^{(t)} ) \, ,\label{eq:condprob}\end{align}
where for events \(A,B\in\mathcal{F}^{(t+1)}\) the conditional probabilities under \(\mathcal{F}^{(t)}\) are defined by \(P( A ~|~ \mathcal{F}^{(t)}) =  E[1_A ~|~ \mathcal{F}^{(t)}]\) and \(P( A ~|~ B , \mathcal{F}^{(t)}) = \frac{E[1_{A\cap B} ~|~ \mathcal{F}^{(t)}]}{ E[1_B ~|~ \mathcal{F}^{(t)}]}\) almost surely (a.s.) on the set \(\{ E[1_B ~|~ \mathcal{F}^{(t)}]>0\}\), while we set \(P( A ~|~ B , \mathcal{F}^{(t)}) = 0\) a.s. on \(\{ E[1_B ~|~ \mathcal{F}^{(t)}]=0\}\). 

In the following, we will make repeated use of the following generalization of Borel-Cantelli's lemma and the strong law of large numbers, which is due to \citet{Dubins1965}.
\begin{theorem}[Dubins and Freedman, 1965]\label{Theorem Borel} 
Let \((\mathcal{F}_n)_{n\in\N_0}\) be a filtration and \(A_n\in\mathcal{F}_n\) for \(n\in\N\). %\\[1mm]
 For \(i\in\N\) define \(q_i:=P(A_i~|~\mathcal{F}_{i-1})\), then:
\begin{itemize}
\item[(a)] On \(\left\{ \sum_{i=1}^\infty q_i < \infty \right\}\) we almost surely have \(\sum_{i=1}^\infty 1_{A_i} <\infty\).
\item[(b)] On \(\left\{ \sum_{i=1}^\infty q_i = \infty \right\}\) we have \[ \frac{\sum_{i=1}^n 1_{A_i}}{\sum_{i=1}^n q_i} \overset{\text{a.s.}}{\longrightarrow} 1, ~ n\rightarrow \infty \, . \]
\end{itemize}
\end{theorem}

A first simple but important observation is that, with probability 1, each variable \(X_j\) with \(j\in\mathcal{P}\) is considered infinitely many times in the model search of AdaSub.  
\begin{lemma}\label{lemma_borel}
Let $j\in\mathcal{P}$. Then it holds \[ P\left(\sum_{t=1}^\infty 1_{V^{(t)}}(j)=\infty\right) = P\left(\sum_{t=1}^\infty Z_j^{(t)} =\infty\right) =1 \, . \]
\end{lemma}

\begin{proof}
\noindent Let \((\mathcal{F}^{(t)})_{t\in\N_0}\) be the filtration given by equation (\ref{eq:filtration}). Fix \(j\in\mathcal{P}\) and for \(t\in\N\) let \(A_j^{(t)}:=\{Z_j^{(t)} = 1 \} \in \mathcal{F}^{(t)}\). For \(t\in\N\) we have 
\[q_j^{(t)} := P(A_j^{(t)} ~|~ \mathcal{F}^{(t-1)}) = r_j^{(t-1)} \geq \frac{q}{p+K(t-1)} \, \]
and therefore \( \sum_{i=1}^\infty q_j^{(i)} \overset{\text{a.s.}}{=} \infty \). So by Theorem \ref{Theorem Borel} we conclude
\[ \frac{\sum_{i=1}^t 1_{A_j^{(i)}}} {\sum_{i=1}^t q_j^{(i)} } \overset{\text{a.s.}}{\longrightarrow}  1 , \, t\rightarrow\infty \,. \]
Since \(\sum_{i=1}^\infty q_j^{(i)} \overset{\text{a.s.}}{=} \infty \), we also have \[ \sum_{i=1}^\infty 1_{A_j^{(i)}} = \sum_{i=1}^\infty Z_j^{(i)} \overset{\text{a.s.}}{=} \infty \, . \]
\end{proof}

The following theorem shows that the convergence of \(p_j^{(t)}\) as \(t\rightarrow\infty\) determines the convergence of \(r_j^{(t)}\). This result will be the key ingredient needed for the proof of the \(C\)-optimal convergence of AdaSub (Theorem \ref{OIP_Theorem}).

\begin{theorem}\label{keylemma}
For each \(j\in\mathcal{P}\) we have: If \(p_j^{(t)}\overset{\text{a.s.}}\rightarrow p_j^*\) as \(t\rightarrow\infty\) for some random variable \(p_j^*\), then \(r_j^{(t)}\overset{\text{a.s.}}\rightarrow p_j^*\) as \(t\rightarrow\infty\).
\end{theorem}

\begin{proof}
\noindent Fix \(j\in\mathcal{P}\) and suppose that \(p_j^{(t)}\overset{\text{a.s.}}\rightarrow p_j^*\) as \(t\rightarrow\infty\). 
We apply Theorem \ref{Theorem Borel} again, but using a different filtration \((\mathcal{G}^{(t)})_{t\in\N_0}\), where 
\[ \mathcal{G}^{(0)} = \sigma\left( \left\{ Z_j^{(1)} :~ j\in\mathcal{P}\right\}\right) \, ,\] 
and
\[ \mathcal{G}^{(t)} = \sigma\left( \left\{ Z_j^{(i)} :~ j\in\mathcal{P}, i=1,\dots,t+1\right\} \cup \left\{ W_j^{(i)} :~ j\in\mathcal{P}, i=1,\dots,t\right\} \right) , t\in\N . \] 
Further let \(A_j^{(t)} := \{W_j^{(t)}=1\}\in \mathcal{G}^{(t)}\), for \(t\in\N\), with 
\[ q_j^{(t)} := P\left(A_j^{(t)} ~|~ \mathcal{G}^{(t-1)} \right) = P\left(W_j^{(t)} = 1 ~|~ \mathcal{G}^{(t-1)} \right)  =p_j^{(t)} Z_j^{(t)} \, \]
and \[\Omega':=\left\{\omega\in\Omega:~ \sum_{i=1}^\infty Z_j^{(i)}(\omega) = \infty \right\}\, .\] By Lemma \ref{lemma_borel} we have \(P(\Omega')=1\).
Let \[\Omega_1:=\{\omega\in\Omega':~ p_j^{(t)}(\omega) \rightarrow p_j^*(\omega), t\rightarrow\infty \text{ with } p_j^*(\omega)\in(0,1]\}\] and \[\Omega_2:=\{\omega\in\Omega':~ p_j^{(t)}(\omega) \rightarrow p_j^*(\omega), t\rightarrow\infty \text{ with } p_j^*(\omega)=0\}\,.\]
Then on \(\Omega_1\) we have
\[ \sum_{i=1}^\infty q_j^{(i)} = \sum_{i=1}^\infty p_j^{(i)} Z_j^{(i)} \,\overset{(\text{a1})}{=}\, \sum_{i=1}^\infty p_j^{(l^\omega_i)} = \infty \, ,\]
where equality in (a1) holds since for each \(\omega\in\Omega_1\) there exists an increasing sequence \((l^\omega_i)_{i\in\N}\) with \(l^\omega_i\in\N\) and \(Z_j^{(l^\omega_i)}(\omega)=1\) for all \(i\in\N\). 
So on \(\Omega_1\) we have for \(t\) large enough (to avoid division by 0)
\[ \lim_{t\rightarrow \infty} \frac{\sum_{i=1}^t q_j^{(i)}}{\sum_{i=1}^t Z_j^{(i)} } = \lim_{t\rightarrow \infty} \frac{\sum_{i=1}^t p_j^{(i)} Z_j^{(i)}}{\sum_{i=1}^t Z_j^{(i)} } 
= \lim_{t\rightarrow \infty} \frac{\sum_{i=1}^t p_j^{(l^\omega_i)}}{t } = p_j^* \, ,\]  
which holds for those increasing sequences \((l^\omega_i)_{i\in\N}\) that additionally fulfil \(Z_j^{(i)}(\omega)\) \(=0\) for all \(i\notin\{l_k^{(\omega)}:~ k\in\N\}\). Here we applied Cauchy's limit theorem using the fact that \(p_j^{(l^\omega_i)} \rightarrow p_j^*\) as \(i\rightarrow\infty\). Combining this result with Theorem \ref{Theorem Borel} it follows that on \(\Omega_1\) we have (for \(t\) large enough)
\[ \frac{\sum_{i=1}^t W_j^{(i)}}{\sum_{i=1}^t Z_j^{(i)} } = \underbrace{\frac{\sum_{i=1}^t W_j^{(i)}}{\sum_{i=1}^t q_j^{(i)} }}_{\overset{\text{a.s.}}{\rightarrow} 1} \underbrace{\frac{\sum_{i=1}^t q_j^{(i)}}{\sum_{i=1}^t Z_j^{(i)} }}_{\overset{\text{a.s.}}{\rightarrow} p_j^*} \overset{\text{a.s.}}{\longrightarrow} p_j^* ,~~ t\rightarrow \infty \, .\]
Now on \(\Omega_2\cap\left\{ \sum_{i=1}^\infty q_j^{(i)} = \infty \right\}\) we can use the same argument as above and obtain 
\[ \frac{\sum_{i=1}^t W_j^{(i)}}{\sum_{i=1}^t Z_j^{(i)} } \overset{\text{a.s.}}{\longrightarrow} p_j^* ,~~ t\rightarrow \infty \, .\]
On \(\Omega_2\cap\left\{ \sum_{i=1}^\infty q_j^{(i)} < \infty \right\}\) we almost surely have \(\sum_{i=1}^\infty W_j^{(i)} < \infty\) by Theorem \ref{Theorem Borel}, but since \(\sum_{i=1}^t Z_j^{(i)} \overset{\text{a.s.}}{\rightarrow} \infty \) it also follows that 
\[ \frac{\sum_{i=1}^t W_j^{(i)}}{\sum_{i=1}^t Z_j^{(i)} } \overset{\text{a.s.}}{\longrightarrow} 0 = p_j^* ,~~ t\rightarrow \infty \, .\]
Noting that \(P(\Omega_1\cup\Omega_2)=1\) by assumption and combining the arguments on \(\Omega_1\) and \(\Omega_2\), we conclude that on \(\Omega\) we have
\[ r_j^{(t)} = \frac{q+K \sum_{i=1}^t W_j^{(i)}}{p + K \sum_{i=1}^t Z_j^{(i)} }  = 
 \frac{ \tfrac{q}{K \sum_{i=1}^t Z_j^{(i)} } + \frac{\sum_{i=1}^t W_j^{(i)}}{ \sum_{i=1}^t Z_j^{(i)} } } { \frac{p}{  K \sum_{i=1}^t Z_j^{(i)} }  + 1 } \overset{\text{a.s.}}{\longrightarrow} p_j^* ,~~ t\rightarrow \infty \, .\]
\end{proof}

\begin{definition}\label{def:OIP'}
Given that data \(\mathcal{D}=(\bs X,\bs Y)\) is observed, let \(C_{\mathcal{D}}:\mathcal{M}\rightarrow\R\) be a selection criterion with well-defined function \(f_C\) and \(C\)-optimal model \(S^*=f_C(\mathcal{P})=\{j_1,\dots,j_{s^*}\}\) of size \(s^*=|S^*|\). 
Then the selection criterion \(C\) is said to fulfil the \textit{ordered importance property (OIP')} for the sample \(\mathcal{D}\), if there exists a permutation \((k_1,\dots,k_{s^*})\) of \((j_1,\dots,j_{s^*})\) such that for each \(i=1,\dots,s^*-1\) it holds
\begin{equation} k_i \in f_C(V) ~~\text{ for all } V\subseteq \mathcal{P} \setminus N_{i-1} \text{ with } \{k_1,\dots,k_i\}\subseteq V \,, \label{OIP:eq'}\end{equation}
where 
\begin{equation} N_0 := \{j\in\mathcal{P}:~ j\notin f_C(V) \text{ for all } V\subseteq \mathcal{P} \}  \end{equation}
and 
\begin{equation} N_i := \{j\in\mathcal{P}:~ j\notin f_C(V) \text{ for all } V\subseteq \mathcal{P} \setminus N_{i-1} \text{ with }  \{k_1,\dots,k_i\}\subseteq V \} \,.\end{equation}
\end{definition}

\begin{remark} \label{rem:OIP}
Note that \(S^*=f_C(V)\) for all \(V\subseteq\mathcal{P}\) with \(S^*\subseteq V\). Therefore (\ref{OIP:eq'}) always holds for \(i=s^*\) since \(k_{s^*}\in S^*\). Furthermore, we have 
%\vspace{-2mm}
\begin{equation*} N_{s^*} = \{j\in\mathcal{P}:~ j\notin f_C(V) \text{ for all } V\subseteq \mathcal{P} \setminus N_{s^*-1} \text{ with }  S^*\subseteq V \} =\mathcal{P}\setminus S^* \,. \end{equation*}
\end{remark}

\begin{remark} \label{rem:OIP'}
Note that OIP' of Definition \ref{def:OIP'} is a weaker condition than OIP of Definition \ref{def:OIP} in Section~\ref{sec:limiting} (i.e.\ OIP implies OIP'). Indeed, equation (\ref{OIP:eq}) of Section~\ref{sec:limiting} implies equation (\ref{OIP:eq'}) since the required condition is only imposed on a generally smaller set of subsets \(V\). 
\end{remark}

The next theorem shows that OIP' (and thus also OIP) is really a sufficient condition for the \(C\)-optimal convergence of AdaSub against \(S^*\).

\begin{theorem}\label{OIP_Theorem'}
Given that dataset \(\mathcal{D}=(\bs X,\bs Y)\) is observed, let \(C_{\mathcal{D}}:\mathcal{M}\rightarrow\R\) be a selection criterion with well-defined function \(f_C\) and \(C\)-optimal model \(S^*\). 
Suppose that the ordered importance property (OIP') is satisfied. Then AdaSub converges to the \(C\)-optimal model in the sense of Definition \ref{def:correct_conv}. 
\end{theorem}

\begin{proof}
\noindent Let \(S^*=f_C(\mathcal{P})=\{j_1,\dots,j_{s^*}\}\) be the \(C\)-optimal model of size \(s^*=|S^*|\).
Since OIP' is satisfied there exists a permutation \((k_1,\dots,k_{s^*})\) of \((j_1,\dots,j_{s^*})\) such that equation (\ref{OIP:eq'}) holds for each \(i=1,\dots,s^*-1\) (with corresponding sets \(N_0\subseteq N_1\subseteq\ldots \subseteq N_{s^*}\)). 
Let \(j\in N_0\). Then by definition we have \(j\notin f_C(V)\) for all \(V\subseteq \mathcal{P}\), so that 
\[ p_j^{(t+1)} = P(j\in f_C(V^{(t+1)})~|~j\in V^{(t+1)}, \mathcal{F}^{(t)}) = 0 \]
for all \(t\in\N_0\). With Theorem \ref{keylemma} we conclude that \(r_{j}^{(t)} \overset{\text{a.s.}}{\rightarrow} 0\) as \(t\rightarrow\infty\) for \(j\in N_0\). \\
Now by OIP' we have \(k_1\in f_C(V)\) for all \(V\subseteq\mathcal{P}\setminus N_0\) with \(\{k_1\}\subseteq V\), so that for all \(t\in\N_0\) we have
\[ P( k_1\in f_C(V^{(t+1)}) ~|~ k_1 \in V^{(t+1)}, N_0 \cap V^{(t+1)}=\emptyset , \mathcal{F}^{(t)} ) = 1 \, .\] 
Note that by the independence of the Bernoulli trials in AdaSub we have 
\begin{align*}  P( N_0 \cap V^{(t+1)} = \emptyset ~|~ k_1 \in V^{(t+1)} , \mathcal{F}^{(t)} )  &= P( N_0 \cap V^{(t+1)} = \emptyset ~|~ \mathcal{F}^{(t)} ) \\
&= \prod_{l\in N_0} \left(1-r_l^{(t)}\right) \overset{\text{a.s.}}{\rightarrow} 1 \end{align*}
and therefore
%\vspace{-2mm}
\[ P( N_0 \cap V^{(t+1)} \neq \emptyset  ~|~ k_1 \in V^{(t+1)} , \mathcal{F}^{(t)} ) = 1- \prod_{l\in N_0} \left(1-r_l^{(t)}\right) \overset{\text{a.s.}}{\rightarrow} 0 \,.\] 
Thus we conclude with the law of total probability that
%\vspace{-2mm}
\begin{align*}
p_{k_1}^{(t+1)} &= P( k_1\in f_C(V^{(t+1)}) ~|~ k_1 \in V^{(t+1)} , \mathcal{F}^{(t)} ) \\
                &= P( k_1\in f_C(V^{(t+1)}) ~|~ k_1 \in V^{(t+1)},N_0 \cap V^{(t+1)} = \emptyset , \mathcal{F}^{(t)} ) 	\\
								&~~~~\times  \prod_{l\in N_0} \left(1-r_l^{(t)}\right)  \\
								&~~~+ P( k_1\in f_C(V^{(t+1)}) ~|~ k_1\in V^{(t+1)} , N_0 \cap V^{(t+1)} \neq \emptyset , \mathcal{F}^{(t)} ) \\
								&~~~~\times \left(1-\prod_{l\in N_0} \left(1-r_l^{(t)}\right)\right)  \\
								&\overset{\text{a.s.}}{\rightarrow} 1 \times 1 + 0 = 1 , ~~~ t\rightarrow\infty \, .\end{align*} 
By Theorem \ref{keylemma} we also obtain \(r_{k_1}^{(t)} \overset{\text{a.s.}}{\rightarrow} 1 \) as \(t\rightarrow\infty\). \\
Now let \(j\in N_1\setminus N_0\). Then by definition we have \(j\notin f_C(V)\) for all \(V\subseteq \mathcal{P}\setminus N_0\) with \(\{k_1\}\subseteq V\), so that 
%\vspace{-2mm}
\[ P(j\in f_C(V^{(t+1)})~|~j\in V^{(t+1)}, N_0 \cap V^{(t+1)} = \emptyset , k_1 \in V^{(t+1)}, \mathcal{F}^{(t)}) = 0 \]
for all \(t\in\N_0\).
Note that again by the independence of the Bernoulli trials in AdaSub we have 
%\vspace{-2mm}
\[  P(N_0 \cap V^{(t+1)} =\emptyset, k_1 \in V^{(t+1)} ~|~ j \in V^{(t+1)} , \mathcal{F}^{(t)} ) =\prod_{l\in N_0} \left(1-r_l^{(t)}\right) \times r_{k_1}^{(t)} \overset{\text{a.s.}}{\rightarrow} 1 \, .\]
Thus we similarly conclude with the law of total probability that
%\vspace{-2mm}
\begin{align*}
p_{j}^{(t+1)} &= P( j\in f_C(V^{(t+1)}) ~|~ j \in V^{(t+1)} , \mathcal{F}^{(t)} ) \\[-1mm]
                &= P( j\in f_C(V^{(t+1)}) ~|~ k_1,j \in V^{(t+1)},N_0 \cap V^{(t+1)} = \emptyset, \mathcal{F}^{(t)} ) \\
								&~~~~\times \prod_{l\in N_0} \left(1-r_l^{(t)}\right) \times r_{k_1}^{(t)} \\[-1mm]
									&~~~+ \ldots \\[-1mm]
								&\overset{\text{a.s.}}{\rightarrow} 0 \times 1 + 0 = 0 , ~~~ t\rightarrow\infty \, .\end{align*} 
By Theorem \ref{keylemma} we also obtain \(r_{j}^{(t)} \overset{\text{a.s.}}{\rightarrow} 0 \) as \(t\rightarrow\infty\) for \(j\in N_1\setminus N_0\). \\
Now by OIP' we have \(k_2\in f_C(V)\) for all \(V\subseteq\mathcal{P}\setminus N_1\) with \(\{k_1,k_2\}\subseteq V\), so that for all \(t\in\N_0\) we have
\[ P( k_2\in f_C(V^{(t+1)}) ~|~ k_2 \in V^{(t+1)}, N_1 \cap V^{(t+1)} = \emptyset , k_1 \in V^{(t+1)}, \mathcal{F}^{(t)} ) = 1 \, .\] 
Note that again by the independence of the Bernoulli trials in AdaSub we have 
\[  P(N_1 \cap V^{(t+1)} = \emptyset , k_1 \in V^{(t+1)} ~|~ \mathcal{F}^{(t)} ) = \prod_{l\in N_1} \left( 1- r_l^{(t)}\right) \times r_{k_1}^{(t)} \overset{\text{a.s.}}{\rightarrow} 1 \, .\]
Thus we similarly conclude with the law of total probability that
\begin{align*}
p_{k_2}^{(t+1)} &= P( k_2\in f_C(V^{(t+1)}) ~|~ k_2 \in V^{(t+1)} , \mathcal{F}^{(t)} ) \\
                &= P( k_2\in f_C(V^{(t+1)}) ~|~ k_1,k_2 \in V^{(t+1)},N_1 \cap V^{(t+1)} = \emptyset ,\, \mathcal{F}^{(t)} ) \\
								&~~~~\times \prod_{l\in N_1} \left( 1- r_l^{(t)}\right) \times r_{k_1}^{(t)}   \\
								&~~~+ \ldots  \\
								&\overset{\text{a.s.}}{\rightarrow} 1 \times 1 + 0 = 1 , ~~~ t\rightarrow\infty \, .\end{align*} 
By Theorem \ref{keylemma} we also obtain \(r_{k_2}^{(t)} \overset{\text{a.s.}}{\rightarrow} 1 \) as \(t\rightarrow\infty\). \\
Proceeding by induction we similarly conclude that for each \(i=2,\dots,s^*-1\) we have \(r_{j}^{(t)} \overset{\text{a.s.}}{\rightarrow} 0 \) as \(t\rightarrow\infty\) for all \(j\in N_i\setminus N_{i-1}\); and for each \(i=3,\dots,s^*-1\) we have \(r_{k_i}^{(t)} \overset{\text{a.s.}}{\rightarrow} 1 \) as \(t\rightarrow\infty\). \\
Note that \(k_{s^*}\in S^* = f_C(V)\) for all \(V\subseteq\mathcal{P}\) with \(\{k_1,\dots,k_{s^*}\}\subseteq V\) and that \(N_{s^*} = \mathcal{P}\setminus S^*\) (see Remark \ref{rem:OIP}). Therefore, by using the same arguments, we also obtain \(r_{k_{s^*}}^{(t)} \overset{\text{a.s.}}{\rightarrow} 1 \) as \(t\rightarrow\infty\) and 
\(r_{j}^{(t)} \overset{\text{a.s.}}{\rightarrow} 0\) as \(t\rightarrow\infty\) for all \(j\in N_{s^*} = \mathcal{P}\setminus S^*\).
This completes the proof.
\end{proof}

\begin{corollary}
If \(|S^*|\leq 1\), then OIP is satisfied and therefore AdaSub converges to the \(C\)-optimal model. 
\end{corollary}

\begin{corollary}\label{OIP_replacement}
Let \(S^*=\{j_1,\dots,j_{s^*}\}\) and let \(D=\{l_1,\dots,l_d\}\subseteq S^*\) be of maximal cardinality \(|D|=d\) such that there exists a permutation \((k_1,\dots,k_d)\) of \((l_1,\dots,l_d)\) such that for all \(i=1,\dots,d\) we have
%\vspace{-2mm}
\begin{equation} k_i\in f_C(V) ~~\text{ for all } V\subseteq \mathcal{P} \setminus N_{i-1} \text{ with } \{k_1,\dots,k_i\}\subseteq V \,, \label{OIPalt:eq}\end{equation}
where the sets \(N_0,\dots,N_d\) are defined as in Definition \ref{def:OIP}. In particular we have 
%\vspace{-2mm}
 \begin{equation}N_d =\{j\in\mathcal{P}:~ j\notin f_C(V) \text{ for all } V\subseteq \mathcal{P} \setminus N_{d-1} \text{ with } \{k_1,\dots,k_d\}\subseteq V\}\,.\end{equation}
Then for all \(j\in D\) we have \(r_{j}^{(t)}\overset{\text{a.s.}}{\rightarrow} 1\), \(t\rightarrow\infty\) and for all \(j\in N_d\) we have \(r_{j}^{(t)}\overset{\text{a.s.}}{\rightarrow}0\), \(t\rightarrow\infty\).
\end{corollary}

\begin{proof}
The proof is along the lines of the proof of Theorem \ref{OIP_Theorem'}, using the (partial) permutation \((k_1,\dots,k_d)\) of variables in \(D\subseteq S^*\) instead of the (full) permutation \((k_1,\dots,k_{s^*})\) of all variables in \(S^*\).  
\end{proof}

\begin{remark}\label{remark:unique}
In Theorem \ref{OIP_Theorem'} it is assumed that the function \(f_C\) is well-defined, i.e.\ that the solutions \(S^{(t)}=f_C(V^{(t)})\) are unique for all subspaces \(V^{(t)}\subseteq\mathcal{P}\).  
In case of non-uniqueness of the solutions to sub-problems \(S^{(t)}=f_C(V^{(t)})\) the AdaSub algorithm and the convergence result can be slightly adjusted (see \citealp[Remark 4.4]{staerk2018} for details). 
In particular, perfect multicollinearity among a set of explanatory variables \(\{\bs X_j; j\in S\}\) with \(S\in\mathcal{M}=\{S\subseteq\mathcal{P}:\, |S|< n-2 \}\) can lead to non-uniqueness of the solution to \(\ell_0\)-type criteria; in such a case it may generally be preferred to deal with the multicollinearity issue before proceeding with data-driven variable selection, e.g.\ by reducing the number of considered variables based on subject-matter knowledge or by aiming to increase the sample size. 
\end{remark}

Finally, the following remark provides analytical results regarding the speed of convergence of AdaSub under the finite-sample PF assumption. %(compare Section \ref{sec:limiting}).  

\begin{remark} \label{remark:speed}
For \(m\in\N\) and \(j\in\mathcal{P}\) let  
\begin{equation}
T_j^{(m)} = \min\Big\{t\in\N:\,\sum_{l=1}^t 1_{V^{(l)}}(j) = m\Big\}
\end{equation}
denote the random number of iterations until variable \(X_j\) is considered \(m\) times in the search of AdaSub. Furthermore, let \begin{equation} S^*_{\text{PF}}=\{j\in S^*: j\in f_C(V) \text{ for all } V\subseteq\mathcal{P} \text{ with } j\in V \}
\end{equation}
denote the set of variables in the \(C\)-optimal model \(S^*\) for which the finite-sample PF property holds.
\begin{enumerate}
\item[(a)] Since \(b_j^{(i)} = 1_{V^{(i)}}(j) \) are independent Bernoulli distributed with success probability \(r_j^{(0)}= \frac{q}{p}\) for \(1\leq i\leq T_j^{(1)}\) (see step (2)(a) in Algorithm~\ref{AdaSub}), the number of iterations \(T_j^{(1)}\) until variable \(X_j\) with \(j\in\mathcal{P}\) is considered for the first time follows a geometric distribution with success probability \(r_j^{(0)}=\frac{q}{p}\) and expectation \(E[T_j^{(1)}] = \frac{p}{q}\), i.e.\
\begin{equation}
T_j^{(1)} = \min\Big\{t\in\N:\,b_j^{(t)} = 1\Big\} \sim \text{Geo}\left(\frac{q}{p}\right) \,.
\end{equation}

\item[(b)] Let \(j\in S^*_{\text{PF}}\). Then, by the definition of AdaSub it holds \(r_j^{(t)} = \frac{q + Ki}{p + Ki}\) for \(T_j^{(i)}  \leq t< T_j^{(i+1)}\) and \(i\in\N\).  
Thus, for \(j\in S^*_{\text{PF}}\) it holds that 
\begin{equation}
T_j^{(i+1)} - T_j^{(i)} \sim \text{Geo}\left(\frac{q + Ki}{p + Ki}\right) \,,
\end{equation}
since, for \(T_j^{(i)} < t \leq T_j^{(i+1)}\), variables \(b_j^{(t)}\) are independent Bernoulli distributed with success probability \(r_j^{(t-1)} = \frac{q + Ki}{p + Ki}\). Hence, for \(m\geq 1\), it holds 
\begin{equation}
E\left[T_j^{(m)}\right] = E\left[T_j^{(1)}\right] + \sum_{i=1}^{m-1} E\left[T_j^{(i+1)} - T_j^{(i)} \right] = \sum_{i=0}^{m-1}  \frac{p + Ki}{q + Ki} \,.
\end{equation}
Let 
\begin{equation} 
T_{\rho,j} = \min\left\{t\in\N: r_j^{(t)}> \rho \right\} 
\end{equation} 
denote the number of iterations until the thresholded model \(\hat{S}_\rho\) with threshold \(\rho\) includes variable \(X_j\). Under the assumption that \(j\in f_C(V)\) for all \(V\subseteq\mathcal{P}\) with \(j\in V\), it holds 
\begin{equation} 
T_{\rho,j}  =  T_j^{(i(\rho))} , ~\text{with}~ i(\rho) = \left\lfloor\frac{\rho p - q}{K(1-\rho)} + 1 \right\rfloor\in\N \,
\end{equation} 
and thus we derive 
\begin{equation} 
E\left[T_{\rho,j}\right] = E\left[ T_j^{(i(\rho))} \right] =  \sum_{i=0}^{i(\rho)-1}  \frac{p + Ki}{q + Ki} \,.
\end{equation} 

\item[(c)]  Suppose that the finite-sample PF property (\ref{PF_analog}) holds for the criterion \(C\), i.e.~it holds \(S^*=S^*_{\text{PF}}\).  %i.e. for \(j\in S^*=f_C(\mathcal{P})\) it holds \(j\in f_C(V)\) for all \(V\subseteq\mathcal{P}\) with \(j\in V\). 
Then the number of iterations 
%\begin{equation} 
%T_\rho = \min\left\{t\in\N: r_j^{(t)}\geq \rho \text{ for all } j\in S^*\right\}
%\end{equation} 
needed so that the thresholded model \(\hat{S}_\rho\) with threshold \(\rho\)  includes all variables in the \(C\)-optimal model \(S^*\) (with \(|S^*|=s^*\)) can be written as \(T_\rho = \max_{j\in S^*}\, T_{j}^{(i(\rho))}\). 
Thus, by using the result from (b), we obtain an upper bound on the expected number of required iterations 
\begin{equation}
E\left[ T_\rho \right] = E\left[\max_{j\in S^*}\, T_{j}^{(i(\rho))}\right] 
                       \leq \sum_{j\in S^*} E\left[ T_{j}^{(i(\rho))} \right] 
											 = s^* \sum_{i=0}^{i(\rho)-1} \frac{p + Ki}{q + Ki}  \,.
\end{equation} 
Note that this is only a crude upper bound; in fact, in simulations (see Figure~\ref{fig:speed}) it is empirically observed that the mean number of iterations \( E\left[ T_\rho \right]\) scales approximately logarithmically with the number of variables \(s^*\) in the \(C\)-optimal model \(S^*\). 

\item[(d)] Suppose that the finite-sample PF property (\ref{PF_analog}) holds for the criterion \(C\), i.e.~it holds \(S^*=S^*_{\text{PF}}\). Let  
\begin{equation} 
T_{\text{b}} = \min\left\{t\in\N: S^{(t)} = S^* \right\}
\end{equation}
denote the number of iterations needed to identify the \(C\)-optimal model. 
If there is no adaptation of the selection probabilities in the algorithm (\(K=0\)), then for all \(t\in\N\) it holds
\[P\left(S^{(t)} = S^*\right) = \left(\frac{q}{p}\right)^{s^*},  \] 
and thus
\[ T_{\text{b}} \sim \text{Geo}\left( \left(\frac{q}{p}\right)^{s^*} \right) \text{ with } E\left[ T_{\text{b}}\right] = \left(\frac{p}{q}\right)^{s^*} \,. \]
However, in the limiting case \(K\rightarrow\infty\) of infinite adaptation, i.e.\ for \(j\in S^*\) it holds \(r_j^{(t)}=\frac{q}{p}\) for \(0\leq t < T_j^{(1)}\) and \(r_j^{(t)}=1\) for \(t\geq T_j^{(1)}\), we obtain \(T_{\text{b}} = \max_{j\in S^*}  T_j^{(1)}\) with expectation
\begin{equation} \label{eq:approxE}
 E\left[T_{\text{b}}\right] = E\left[\max_{j\in S^*} \,  T_j^{(1)} \right] 
             \approx \frac{1}{2} + \frac{1}{\log\left(\frac{p}{p-q}\right)} \sum_{i=1}^{s^*} \frac{1}{i} \,.
\end{equation}
Here we have used a result by \cite{eisenberg2008} regarding the expectation of the maximum of independent geometrically distributed variables \(T_j^{(1)}\sim\text{Geo}(q/p)\), \(j\in S^*\), with approximation error of the expectation in (\ref{eq:approxE}) bounded by \(\frac{1}{2}\).  Thus, under the finite-sample PF assumption with \(K\rightarrow\infty\), the expected number of iterations \(E\left[T_{\text{b}}\right]\) grows logarithmically with \(s^*\) (as the harmonic series diverges logarithmically). Furthermore, with a Taylor expansion the term \(1/\log\left(\frac{p}{p-q}\right)\) in (\ref{eq:approxE}) can be approximated by \(\frac{p-q}{q}\), showing that the expectation \(E\left[T_{\text{b}}\right]\) grows approximately linearly with~\(p\). \
\end{enumerate}
%\(T_l^{(j)} = \sum_{i=1}^m T_i^{(j)}\) with expectation \(E[T_m^{(j)}] = \sum_{i=1}^m E[T_i^{(j)}] = \sum_{i=1}^m \frac{p + K(i-1)}{q + K(i-1)}\). 
\end{remark}

\clearpage

\subsection{Illustrative example of AdaSub} \label{sec:illustrative}

In order to illustrate the performance of AdaSub in a high-dimensional set-up, we consider a simulated example with \(p=1000\) and \(n=60\). 
 We generate one particular dataset \(\mathcal{D}=(\bs X,\bs Y)\) by simulating \( \bs X=(X_{ij})\in\mathbb{R}^{n\times p}\) with independent rows \(\bs X_{i,*}\sim\mathcal{N}_p(0,\bs \Sigma)\), where \(\Sigma_{kl}=0\) for \(k\neq l\) and \(\Sigma_{kk}=1\). Furthermore, let \[\bs \beta^0=(0.4,0.8,1.2,1.6,2.0,0,\dots,0)^T\in\mathbb{R}^p\] be the true vector of regression coefficients with active set \(S_0=\{1,\dots,5\}\). 
 The response \(\bs Y=(Y_1,\dots,Y_n)^T\) is simulated via \(Y_i\overset{\text{ind.}}{\sim} N(\bs X_{i,*} \bs \beta^0, 1)\), \(i=1,\dots,n\). We adopt the (negative) extended BIC (\(\text{EBIC}_\gamma\)) as the criterion \(C\) and consider the tuning parameter choices \(\gamma=0.6\) and \(\gamma=1\) in \(\text{EBIC}_\gamma\). For both cases, we apply AdaSub with \(T=10,000\) iterations on the same dataset simulated as above and choose \(q=10\) and \(K=n\) as the tuning parameters of AdaSub. 

\begin{figure}[ht]
\centering
\subfloat[\(\gamma=0.6\)]{\includegraphics[width=.45\linewidth]{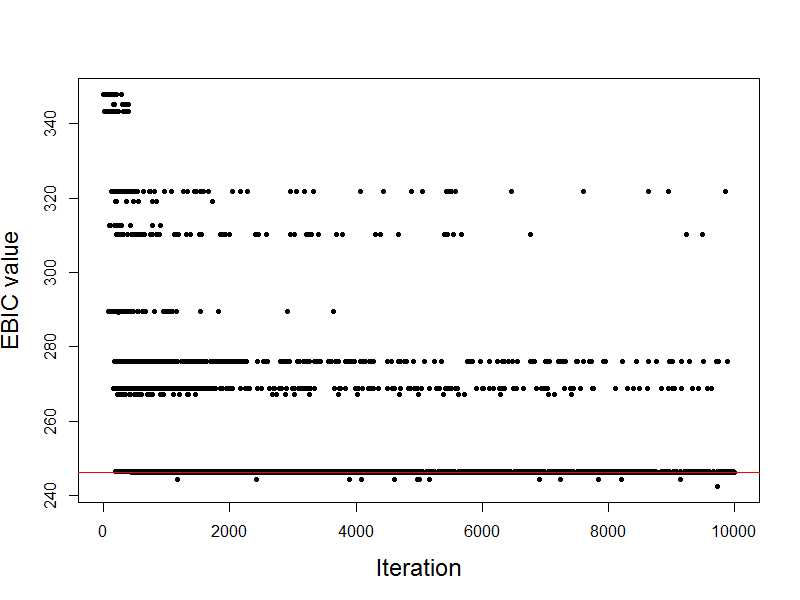}
 \label{Evolution_Ex_06}}\qquad
\subfloat[\(\gamma=1\)]{\includegraphics[width=.45\linewidth]{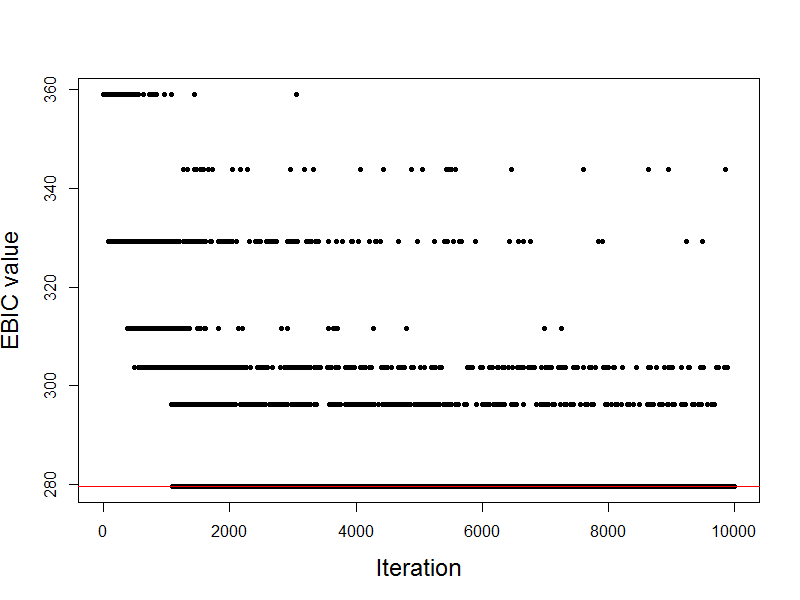}
 \label{Evolution_Ex_1}}
\caption{AdaSub for the high-dimensional simulated example. Plots of the evolution of \(\text{EBIC}_\gamma(S^{(t)})\) along the iterations \(t\) for (a) \(\gamma=0.6\) and (b) \(\gamma=1\). The red lines indicate the \(\text{EBIC}_{\gamma}\)-values of the thresholded model \(\hat{S}_{0.9}\).}
\label{fig:EBICvalues}
\end{figure}

We present some typical ``diagnostic plots'' for the described simulated data example, which are generally very helpful for examining the convergence of the AdaSub algorithm. Figure~\ref{fig:EBICvalues} shows the evolution of the \(\text{EBIC}_\gamma(S^{(t)})\)-values along the iterations \(t\) for \(\gamma=0.6\) and \(\gamma=1\) (recall that \(S^{(t)}=f_C(V^{(t)})\) denotes the ``best'' submodel contained in \(V^{(t)}\)), while the red lines indicate the values of \(\text{EBIC}_{\gamma}\) for the thresholded model \(\hat{S}_{0.9}\). For \(\gamma=0.6\) it is obvious that the algorithm does not converge against the ``best'' sampled model \(\hat{S}_{\textrm{b}}=\arg\min\{\text{EBIC}_{0.6}(S^{(1)}),\dots,\text{EBIC}_{0.6}(S^{(T)})\}\) and thus OIP' does not hold here. The ``best'' model identified by AdaSub is given by \(\hat{S}_{\textrm{b}}=\{2, 3 , 4 ,  5 ,519 ,731 ,950\}\), while the thresholded model \(\hat{S}_{0.9}=\{2, 3 , 4 , 5 , 950 \}\) with threshold \(\rho=0.9\) does not include the ``noise variables'' \(X_{519}\) and \(X_{731}\) and is therefore closer to the true underlying model. This is an example, where the thresholded model from AdaSub reduces the number of false positives in a situation where the criterion used is too liberal (compare Corollary \ref{OIP_replacement}). On the other hand, for \(\gamma=1\), the algorithm appears to have converged against the \(\text{EBIC}_{0.6}\)-optimal model; the ``best'' sampled model \(\hat{S}_{\textrm{b}}\) and the thresholded model \(\hat{S}_{0.9}\) agree: \(\hat{S}_{\textrm{b}} = \hat{S}_{0.9} = \{2,3,4,5\}\). This indicates, that the model identified by AdaSub is ``stable'' in the sense of OIP.

\begin{figure}[ht]
\centering
\subfloat[\(\gamma=0.6\)]{\includegraphics[width=.47\linewidth]{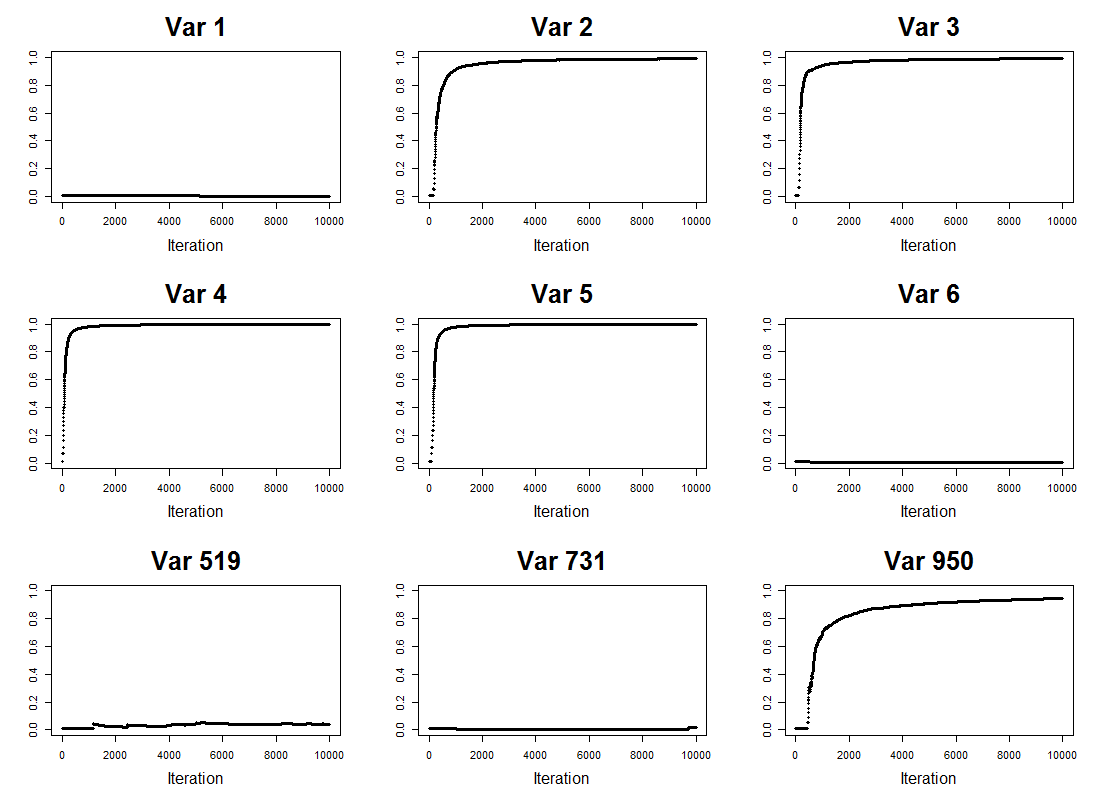}} \qquad
\subfloat[\(\gamma=1\)]{\includegraphics[width=.47\linewidth]{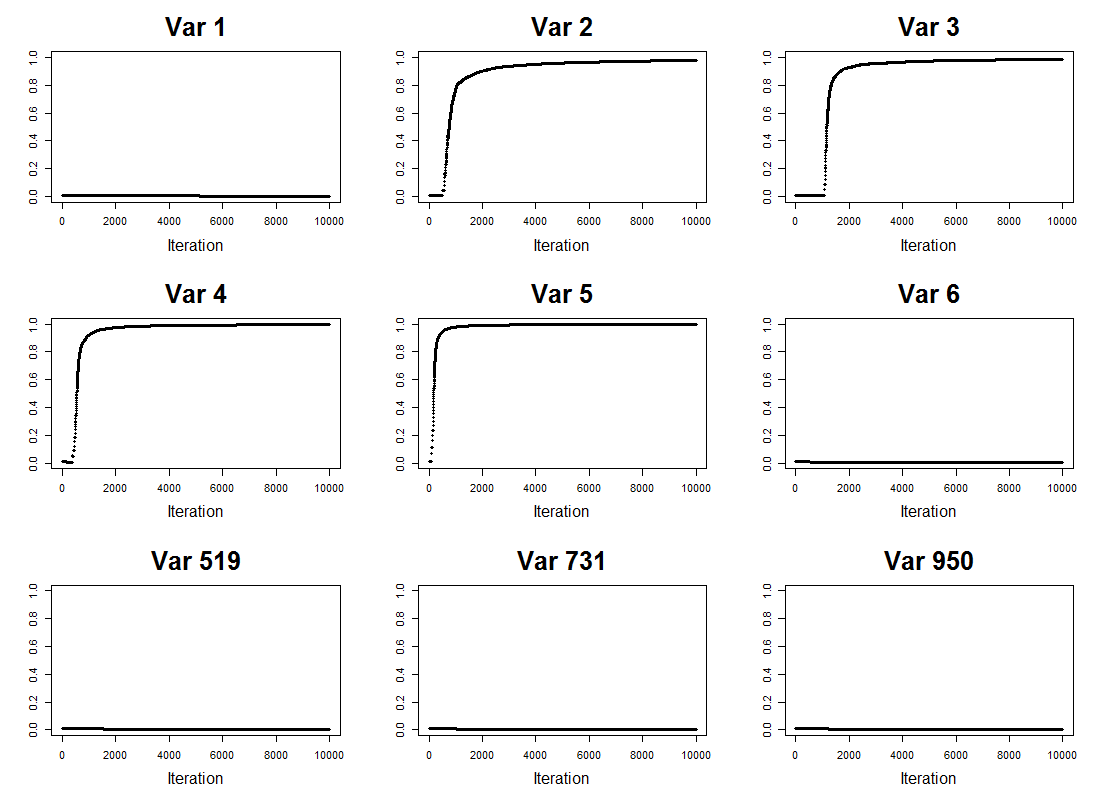}}
\caption{AdaSub for the high-dimensional simulated example. Plots of the evolution of \(r_j^{(t)}\) (with \(j\in\{1,\dots,6,519,731,950\}\)) along the iterations \(t\) for (a) \(\gamma=0.6\) and (b) \(\gamma=1\).}
\label{fig:SelectionProb}
\end{figure}

\begin{figure}[ht]
\centering
\subfloat[\(\gamma=0.6\)]{\includegraphics[width=.45\linewidth]{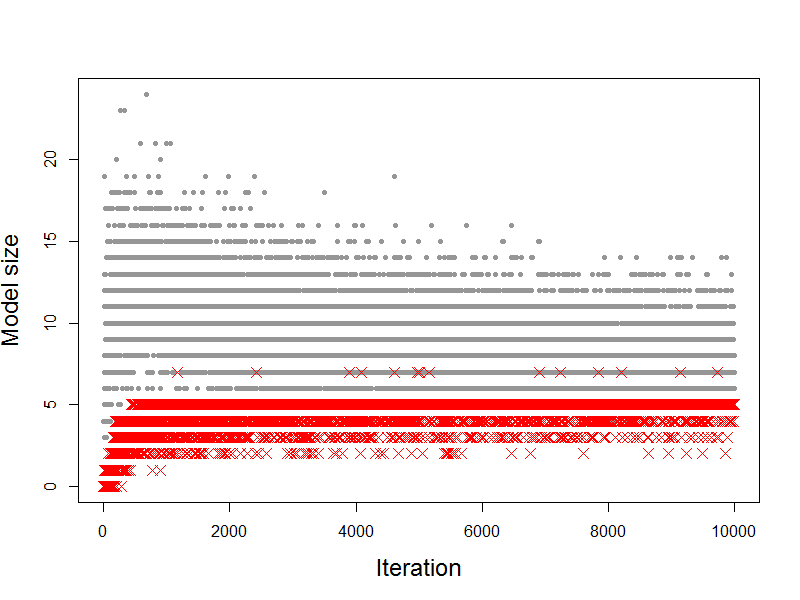}} \qquad
\subfloat[\(\gamma=1\)]{\includegraphics[width=.45\linewidth]{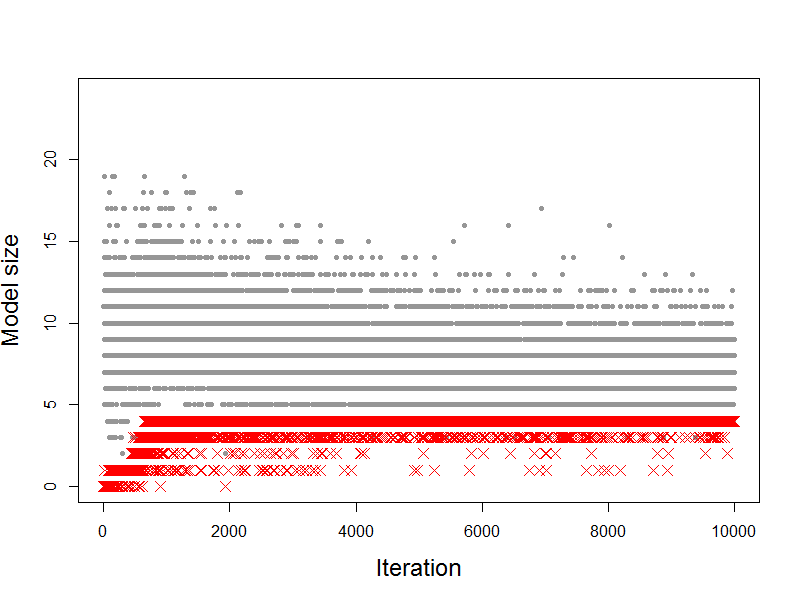}}
\caption{AdaSub for the high-dimensional simulated example. Plots of the evolution of the sizes of the sampled sets \(V^{(t)}\) (grey dots) and the sizes of the selected subsets \(f_C(V^{(t)})=S^{(t)}\) (red crosses) along the iterations \(t\) for (a) \(\gamma=0.6\) and (b) \(\gamma=1\).}
\label{fig:Sizes}
\end{figure}

Figure~\ref{fig:SelectionProb} shows the evolution of some of the selection probabilities \(r_j^{(t)}\) along the iterations \(t\) for \(\gamma=0.6\) and \(\gamma=1\). In both cases, the selection probabilities \(r_j^{(t)}\) for \(j\in\{2,3,4,5\}\) quickly approach the value of one while \(r_6^{(t)}\) tends to zero. On the other hand, \(r_1^{(t)}\) tends to zero and hence the ``signal variable'' \(X_1\) is not selected in both cases (note that \(\beta_1=0.4\) is quite small). Additionally, the evolution of the selection probabilities \(r_j^{(t)}\) for \(j\in\{519,731,950\}\) is shown. While for \(\gamma=1\) these selection probabilities all tend to zero as desired, the behaviour is different for \(\gamma=0.6\): \(r_{950}^{(t)}\) tends to one; \(r_{519}^{(t)}\) and \(r_{731}^{(t)}\)  seem to converge to values close but not exactly zero. This reflects a situation, where OIP does not hold and variables \(X_{519}\) and \(X_{731}\) are not ``stable'' in the sense of OIP.

\begin{figure}[ht]
\centering
\subfloat[\(\gamma=0.6\)]{\includegraphics[width=.45\linewidth]{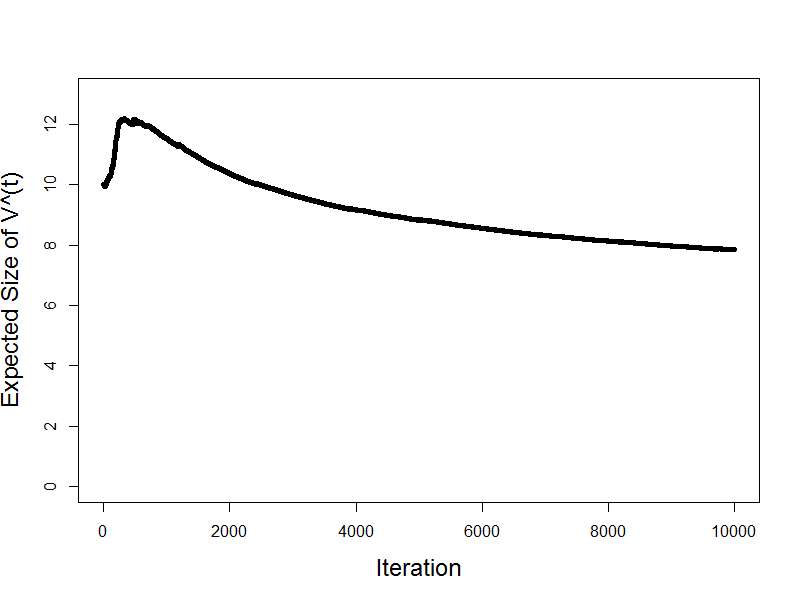}} \qquad
\subfloat[\(\gamma=1\)]{\includegraphics[width=.45\linewidth]{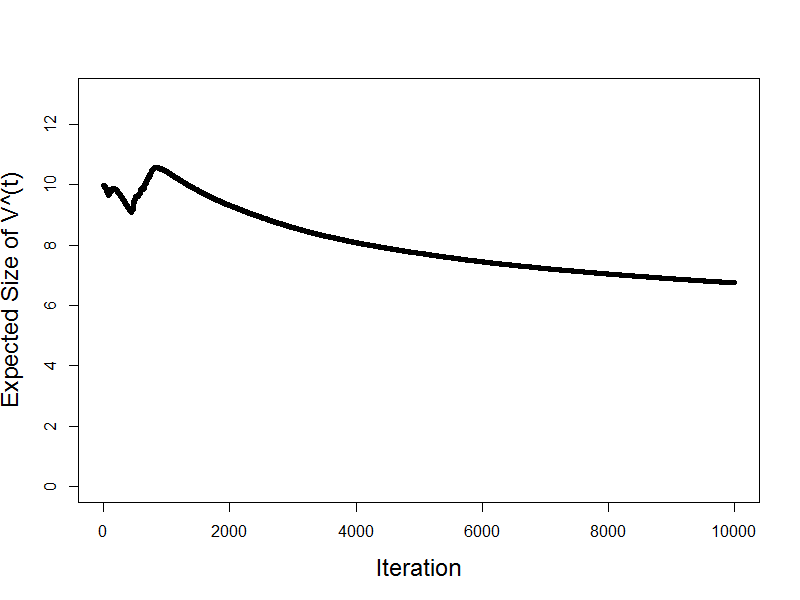}}
\caption{AdaSub for the high-dimensional simulated example. Plots of the evolution of the expected search size \(E|V^{(t)}|\) along the iterations \(t\) for (a) \(\gamma=0.6\) and (b) \(\gamma=1\).}
\label{fig:ExpectedSizes}
\end{figure}

Figure~\ref{fig:Sizes} shows the evolution of the sizes of the sampled sets \(V^{(t)}\) and the sizes of the selected subsets \(S^{(t)}\) along the iterations \(t\); additionally, Figure~\ref{fig:ExpectedSizes} depicts the evolution of the expected search size \(E|V^{(t)}|=\sum_{j\in\mathcal{P}}r_j^{(t-1)}\) along the iterations \(t\) for \(\gamma=0.6\) and \(\gamma=1\). Starting with initial expected search size \(E|V^{(1)}|=q=10\), the AdaSub algorithm automatically adjusts the expected search sizes which, after some time, start to decrease with the number of iterations. For \(\gamma=0.6\), the search sizes are a bit larger, since the criterion \(\text{EBIC}_{0.6}\) enforces less sparsity than \(\text{EBIC}_{1}\). The computation times for \(T=10,000\) iterations of AdaSub were approximately 15.1 seconds for \(\gamma=0.6\) and 13.5 seconds for~\(\gamma=1\).

\clearpage

\subsection{Additional results of simulation study} \label{sec:additional}
We present further results of the simulation study given in Section \ref{sec:simstudy}.  
The low- and high-dimensional simulation set-ups are as described in Section \ref{sec:simstudy}. In particular, the design matrix \(\bs X=(X_{i,j})\in\mathbb{R}^{n\times p}\) is simulated via \(\bs X_{i,*}\sim\mathcal{N}_p(\bs 0,\bs \Sigma)\). Here, we consider the following correlation structures between the explanatory variables induced by the matrix \(\bs \Sigma\in\R^{p \times p}\):
\begin{enumerate}
%\vspace{-1mm}
\item[(a)] \textbf{Toeplitz-Correlation Structure:} For some \(c\in(-1,1)\) let \(\Sigma_{k,l}=c^{|k-l|}\) for all \(k\neq l\).
%\vspace{-1mm}
\item[(b)] \textbf{Equal-Correlation Structure:} For some \(c\in[0,1)\) let \(\Sigma_{k,l}=c\) for all \(k\neq l\). 
 %
%\vspace{-1mm}
\item[(c)] \textbf{Block-Correlation Structure:} For some \(c\in(0,1)\) and a fixed number of blocks \(b\in\N\) let \(\Sigma_{k,l}=c\) for all \(k\neq l\) with \((k-l) ~\text{mod}~ b = 0\), and let \(\Sigma_{k,l}=0\) otherwise. 
\end{enumerate}
%\vspace{-1mm}

\begin{figure}[!ht]
\centering
\includegraphics[width=14cm]{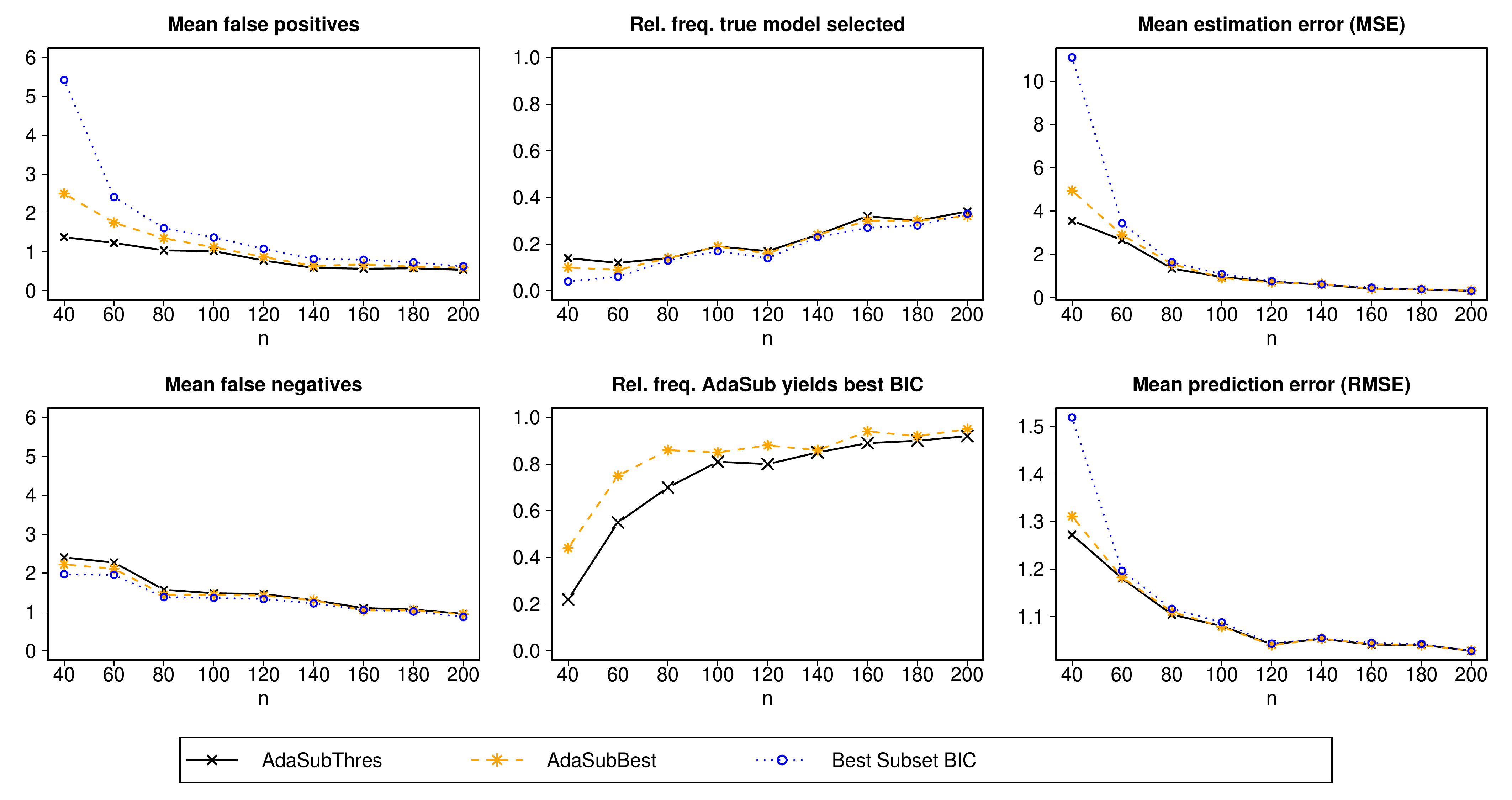}
\caption{\label{LOW_Toeplitz} Results for low-dimensional setting (\(p=30\)) with Toeplitz-correlation structure (\(c=0.9\)): Comparison of thresholded model \(\hat{S}_{0.9}\) (AdaSubThres) and ``best'' model \(\hat{S}_{\textrm{b}}\) (AdaSubBest) from AdaSub with BIC-optimal model \(S^*\) (Best Subset BIC) in terms of mean number of false positives/ false negatives, relative frequency of selecting the true model \(S_0\), relative frequency of agreement between AdaSub models and \(S^*\), Mean Squared Error (MSE) and Root Mean Squared Prediction Error (RMSE) on independent test set.} 
\end{figure}

\begin{figure}[!ht]
\centering
\subfloat[Equal-correlation structure (\(c=0.7\))]{\includegraphics[width=14cm]{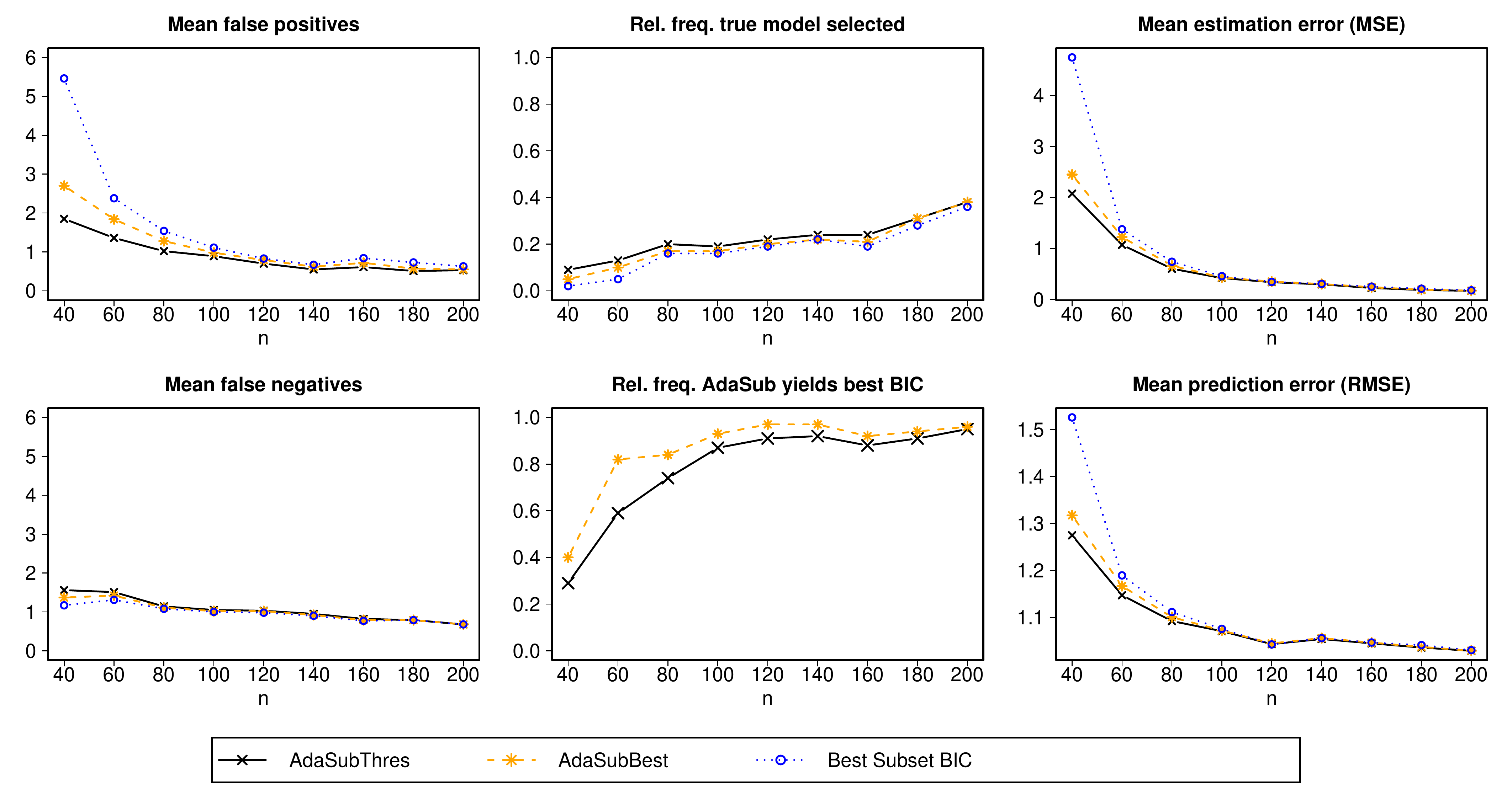}}  \\[2mm]
\subfloat[Block-correlation structure (\(b=10\) blocks and \(c=0.5\))]{\includegraphics[width=14cm]{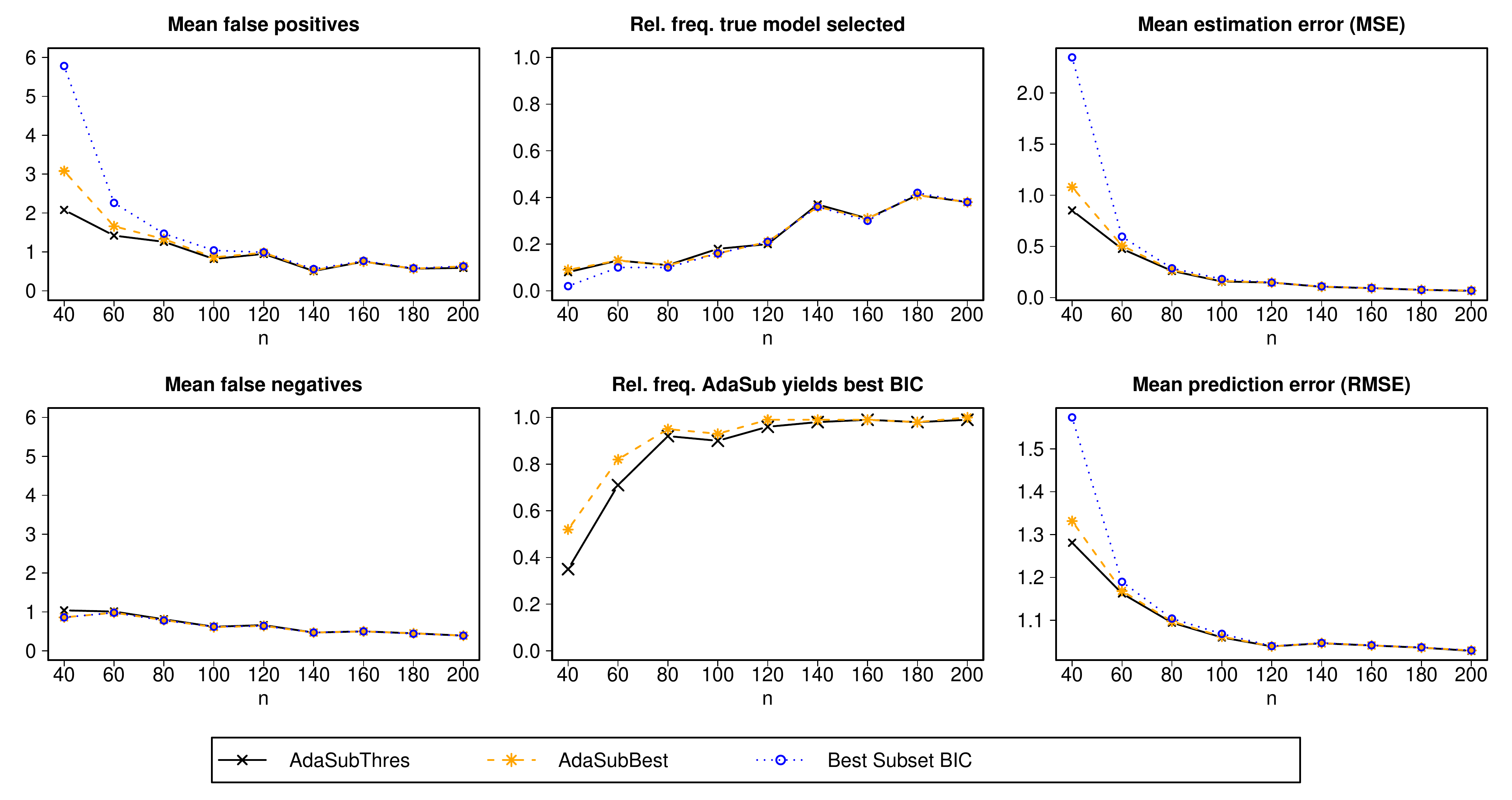}}
\caption{Results for low-dimensional setting (\(p=30\)) with (a) equal-correlation structure and (b) block-correlation structure: Comparison of thresholded model \(\hat{S}_{0.9}\) (AdaSubThres) and ``best'' model \(\hat{S}_{\textrm{b}}\) (AdaSubBest) from AdaSub with BIC-optimal model \(S^*\) (Best Subset BIC) in terms of mean number of false positives/ false negatives, relative frequency of selecting the true model \(S_0\), relative frequency of agreement between AdaSub models and \(S^*\), Mean Squared Error (MSE) and Root Mean Squared Prediction Error (RMSE) on independent test set.}
\label{fig:lowADD}
\end{figure}

\begin{figure}[!ht]
\centering
\subfloat[Results for \(\text{EBIC}_\gamma\) with \(\gamma = 0.6\).]{\includegraphics[width=14cm]{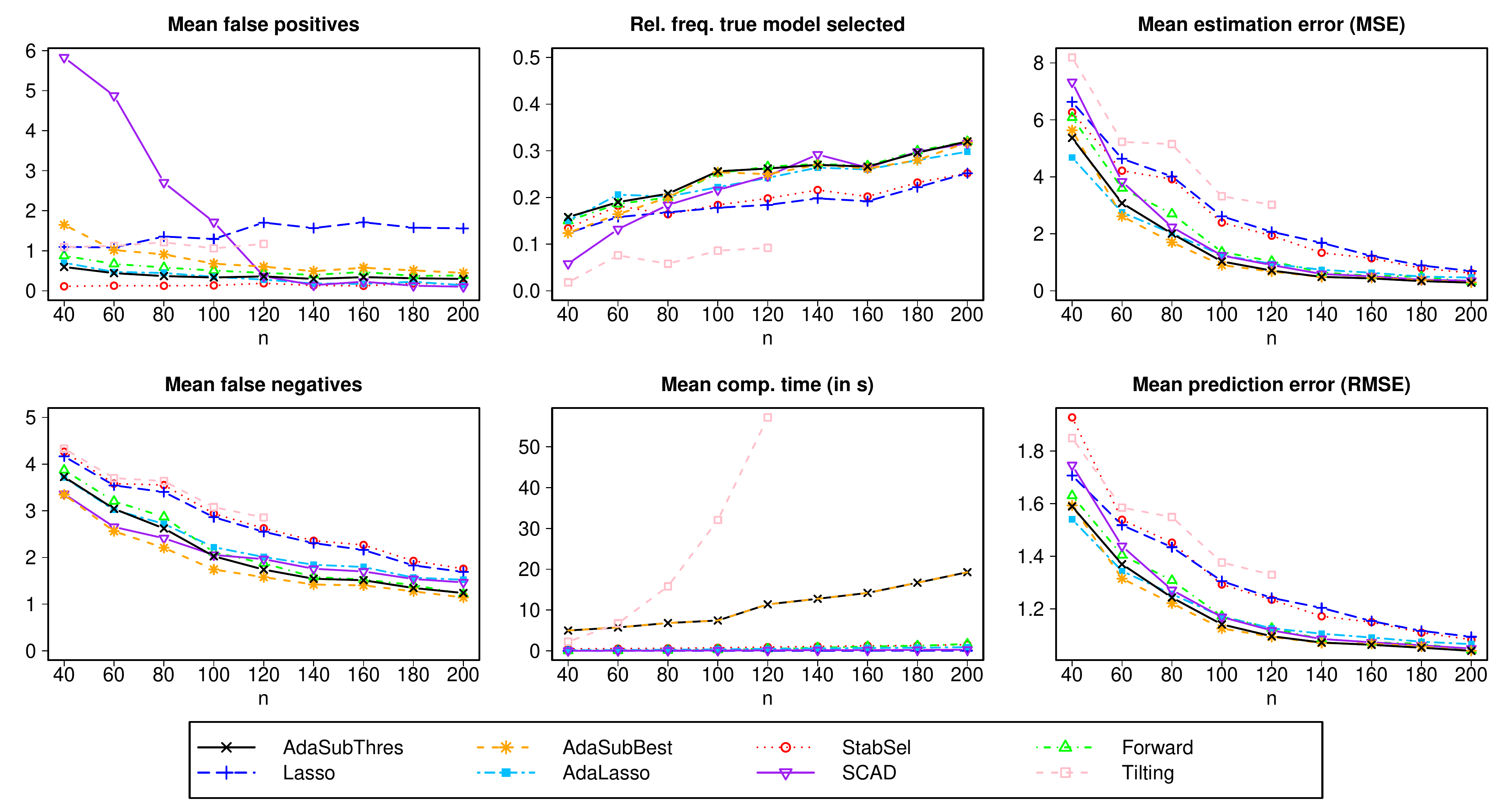}}  \\[2mm]
\subfloat[Results for \(\text{EBIC}_\gamma\) with \(\gamma = 1\).]{\includegraphics[width=14cm]{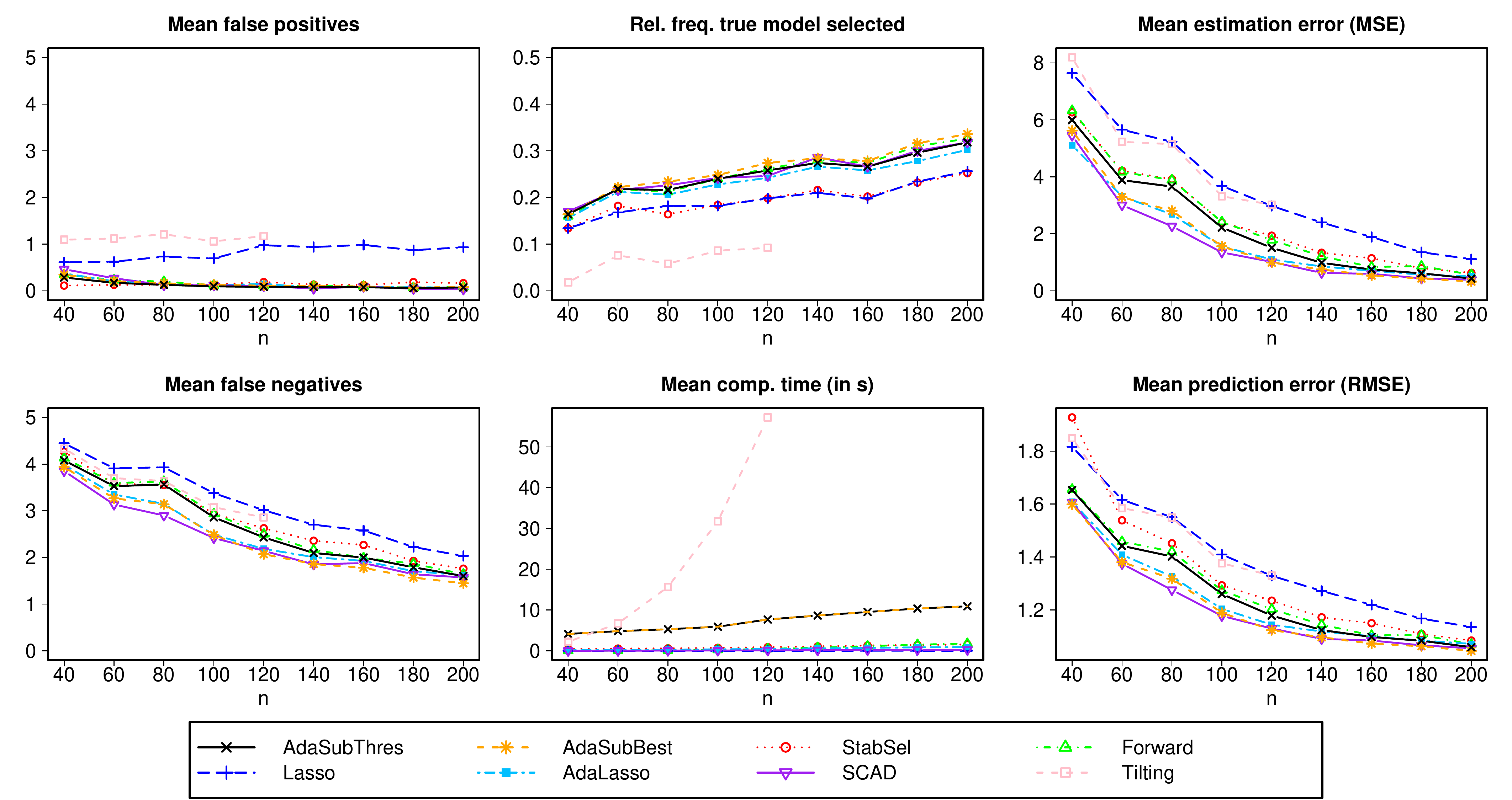}}
\caption{\label{fig:highGlobal07}Results for high-dimensional setting  (\(p=10n\)) with equal-correlation structure (\(c=0.7\)): Comparison of thresholded model (AdaSubThres) and ``best'' model (AdaSubBest) from AdaSub with Stability Selection (StabSel), Forward Stepwise, Lasso, Adaptive Lasso (AdaLasso), SCAD and Tilting in terms of mean number of false positives/ false negatives, rel. freq. of selecting the true model, mean comp. time, MSE and RMSE.} 
\end{figure}

\begin{figure}[!ht]
\centering
\subfloat[Results for \(\text{EBIC}_\gamma\) with \(\gamma = 0.6\).]{\includegraphics[width=14cm]{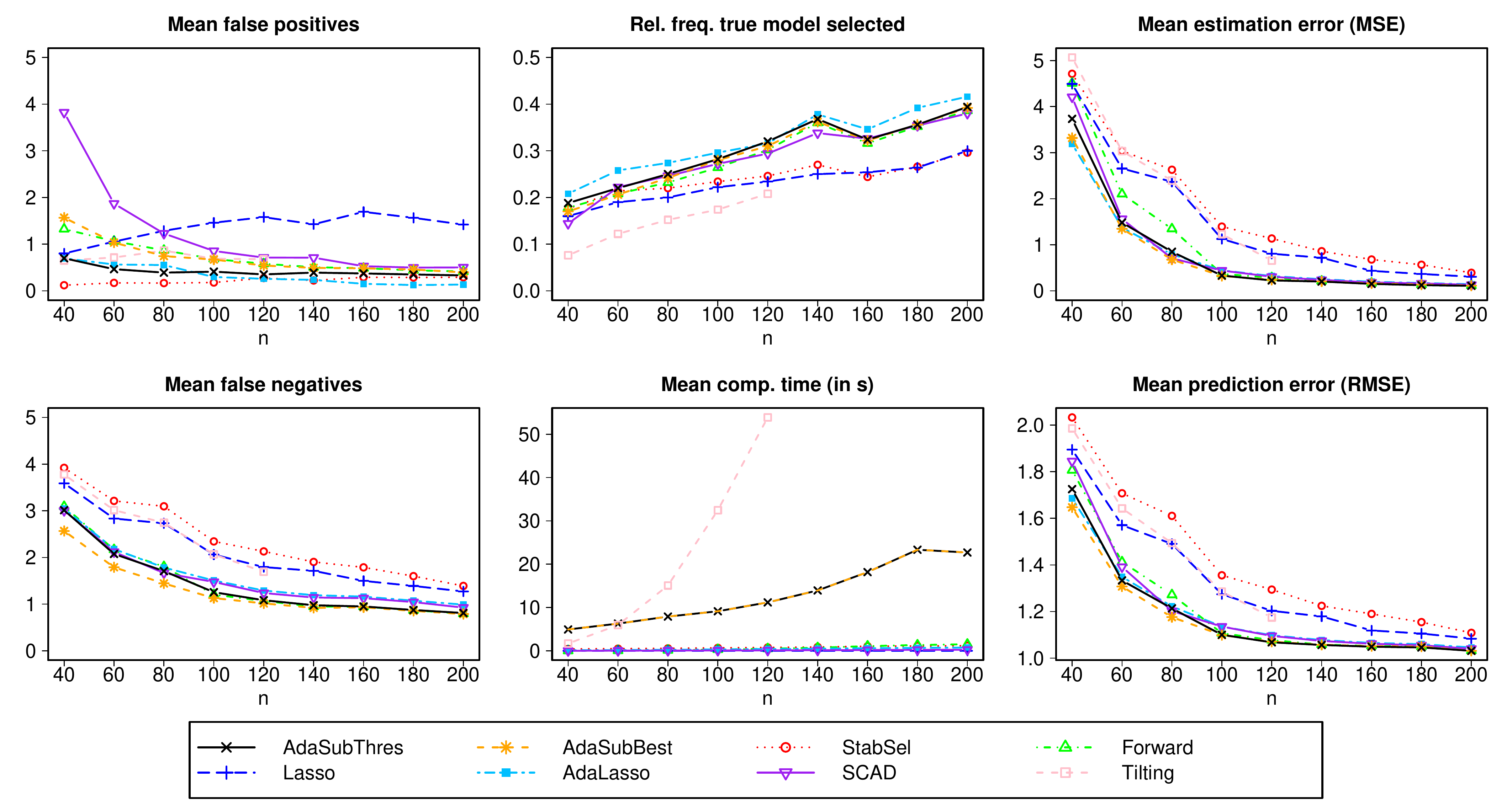}}  \\[2mm]
\subfloat[Results for \(\text{EBIC}_\gamma\) with \(\gamma = 1\).]{\includegraphics[width=14cm]{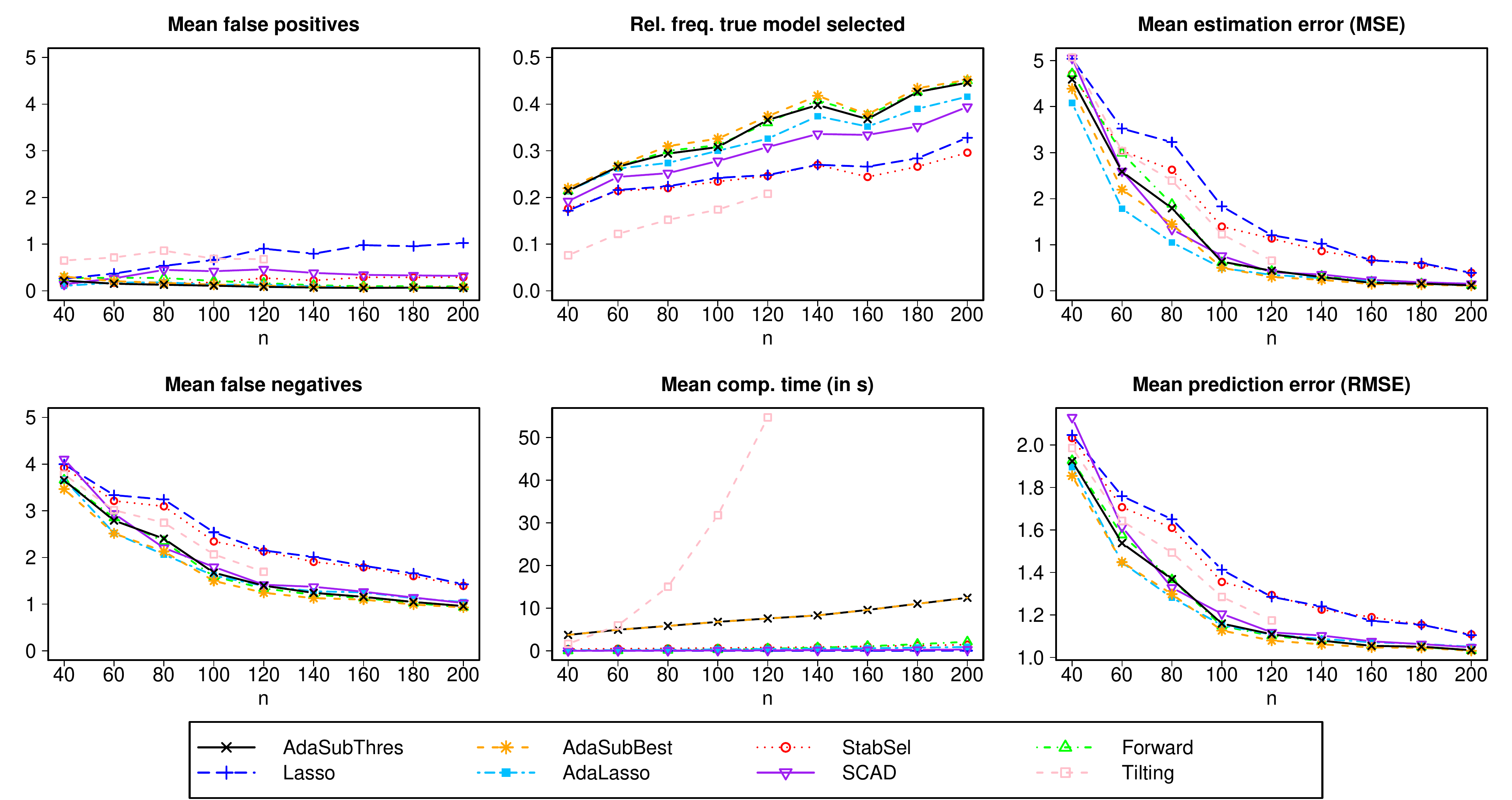}}
\caption{\label{fig:highBlocks05}Results for high-dimensional setting (\(p=10n\)) with block-correlation structure (\(b=10\) blocks and \(c=0.5\)): Comparison of thresholded model (AdaSubThres) and ``best'' model (AdaSubBest) from AdaSub with Stability Selection (StabSel), Forward Stepwise, Lasso, Adaptive Lasso (AdaLasso), SCAD and Tilting in terms of mean number of false positives/ false negatives, rel. freq. of selecting the true model, mean comp. time, MSE and RMSE.} 
\end{figure}

Figure~\ref{LOW_Toeplitz} depicts the results in a low-dimensional situation (\(p=30\)) with large correlations between the explanatory variables (Toeplitz-correlation structure with \(c=0.9\)). The relative frequency of agreement between the models selected by AdaSub and the BIC-optimal model increases towards one when the sample size increases, but the ``convergence'' is markedly slower than in the independent case (see Figure~\ref{LOW_Global00}). This shows that the models from AdaSub may yield different (and in the given setting preferable) results in comparison to the BIC-optimal model even if the sample size is moderately large. 

Next, we consider an equal-correlation structure (correlation \(c=0.7\)) and a block-correlation structure (\(b=10\) blocks and \(c=0.5\) as the correlation within blocks). Figure~\ref{fig:lowADD} shows the results of the low-dimensional examples, while Figures~\ref{fig:highGlobal07} and \ref{fig:highBlocks05} depict the results of the high-dimensional examples. In the low-dimensional examples the observations are very similar to the other situations described; the high-dimensional examples further demonstrate, that the performance of AdaSub is very competitive in comparison to the other methods considered. 

\begin{figure}[!ht]
\centering
\includegraphics[width=0.8\linewidth]{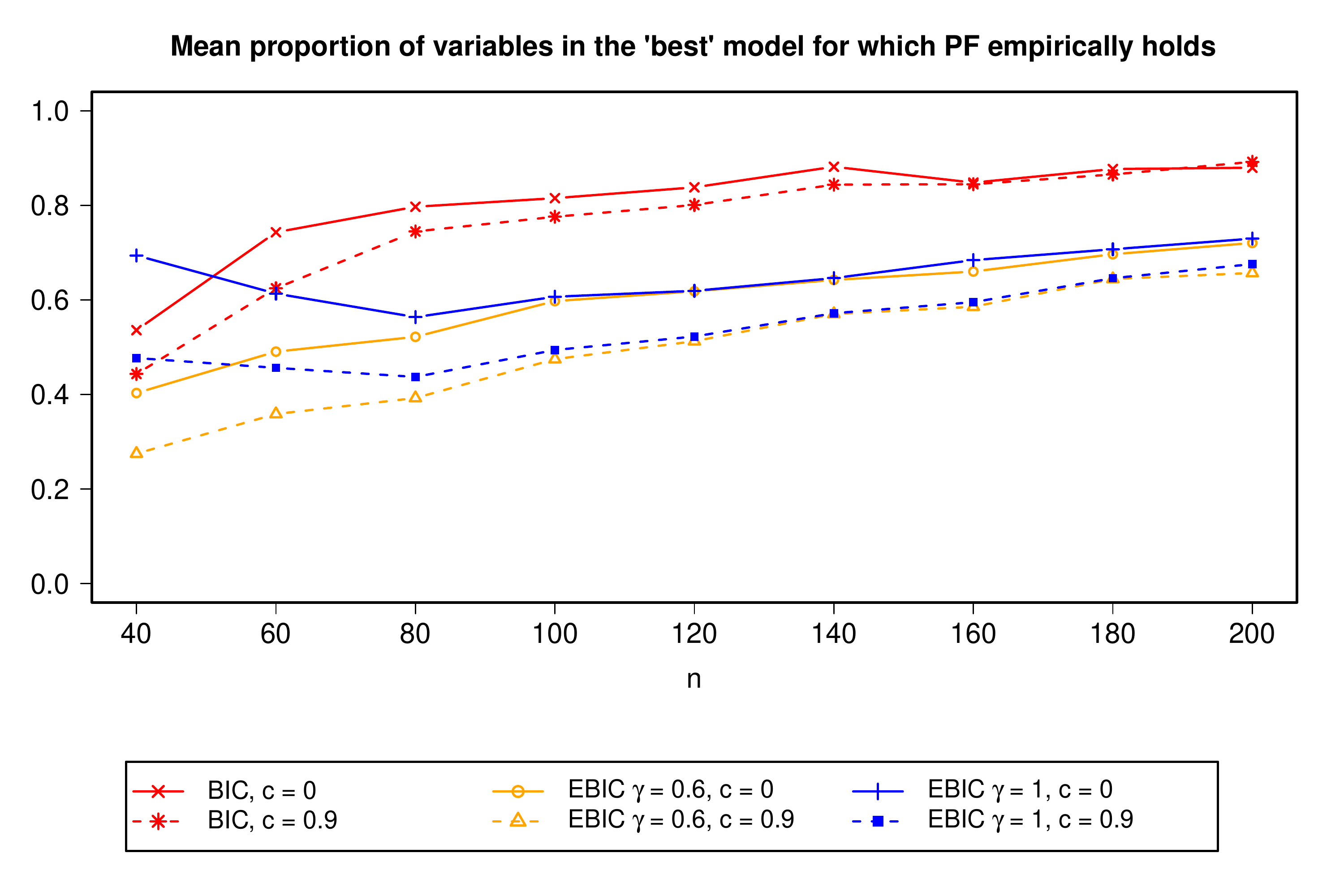}
\caption{\label{fig:PF} Additional results for low-dimensional setting (\(p=30\)) with BIC and high-dimensional setting (\(p=10n\)) with \(\text{EBIC}_\gamma\), \(\gamma\in\{0.6,1\}\), for independence correlation structure (\(c=0\)) and Toeplitz correlation structure (\(c=0.9\)): Mean proportion of variables inside the ``best'' model \(\hat{S}^*\) for which the finite-sample PF property (\ref{PF_analog}) empirically holds (excluding cases with \(\hat{S}^*=\emptyset\)), where the ``best'' model \(\hat{S}^*\) refers to \(S^*\) in the low-dimensional setting and to the best model \(\hat{S}_{\text{b}}\) identified by AdaSub in the high-dimensional setting. Here, PF empirically holds for variable \(X_j\) with \(j\in \hat{S}^*\), if \(j\in V^{(t)}\) implies \(j\in f_C(V^{(t)})\) for all \(t=1,\dots,T\).} 
\end{figure}

Figure~\ref{fig:PF} depicts additional results regarding the finite-sample PF property (\ref{PF_analog}) for the low-dimensional setting (\(p=30\)) with the BIC and the high-dimensional setting (\(p=10n\)) with the \(\text{EBIC}_\gamma\), \(\gamma\in\{0.6,1\}\), considering an independence correlation structure (\(c=0\)) and a Toeplitz correlation structure with large correlations (\(c=0.9\)). Note that for a large number of variables \(p\) it is computationally prohibitive to  compute the \(C\)-optimal model \(S^*\) and to assess whether the finite-sample PF property holds for the criterion \(C\), as one would have to compute the best sub-models \(f_C(V)\) for all possible subspaces \(V\subseteq\mathcal{P}\). Thus, here we check whether the finite-sample PF property is satisfied for all subspaces \(V^{(t)}\) sampled by AdaSub after \(T=5000\) iterations. In particular, we say that PF empirically holds for variable \(X_j\) with \(j\in \hat{S}^*\), if \(j\in V^{(t)}\) implies \(j\in f_C(V^{(t)})\) for all \(t=1,\dots,T\), where the ``best'' model \(\hat{S}^*\) refers to the actual \(C\)-optimal model \(S^*\) in the low-dimensional setting and to the best model \(\hat{S}_{\text{b}}\) identified by AdaSub in the high-dimensional setting (as an estimate for \(S^*\)). 
 
It can be observed that the mean proportion of variables in the ``best'' model \(\hat{S}^*\), for which the PF property (\ref{PF_analog}) empirically holds, tends to increase with the sample size \(n\) for both the low- and high-dimensional settings. This indicates that a faster convergence of AdaSub can be achieved for larger values of \(n\), as more ``important'' covariates in \(\hat{S}^*\) are always selected to be in the best sub-model when considered in the model search of AdaSub. Figure~\ref{fig:PF} further shows that the finite-sample PF property is less likely to hold in case of large correlations between the covariates (Toeplitz structure with \(c=0.9\)) in comparison to the case of independent covariates (\(c=0\)).

\begin{figure}[!ht]
\subfloat[Independence correlation structure (\(c=0\))]{\includegraphics[width=\linewidth]{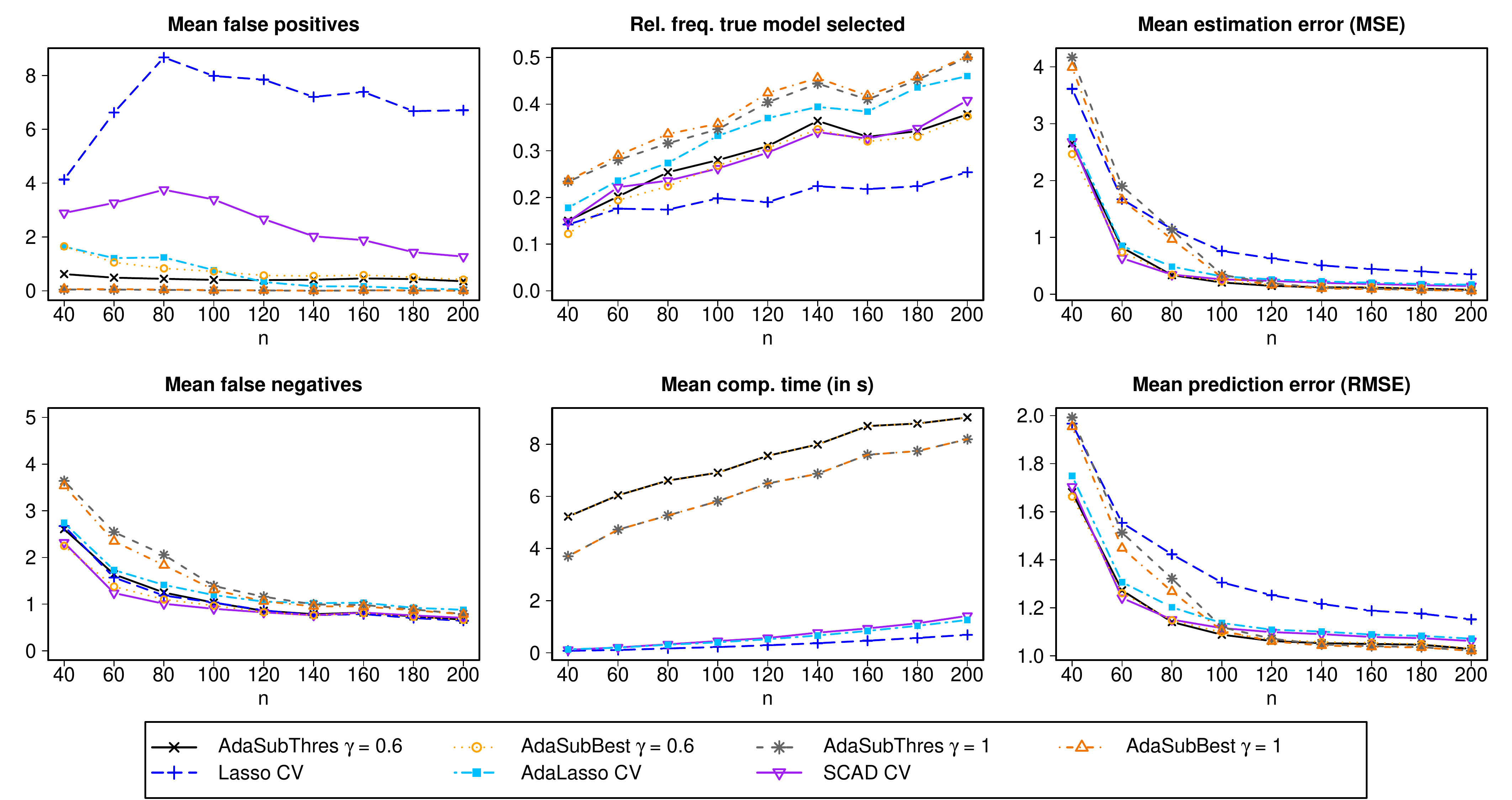}}  \\[2mm]
\subfloat[Toeplitz correlation structure (\(c=0.9\))]{\includegraphics[width=\linewidth]{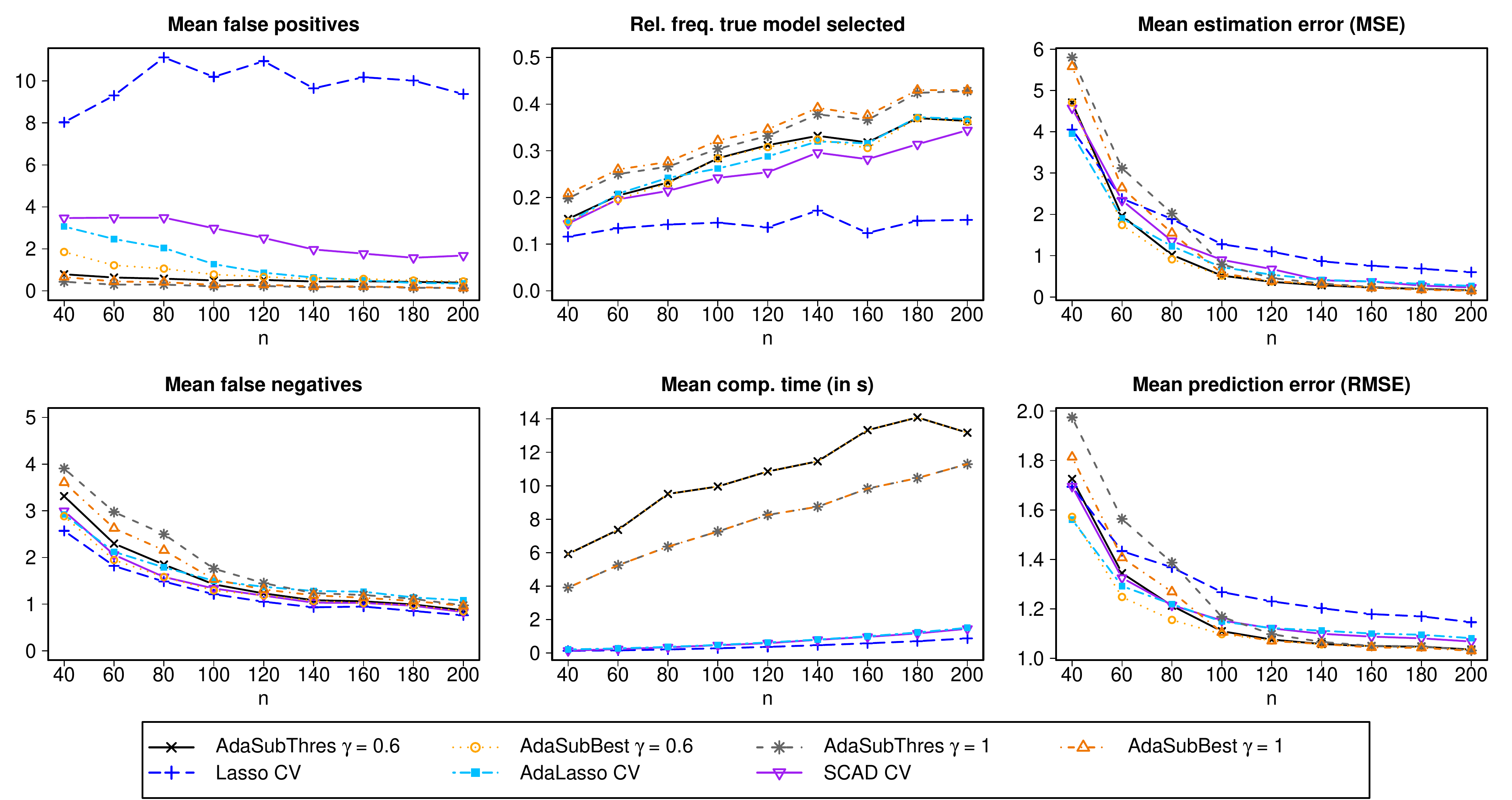}}
\caption{\label{fig:CV1}Results for high-dimensional setting (\(p=10n\)) with (a) independent covariates and (b) Toeplitz correlation structure: Comparison of thresholded model (AdaSubThres) and ``best'' model (AdaSubBest) from AdaSub (for \(\text{EBIC}_\gamma\) with \(\gamma \in\{0.6,1\}\)) with Lasso, Adaptive Lasso (AdaLasso) and SCAD tuned with ten-fold cross-validation (CV) in terms of mean number of false positives/ false negatives, rel. freq. of selecting the true model, mean comp. time, MSE and RMSE.} 
\end{figure}

\begin{figure}[!ht]
\subfloat[Equal-correlation structure (\(c=0.7\))]{\includegraphics[width=\linewidth]{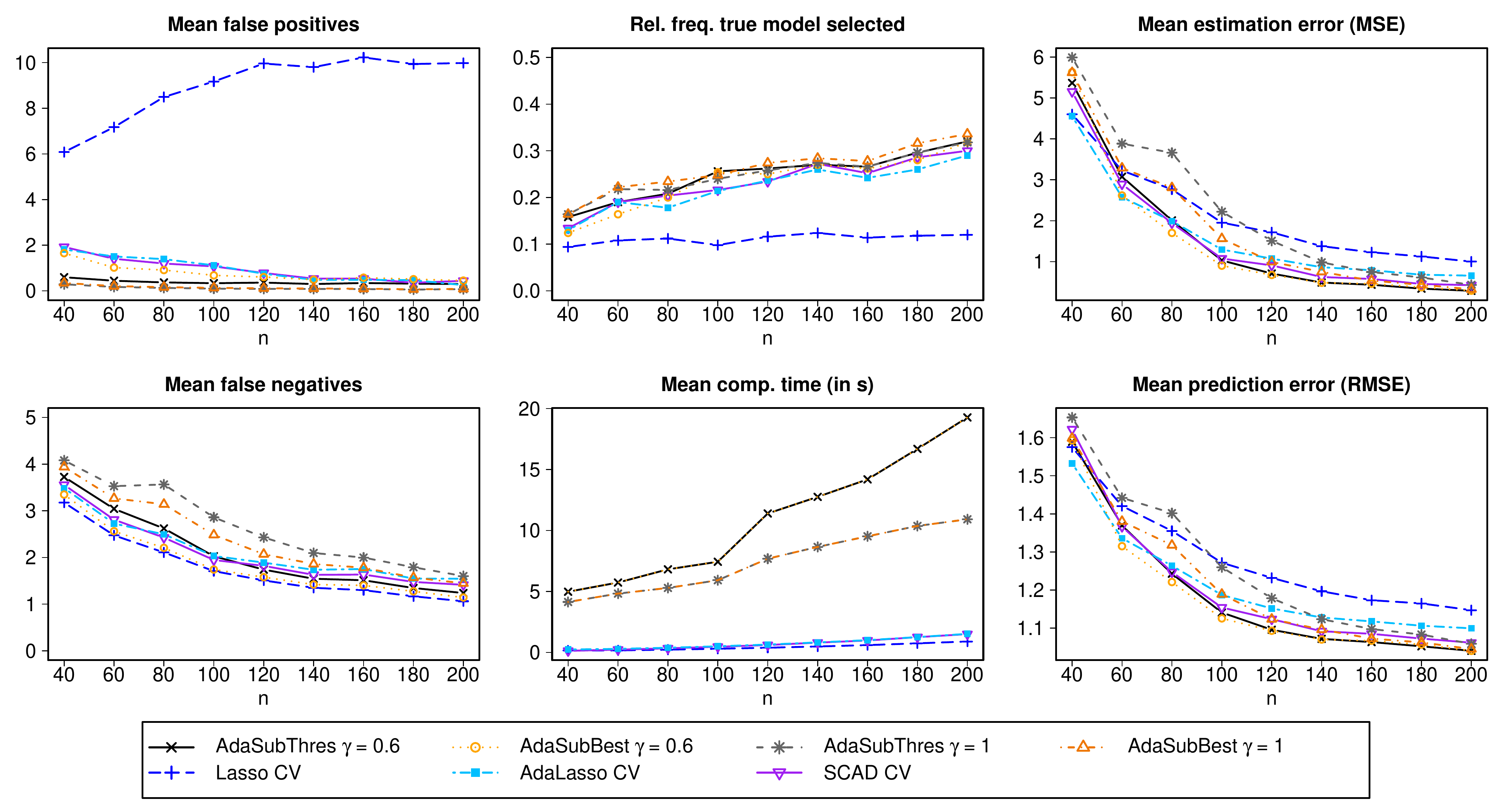}}  \\[2mm]
\subfloat[Block-correlation structure (\(b=10\) blocks and \(c=0.5\))]{\includegraphics[width=\linewidth]{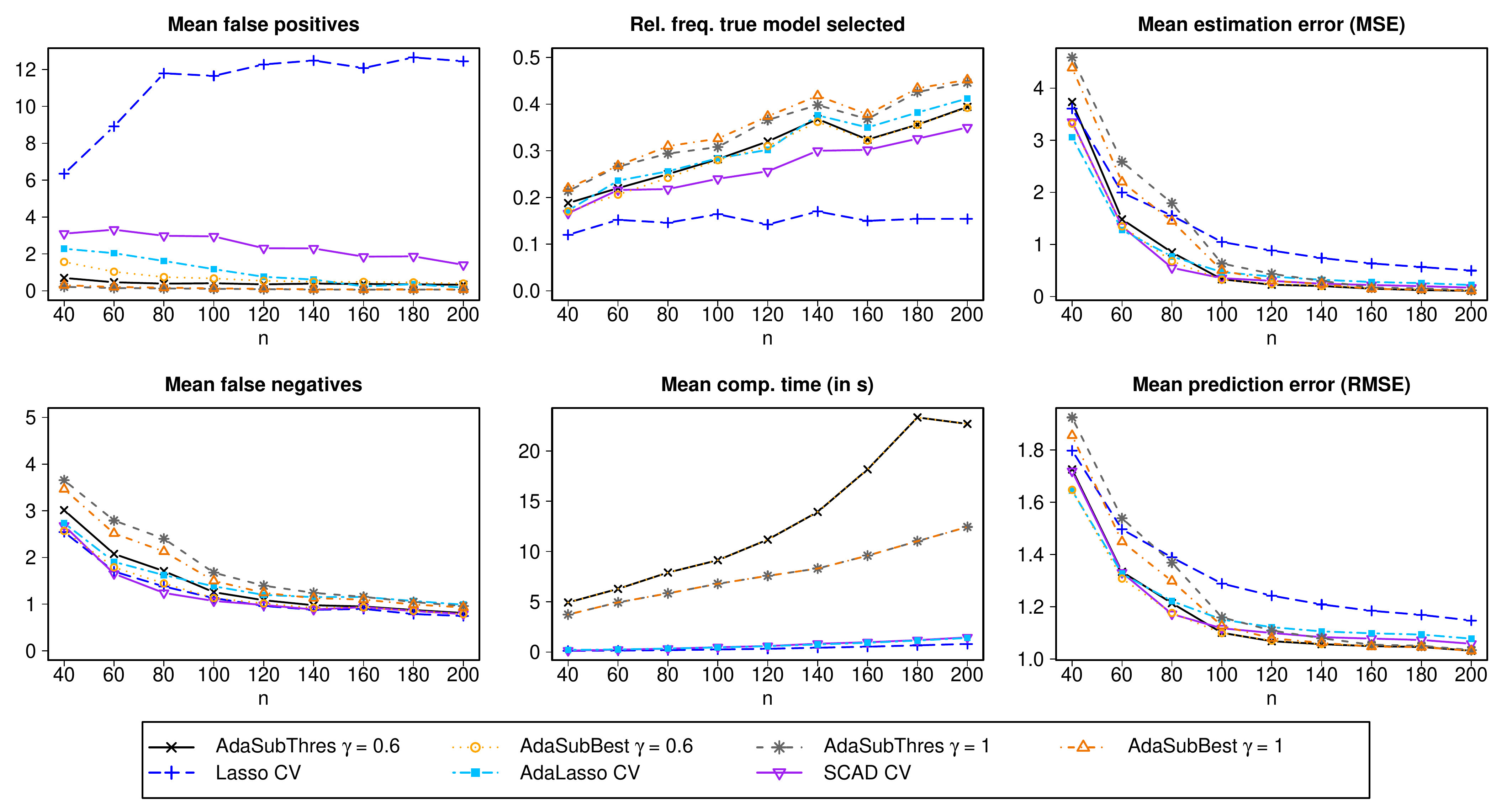}}
\caption{\label{fig:CV2}Results for high-dimensional setting (\(p=10n\)) with (a) equal-correlation structure and (b) block-correlation structure: Comparison of thresholded model (AdaSubThres) and ``best'' model (AdaSubBest) from AdaSub (for \(\text{EBIC}_\gamma\) with \(\gamma \in\{0.6,1\}\)) with Lasso, Adaptive Lasso (AdaLasso) and SCAD tuned with ten-fold cross-validation (CV) in terms of mean number of false positives/ false negatives, rel. freq. of selecting the true model, mean comp. time, MSE and RMSE.} 
\end{figure}

Finally, in the considered high-dimensional setting we additionally compare the performance of the AdaSub models (for \(\text{EBIC}_\gamma\) with \(\gamma \in\{0.6,1\}\)) with regularization methods tuned via cross-validation for ``optimal'' predictive performance, instead of applying the same criterion \(\text{EBIC}_\gamma\) for tuning parameter selection (see Section \ref{sec:simstudy}). In particular, we reconsider the Lasso, the Adaptive Lasso and SCAD where the penalty parameters are chosen with ten-fold cross-validation via the one-standard error rule, i.e.\ the final estimator is obtained by selecting the largest penalty parameter which is within one standard error of the minimal cross-validation error.

The results of the CV-tuned regularized estimators for the independence correlation structure (\(c=0\)) and the Toeplitz correlation structure (\(c=0.9\)) are depicted in Figure~\ref{fig:CV1}, while the results for the equal-correlation structure (\(c=0.7\)) and the block-correlation structure (\(b=10\) blocks and \(c=0.5\)) can be found in Figure~\ref{fig:CV2}. 
It is apparent that the Lasso tends to select many false positives, confirming the observation in \cite{feng2013} that tuning the penalty parameter of the Lasso via ten-fold cross-validation is not optimal when the aim is the identification of the true model. Even though it can be beneficial for variable selection to reserve a larger fraction %(than \(1/10\)) 
of the observed data for validation in the cross-validation procedure \citep{shao1993, feng2013}, here we have considered ten-fold cross-validation in combination with the one-standard error rule as a benchmark which is commonly used in practice (as the default in the R-package \texttt{glmnet}). 
The CV-tuned SCAD estimator tends to select less false positives than the Lasso in most of the cases, but still yields considerably more noise variables than the AdaSub models. In contrast, the Adaptive Lasso does not suffer from a very large number of false positives and generally performs quite well in the considered settings. 

Results in Figures~\ref{fig:CV1} and~\ref{fig:CV2} further show that the thresholded models \(\hat{S}_\rho\) (with \(\rho=0.9\)) from AdaSub yield the smallest mean numbers of false positives in all considered scenarios, while the ``best'' models \(\hat{S}_{\text{b}}\) identified by AdaSub can provide favorable predictive performance at the price of slightly increased numbers of false positives (especially for the choice \(\gamma = 0.6\) in \(\text{EBIC}_\gamma\)). 
Regarding the choice of \(\gamma\) for the selection criterion \(\text{EBIC}_\gamma\) in AdaSub, a smaller value (\(\gamma=0.6\)) is to be preferred if the main focus is predictive performance (minimizing mean prediction error on test data), while a larger value (\(\gamma=1\)) tends to provide sparser models and the largest relative frequencies for the correct identification of the true model.

\section*{Acknowledgements}
We would like to thank the reviewer for constructive and very helpful comments that improved an earlier version of the manuscript.   
\clearpage

\bibliographystyle{Chicago}
\footnotesize
\setlength{\bibsep}{0pt}
\bibliography{AdaSub_ref_short}

\end{document}